\documentclass[acmsmall]{acmart}

\usepackage{dsfont}
\usepackage{semantic,multirow,multicol}
\usepackage[algoruled,linesnumbered,noend]{algorithm2e}
\usepackage{color}
\usepackage{enumitem}
\usepackage{lscape}
\usepackage{colortbl}
\definecolor{mygray}{gray}{.85}

\usepackage{listings}
\newcommand{\tool}{{\sc HOME}\xspace}

\newcommand{\Bb}{\mathds{I}}

\newcommand{\Dd}{\mathcal{D}}
\newcommand{\Op}{{\bf {OP}}}
\newcommand{\Oo}{\mathcal{O}}
\newcommand{\type}{\mathcal{T}}

\newcommand{\return}{{\tt return}}
\newcommand{\rud}{{\tau_{{\tt uf}}}}
\newcommand{\sid}{{\tau_{{\tt si}}}}

\newcommand{\Expr}{\mathcal{E}}

\newcommand{\simply}{{\sf Simply}}

\newcommand{\sdd}{{\tau_{{\tt lk}}}}
\newcommand{\Dom}{{\sf Dom}}
\newcommand{\Col}{{\sf Col}}
\newcommand{\Alg}{{\sf Alg}}

\newcommand{\SI}{{\bf SI}}

\newcommand{\Var}{{\sf Var}}
\newcommand{\RVar}{{\sf RVar}}

\newcommand {\sem}[1]{{\llbracket{#1}\rrbracket}}

\AtBeginDocument{%
  \providecommand\BibTeX{{%
    \normalfont B\kern-0.5em{\scshape i\kern-0.25em b}\kern-0.8em\TeX}}}

\setcopyright{acmcopyright}
\copyrightyear{2020}
\acmYear{2020}
\acmDOI{10.1145/1122445.1122456}

\acmJournal{TOSEM}
\acmVolume{37}
\acmNumber{4}
\acmArticle{111}
\acmMonth{8}

\begin{document}

\title[Formal Verification of Higher-Order Masked Arithmetic Programs]{A Hybrid Approach to Formal Verification of Higher-Order Masked Arithmetic Programs}

\author{Pengfei Gao}
\email{gaopf@shanghaitech.edu.cn}
\affiliation{%
  \institution{ShanghaiTech University}
  \city{Shanghai}
  \country{China}
  \postcode{201210}
}
\additionalaffiliation{%
  \institution{University of Chinese Academy of Sciences}
  \city{Beijing}
  \country{China}
}
\additionalaffiliation{%
  \institution{Shanghai Institute of Microsystem and Information Technology}
  \city{Shanghai}
  \country{China}
}

\author{Hongyi Xie}
\email{xiehy@shanghaitech.edu.cn}
\affiliation{%
  \institution{ShanghaiTech University}
  \city{Shanghai}
  \country{China}
  \postcode{201210}
}

\author{Fu Song}
\affiliation{%
  \institution{ShanghaiTech University}
  \city{Shanghai}
  \country{China}
  \postcode{201210}
}
\email{songfu@shanghaitech.edu.cn}

\author{Taolue Chen}
\affiliation{%
  \institution{Birkbeck, University of London}
  \city{London}
  \country{UK}
}
\email{taolue@dcs.bbk.ac.uk}

\renewcommand{\shortauthors}{Gao et al.}


\begin{abstract}
Side-channel attacks, which are capable of breaking secrecy via side-channel information, pose a growing threat to the implementation of cryptographic algorithms.
Masking is an effective countermeasure against side-channel attacks by removing the statistical dependence between secrecy and power consumption via randomization.
However, designing efficient and effective masked implementations turns out to be an error-prone task.
Current techniques for verifying whether masked programs are secure are limited in their applicability and accuracy, especially when they are applied.
To bridge this gap, in this article, we first propose a sound type system, equipped with an efficient type inference algorithm,
for verifying masked arithmetic programs against higher-order attacks. We then give novel model-counting based and pattern-matching based methods
which are able to precisely determine whether the potential leaky observable sets detected by the type system are genuine or simply spurious.
We evaluate our approach on  various implementations of arithmetic cryptographic programs. The experiments confirm that our approach outperforms  
the state-of-the-art baselines in terms of applicability, accuracy and efficiency.

%
\end{abstract}



\maketitle

\section{Introduction}\label{sec:intro}
Side-channel attacks are capable of breaking secrecy via side-channel information such as power consumption~\cite{KJJ99,MOP07}, execution time~\cite{WGS18}, faults~\cite{KSBM19,Wagner04},
acoustic~\cite{GST17} and cache~\cite{GWW18}, posing a growing threat to implementations of cryptographic algorithms. In this work, we focus on power side-channel attacks, arguably the most effective physical side-channel attack~\cite{ISW03},
where power consumption data are used as the side-channel information.

To thwart power side-channel attacks, \emph{masking} is one of the most widely-used and effective countermeasures~\cite{ChariRR02,ISW03}.
Essentially, masking is designed to remove the statistical dependence between secrecy and power consumption via randomization.  
Fix a sound security parameter $d$, an order-$d$ masking typically makes use of a secret-sharing scheme to logically split the secret data into $(d+1)$ shares such that any $d'\leq d$ shares are \emph{statistically independent} on the secret data.
Masked implementations of some specific cryptographic algorithms such as PRESENT, AES and its non-linear component (Sbox)
(e.g.,~\cite{ISW03,SP06,RP10,MPLPW11,PR13}), as well as secure conversion algorithms between Boolean and arithmetic maskings (e.g.,~\cite{Goubin01,CGV14,Coron17,BCZ18,HT19}),  have been published over years.
It is crucial to realize that an implementation that is based on a secure scheme does  \emph{not} provide the secure guarantee in practice automatically.
For instance, the order-$d$ masking of AES proposed in~\cite{RP10} and its extensions~\cite{KHL11,CGPQR12} were later shown to be vulnerable to an attack of order-$(\lceil\frac{d}{2}\rceil+1)$~\cite{CPRR13}.
Indeed, designing efficient and effective masked implementations is an error-prone process.
Therefore, it is vital to verify masked programs in addition to the underlying security scheme, which should ideally be done automatically.

The predominant approach addressing this problem is the empirical leakage assessment by
statistical significance tests or launching state-of-the-art side-channel attacks, e.g.,~\cite{GJJR11,ANR17,Standaert17} to cite a few.
Although these approaches are able to identify some flaws, they can neither prove their absence nor identify all possible flaws exhaustively.
In other words, even if no flaw is detected,
it is still inconclusive, as it is entirely possible that the implementation could be broken with a better measurement setup or more leakage traces.
Recently, approaches based on formal verification are emerging for automatically verifying
masked programs~\cite{BYT17,BBDFGS15,BBDFGSZ16,BBFG18,BMZ17,Coron18,BGIKMW18,OMHE17,EWS14a,EWS14b,ZGSW18,GXZSC19}.
As the state of the art, most of these methods can only tackle Boolean programs~\cite{BBFG18,BMZ17,BGIKMW18,EWS14a,EWS14b,BYT17,ZGSW18,BelaidGR18} or first-order security~\cite{OMHE17,OuahmaMHE19,GXZSC19},
thus are limited in applicability and usability. Some work~\cite{BBDFGS15,BBDFGSZ16,Coron18} is able to verify arithmetic programs against higher-order attacks, but is limited in accuracy in the sense that secure programs may fail to pass the verification whereas
potential leaky observable sets are hard to be resolved automatically
so tedious manual examination is usually necessary to differentiate genuine and spurious ones.
Therefore, formal verification of masked arithmetic programs against higher-order attacks (with full tool support to
automatically resolve potential leaky observable sets) is still an unsolved question and requires further research.


\smallskip
\noindent{\bf Main contributions}.
Our work focuses on formal verification of \emph{higher-order masked arithmetic programs} based on
the standard probing model (ISW model) proposed by Ishai, Sahai and Wagner~\cite{ISW03}. Arithmetic programs admit considerably richer operations such as finite-field multiplication, and are much more challenging than their Boolean counterparts whose variables are over the Boolean domain only. Transforming arithmetic programs to equivalent Boolean ones and then applying existing tools is theoretically possible,
but suffers from several disadvantages: (1) complicated arithmetic operations (e.g., finite-field multiplication)
have to be encoded as bitwise operations; (2) verifying order-$d$ security of a 8-bit arithmetic program
must be done by verifying order-$(8^d)$ security over its Boolean translation which has considerably more observable variables (at least 8$\times$).
Because of this, we hypothesize that this approach is practically unfavourable, if not infeasible. Indeed, the state-of-the-art tool {\tt maskVerif}~\cite{BBFG18} has already required over 18 minutes to accomplish verification of the fifth-order masked Boolean implementation of DOM Keccak Sbox~\cite{GSM17}
which  has only 618 observable variables.

In light of this, we pursue a direct verification approach for higher-order masked arithmetic programs.
To guarantee that a masked program is order-$d$ secure, one has to ensure that the joint distributions of all size-$d$ sets of observable variables (observable sets)  that are potentially exposed to an attacker are independent of secret data.
There are two key challenges: (1) the combinatorial explosion problem of observable sets when the number of observable
variables and the security order increasing, and (2) how to efficiently and automatically resolve potential leaky observable
sets. The first challenge is addressed by the \textbf{first} step of our hybrid approach, for which we propose a sound type system together with an efficient type inference algorithm, which can prescribe a \emph{distribution type} for each observable set.
One can often---but not always---deduce leakage-freeness of observable sets from their distribution types, whereas observable sets that cannot be solved by the type system are regarded as potential leaky observable sets.

In case that \emph{potential} leaky observable sets are produced by the type system (i.e., the second challenge),
we provide automated \emph{resolution methods} which are the \textbf{second} step of our hybrid approach. This step is important:  for instance,  \cite{BBDFGS15} reported 98,176 potential third-order sets on Sbox~\cite{SP06}, which are virtually impossible to check individually by human beings.
Technically, the second step is based on \emph{model-counting} and \emph{pattern matching} based methods.
For the model-counting, we consider two baseline algorithms: the first one transforms the problem
to the satisfiability problem of a (quantifier-free) first-order logic formula that can be solved by SMT solvers (e.g. Z3~\cite{MB08}), an extension of our previous one for first-order security~\cite{GXZSC19}; the second one computes the probability distribution of an observable set
by naively enumerating all possible valuations of variables. We give, for the first time in the current paper, a third, \emph{GPU-accelerated parallel algorithm}, to leverage GPU's  computing capability. Instead of creating a general GPU-based solver which is control-flow intensive and would downgrade the GPU performance, we automatically synthesize a GPU program for each potential leaky observable set,
which, in a nutshell, enumerates all possible valuations of variables by leveraging GPU parallel computing.
It turns out that the GPU-based parallel algorithm significantly outperforms the two baseline algorithms.

The pattern matching based method is devised to further reduce the cost of model-counting.
This method infers the distribution type of an observable set from observable sets whose distribution types are known,
by searching an ``isomorphism'' between the computation expressions of the variables in two observable sets.
If such an isomorphism exists, one can conclude that the
two observable sets have the same distribution type, by which one can save costly model-counting
procedures. The pattern matching based method also automatically summarizes patterns of leaky observable sets
which can be used for diagnosis and debugging.

Our hybrid approach enjoys several advantages over the existing approaches. Compared to the empirical methods based on the statistical
analysis of leakage traces, our approach is able to give conclusive  security assertions 
independent of assessment conditions, testing strategies or the amount of gathered leakage
traces. Compared to the existing formal verification approaches, our overall hybrid approach is both sound and complete, and is able to verify more types of masked implementations. Remarkably, our model-counting and pattern matching based methods could also be integrated into existing formal verification approaches, effectively making them complete and more efficient.

We implement our approach in a tool \textbf{\tool} ({\bf H}igher-{\bf O}rder {\bf M}asking v{\bf{E}}rifier),
and evaluate on various benchmarks including masked implementations of full AES and MAC-Keccak programs.
The results are very encouraging:  \tool can handle benchmarks that have never been verified by existing formal verification approaches,
e.g., implementations of Boolean to arithmetic mask conversion from~\cite{SchneiderPOG19}, arithmetic to Boolean mask conversion from~\cite{CGTV15},
the non-linear transformation and round function of {\sc Simon} from~\cite{TangZZQ15}.
Our tool is also significantly faster than \cite{BBDFGS15} on almost all secure programs (e.g., 110$\times$ and 31$\times$ speed-up for Key schedule~\cite{RP10} and 4th-order Sbox~\cite{SP06}; cf.\  Table~\ref{tab:higherorder}), which is the only available tool to verify masked 
higher-order arithmetic programs under an equivalent leakage model to the ISW model.
The experimental results confirm that \tool is superior to existing tools in both functionality and performance. 
%

\smallskip
To sum up, the main contributions of this work are as follows:
\begin{itemize}
	\item We propose a sound type system and provide an efficient type inference algorithm for proving security of masked arithmetic programs;

	\item We propose a novel GPU-accelerated parallel algorithm to resolve potential leaky observable sets which significantly outperforms two baselines;

    \item We propose a novel pattern matching based method to automatically summarize patterns of leakage sets, which can reduce the cost of model-counting;

    \item We implement our algorithms in a software tool and demonstrate the effectiveness and efficiency of our approach on various benchmarks.
\end{itemize}

\medskip
The remainder of this article is organized as follows.
In Section~\ref{sec:pre}, we introduce basic notations and recall the probing leakage model.
In Section~\ref{sec:overview}, we present a motivating example and the overview of our approach.
In Section~\ref{sec:typesystem}, we present the sound type system, its inference algorithm and sound transformations to facilitate type inference.
In section~\ref{sec:modelcounting}, we describe the model-counting and patter matching based methods.
In Section~\ref{sec:exper}, we evaluate the performance of our
approach on representative examples from the literature.
We discuss related work in Section~\ref{sect:relatedwork}.
Finally, we conclude the article in Section~\ref{sec:concl}.

\section{Preliminaries}\label{sec:pre}

In this section, we describe masked cryptographic programs, masking schemes,
leakage models and security notions.

\subsection{Masked Cryptographic Programs}
We fix an integer $\kappa>0$ and the domain $\Bb=\{0,\cdots,2^\kappa-1\}$.
For a set $R$ of random variables, let $\Dd(R)$ denote the set of joint  distributions over $R$.
%

\medskip
\noindent{\bf Syntax.} We focus on programs written in C-like code that implement cryptographic algorithms such as AES, as opposed to arbitrary software programs. The syntax is given as follows.
\[\begin{array}{lrl}
   \mbox{Operation:} &  \Op \ni  \circ::= \oplus \mid \wedge \mid \vee \mid  \odot \mid +\mid -\mid \times \\
   \mbox{Expression:}&          e ::= c\in\Bb \mid x \mid e\circ e \mid \neg e \mid e\ll c\mid e\gg c \\
   \mbox{Statememt:} &          {\tt stmt} ::=  x\leftarrow e \mid {\tt stmt}; {\tt stmt}  \\
   \mbox{Program:}  &          P::=  {\tt stmt}; \ \return \ x_1,...,x_m;
  \end{array}\]
A program $P$ is a sequence of assignments $x\leftarrow e$ followed by a $\return$ statement,
where $e$ is an expression building from a set of variables and $\kappa$-bit constants using the
bitwise operations: \emph{negation} ($\neg$), \emph{and} ($\wedge$),
\emph{or} ($\vee$), \emph{exclusive-or} ($\oplus$), \emph{left shift} $\ll$ and \emph{right shift} $\gg$;
modulo $2^\kappa$ arithmetic operations: \emph{addition} ($+$), \emph{subtraction} ($-$), \emph{multiplication} ($\times$);
and finite-field \emph{multiplication} ($\odot$) over the finite field $\mathds{F}_{2^\kappa}$.

To analyze a cryptographic program $P$, it is common to assume that it is in straight-line form (i.e., branching- and loop-free)~\cite{EWS14b,BBDFGS15}. Remark that our tool supports programs with non-recursive procedure calls and static loops by transforming to straight-line form by procedure inlining and loop unfolding.
We further assume that $P$ is in the \emph{single static assignment} (SSA) form (i.e., there is at most one assignment $x\leftarrow e$ in $P$ for $x$.) and each expression uses at most one operator. (One can easily transform an arbitrary straight-line program to the SSA form.) 
For an assignment $x\leftarrow e$, we will denote by ${\tt Operands}(x)$ the set of operands associated
with the operator of $e$.

We fix a program $P$ annotated by the \emph{public}, \emph{private} (such as keys) and \emph{random} input variables.
The set $X$ of variables in $P$ is partitioned into four sets: $X_p,\ X_k,\ X_i$ and $X_r$,
where $X_p$ denotes the set of public input variables,  $X_k$ denotes the set of private input variables,
$X_i$ denotes the set of \emph{intermediate} variables,
and $X_r$ denotes the set of uniform random variables  on the domain $\Bb$.
In general, random variables are used for masking the private input variables.

\medskip
\noindent{\bf Semantics.}
For each variable $x\in X$, we define the \emph{computation} of $x$, $\Expr(x)$, as an expression over input variables $X_p\cup X_k\cup X_r$.
Formally, for each $x\in X$, $\Expr(x)=x$ if $x\in X_p\cup X_k\cup X_r$,
otherwise  $x$ is an intermediate variable (i.e., $x\in X_i$) which must be uniquely defined by an assignment statement $x\leftarrow e$ (thanks to SSA form of $P$), and thus $\Expr(x)$ is defined as the expression obtained from $e$ by sequentially replacing all the occurrences of the intermediate variables in $e$ by their defining expressions in $P$.

A \emph{valuation} is a function $\eta:X_p\cup X_k\rightarrow \Bb$ that assigns a concrete value to each input variable in $X_p\cup X_k$. Let $\Theta$ denote the set of valuations. Two valuations $\eta_1$ and $\eta_2$ are \emph{$X_p$-equivalent}, denoted by $\eta_1\approx_{X_p}\eta_2$,
if $\eta_1(x)=\eta_2(x)$  for  $x\in X_p$, i.e., $\eta_1$ and $\eta_2$ must agree on their values on public input variables.
We denote by $\Theta^2_{=X_p}\subseteq\Theta\times\Theta$ the set of pairs of $X_p$-equivalent valuations.
For each variable $x\in X$, let $\Expr_\eta(x)$ denote the expression obtained from $\Expr(x)$ by
instantiating variables $y\in X_p\cup X_k$ with concrete values $\eta(y)$.

Given a valuation $\eta\in \Theta$, for each variable $x\in X$, the computation $\Expr(x)$ of $x$
under the valuation $\eta$ can be interpreted as the probability distribution, denoted by $\sem{x}_{\eta}$, over
the domain $\Bb$ with respect to the uniform distribution of the random variables $\Expr_\eta(x)$ may contain.
Concretely, for each $c\in \Bb$, $\sem{x}_{\eta}(c)$ is the probability
that $\Expr_\eta(x)$ evaluates to $c$. Accordingly, let $\sem{P}_\eta$ denote the joint distribution over $\Bb^{|X|}$
such that for every $C=(c_x\in\Bb)_{x\in X}$, $\sem{P}_\eta(V)$
is the probability that $(\Expr_\eta(x))_{x\in X}$ evaluates to $(c_x\in\Bb)_{x\in X}$. 
Therefore, the semantics of $P$ is interpreted as a function
$\sem{P}:\Theta\rightarrow \Dd(X)$
which takes a valuation $\eta\in \Theta$ as the input and  returns a joint  distribution $\sem{P}_\eta$.
For each subset $Y\subseteq X$ of variables, let $\sem{P}^{Y}$ denote the function $\sem{P}^{Y}:\Theta\rightarrow \Dd(Y)$
that returns, for each $\eta\in \Theta$, the marginal distribution of $Y$ under $\sem{P}_\eta$.

\subsection{Masking}\label{sec:masking}
Masking is a randomization technique used to break the statistical dependence of the private input variables and
observable variables of the adversary~\cite{ChariRR02,ISW03}.
Fix a sound security parameter $d$, an \emph{order-$d$ masking} typically
makes use of a secret-sharing scheme to logically split the private data into $(d+1)$ shares such that any $d'\leq d$ shares are \emph{statistically independent} on the value of the private input.
The computation of shares for each private input is usually called \emph{presharing}.
A masking transformation aims at transforming an unmasked program $P$ that directly operates on the private inputs
into a masked program $P'$ that operates on their shares.
Finally, the desired data are recovered via de-masking of the outputted shares of the masked program $P'$.

For example, using the order-$d$ Boolean masking~\cite{ISW03}, the
$(d+1)$ shares of a key $k$ are $(r_1,\cdots, r_{d+1})$, where the shares $r_1,\dots,r_d$ are generated uniformly at random and $r_{d+1}$ is computed such that $r_{d+1}=k\oplus\bigoplus_{i=1}^d r_i$.
The value of $k$ can be recovered via performing exclusive-or ($\oplus$) operations on all
the shares, i.e., $\bigoplus_{i=1}^{d+1} r_i$.

Besides Boolean masking schemes, there are arithmetic masking
schemes such as additive (e.g., $(k + r) \bmod n$) and multiplicative masking schemes (e.g., $(k\times r) \bmod n$)
for protecting arithmetic operations.
Secure conversion algorithms between them (e.g.,~\cite{Goubin01,CGV14,Coron17,BCZ18,HT19}) have been proposed
for masking cryptographic algorithms that embrace both
Boolean and arithmetic operations (such as IDEA~\cite{LM90} and RC6~\cite{CRRY99}). 

When increasing the masking order $d$, the attack cost usually increases exponentially, but the performance of the masked programs degrades polynomially~\cite{GrossM18}. Therefore, the masking order is chosen by a trade-off between attack cost and performance.

\subsection{Leakage Model and Security Notions} \label{sec:threatmodel}
To formally verify the security of masked programs, it is necessary to define
the set of observable variables to the adversary and a leakage model that formally captures the
leaked information from the  set of observable variables.

\medskip
\noindent{\bf Observable variables.}
In the context of side-channel attacks, the adversary is assumed to be able to observe the public ($X_p$), random ($X_r$), and intermediate ($X_i$) variables via side-channel information, but is not able to observe the private input variables $X_k$ or the intermediate variables
of presharing. Indeed, presharing of each private input variable
is performed outside of the program
and is included only for the verification purpose~\cite{BBDFGS15,RP10}.
Therefore, for each program $P$, it is easy to automatically identify the set of observable variables $X_o\subseteq X_p\cup X_i\cup X_r$
which is assumed to be observed by the adversary.
Each subset $\Oo\subseteq X_o$ of observable variables is called an \emph{observable set}.

\medskip
\noindent{\bf  Leakage model.} In this article, we adopt the standard $d$-threshold probing model proposed by Ishai, Sahai and Wagner~\cite{ISW03}, usually referred to as \emph{ISW $d$-threshold probing model} (ISW model for short),
where the adversary may have access to the values of at most $d$ observable variables of his/her choice (e.g., via side-channel information). The more variables  an adversary observes, the higher the attack cost is.

\medskip
\noindent{\bf Uniform and statistical independence.}
Given a program $P$ and an observable set $\Oo\subseteq X_o$,
\begin{itemize}
  \item $P$ is \emph{uniform} w.r.t. $\Oo$, denoted by \emph{$\Oo$-uniform}, iff
for all valuations $\eta \in \Theta$: $\sem{P}^\Oo_\eta$ is a uniform joint distribution;
  \item $P$ is \emph{statistically independent} of $X_k$ with respect to $\Oo$, denoted by $\Oo$-$\SI$, iff
for every $(\eta_1,\eta_2)\in \Theta^2_{=X_p}$ (i.e., $\eta_1$ and $\eta_2$ agreeing on their values on public input variables):
$\sem{P}^\Oo_{\eta_1}=\sem{P}^\Oo_{\eta_2}.$
\end{itemize}
We say $P$ is $\Oo$-leaky if it is \emph{not} $\Oo$-$\SI$.

\medskip
According to the above definitions, it is straightforward to verify the following proposition.
\begin{proposition}\label{prop:uniform2SI}
Given an observable set $\Oo$ of a program $P$,
\begin{enumerate}
  \item if $P$ is $\Oo$-uniform, then $P$ is $\Oo$-$\SI$ and $\Oo'$-uniform for all $\Oo'\subseteq\Oo$;
  \item if $P$ is $\Oo$-$\SI$, then $P$ is $\Oo'$-$\SI$ for all $\Oo'\subseteq\Oo$.
\end{enumerate}
\end{proposition}

\begin{definition}[Security under the ISW $d$-threshold probing model~\cite{ISW03}]
A program $P$ is \emph{order-$d$ secure} if  $P$ is $\Oo$-$\SI$ or $\Oo$-uniform for every observable set $\Oo\subseteq X_o$ with $|\Oo| = d$.
\end{definition}
Intuitively, if $P$ is $\Oo$-$\SI$ or $\Oo$-uniform, then the distribution of the variables in $\Oo$ (hence power consumptions based on the variables in $\Oo$ in the ISW $d$-threshold probing model) does not rely on private data, and thus the adversary cannot deduce any information 
by observing variables in $\Oo$.

In literature, the ISW $d$-threshold probing model is also called order-$d$ perfect masking~\cite{EWS14b} or $d$-non-interference ($d$-NI)~\cite{BBDFGS15}.
There are other leakage models such as noise leakage model~\cite{PR13}, bounded moment model~\cite{BDFGSS17}, ISW model with transitions~\cite{CGPRRV12}
and with glitches~\cite{BGIKMW18,MPO05} and strong $d$-non-interference ($d$-SNI)~\cite{BBDFGSZ16,FGPPS18,BBFG18}.
It is known that all these models (except for $d$-SNI and $d$-NI  introduced in~\cite{BBDFGSZ16} and the extensions thereof)
can be reduced to the ISW $d$-threshold probing model~\cite{DDF14,BDFGSS17,BBDFGS15,BBFG18}
possibly at the cost of introducing higher orders when chosen plaintext attacks are adopted, namely, the adversary can use any plaintext during attack.
The $d$-SNI and $d$-NI models defined in~\cite{BBDFGSZ16} are  however stronger than the ISW $d$-threshold probing model. 
Namely, not all secure masked programs under the ISW $d$-threshold probing model are safe under $d$-SNI/$d$-NI, so cannot pass verification under this notion \cite{BDFGSS17}. In this work, we adopt the ISW $d$-threshold probing model, which is more common in side-channel analysis~\cite{EWS14b,ZGSW18,BGIKMW18,GXZSC19}.

Remark that the $d$-NI notion defined in~\cite{BBDFGSZ16,BBFG18} is strictly stronger than the one defined in~\cite{BBDFGS15}, although they bear the same name. Namely, in~\cite{BBDFGSZ16,BBFG18}, the $d$-NI notion requires that	the number of shares of each private input variable that can be accessed by the adversary is \emph{strictly} less than	$d+1$.

\medskip
\noindent{\bf Research objective.}
Our goal is to develop automated verification methods to determine
whether a given masked arithmetic program is order-$d$ secure under the
ISW $d$-threshold probing model.

\section{Motivating Example and  Overview of Approaches}\label{sec:overview}
In this section, we  present a motivating example and an overview of our approach.

\subsection{Motivating Example}
\label{sec:motivatingexample}

\begin{figure}[t]
\begin{lstlisting}[multicols=1]
BooleanToArithmetic$(k,r,r')${
  $x'\leftarrow k\oplus r$; //presharing
  $y_0 \leftarrow x' \oplus r'$;
  $y_1 \leftarrow y_0 - r'$;
  $y_2 \leftarrow y_1 \oplus x'$;
  $y_3 \leftarrow r' \oplus r$;
  $y_4 \leftarrow y_3 \oplus x'$;
  $y_5 \leftarrow y_4 - y_3$;
  $A \leftarrow y_5 \oplus y_2$;
  return $A$;
}
\end{lstlisting}
\caption{Goubin's Boolean to arithmetic mask conversion algorithm~\cite{Goubin01}.}
\label{fig:Runningexample}
\end{figure}

Figure~\ref{fig:Runningexample} presents an example which is an implementation of the Boolean to arithmetic mask conversion algorithm of Goubin~\cite{Goubin01}. The program assumes that the inputs are the private key $k$ and two random variables $r,r'$.
Line 2 is presharing which computes two shares $(x',r)$ of the private key $k$ via Boolean masking.
(Remark that Line 2 should be performed outside of the function \emph{BooleanToArithmetic} and is introduced for verification purpose only.
The actual implementation in Goubin~\cite{Goubin01} takes two shares $(x',r)$ as input and assigns $r'$ by a uniformly sampled random value.)
The function \emph{BooleanToArithmetic} returns two shares $(A,r)$ of the arithmetic masking of the private key $k$
such that $A+r=k$, but without directly recovers the key $k$ by $x'\oplus r$.

As setup for further use, we have: $X_p=\emptyset$, $X_k=\{k\}$, $X_r=\{r,r'\}$,
$X_i=\{x',A,y_1,\cdots,y_5\}$ and
$X_o=\{x',A,y_0,\cdots,y_5,r,r'\}$. The computations of variables in $X_i$ are:
\[\begin{array}{l}
\Expr(x')=k\oplus r; \\
\Expr(y_0)=(k\oplus r) \oplus r'; \\
\Expr(y_1)=((k\oplus r) \oplus r')-r'; \\
\Expr(y_2)=(((k\oplus r) \oplus r')-r') \oplus (k\oplus r);  \\
\Expr(y_3)=r' \oplus r;  \\
 \Expr(y_4)=(r'\oplus r)\oplus (k\oplus r);   \\
\Expr(y_5)=((r'\oplus r)\oplus (k\oplus r))- (r' \oplus r);  \\
\Expr(A)=\big(((r'\oplus r)\oplus (k\oplus r))- (r' \oplus r)\big) \oplus
    \big((((k\oplus r) \oplus r')-r') \oplus (k\oplus r)\big).
\end{array}\]

For each observable variable $z\in X_o$, the program is $\{z\}$-uniform (note that $\Expr(A)$ is equivalent to $k-r$), and thus
it is first-order secure. However, this program is \emph{not} second-order secure, e.g.,
$y_0\oplus y_3\equiv x'\oplus r\equiv k$ allowing to extract private key  $k$ by observing $\{y_0,y_3\}$.

\subsection{Overview of Approach}
The overview of our approach \tool\ is depicted in Figure~\ref{fig:overview},
consisting of four main components: pre-processor, type system, pattern matching based method,
and model-counting based method. Given a masked program $P$ and the security order $d$, \tool checks whether
the masked program $P$ is order-$d$ secure or not. If $P$ is not order-$d$ secure, then
\tool outputs the leaks, i.e., all the size-$d$ observable sets $\Oo$ such that $P$ is $\Oo$-leaky.

\begin{figure}[t]
  \centering
  \includegraphics[width=1\textwidth]{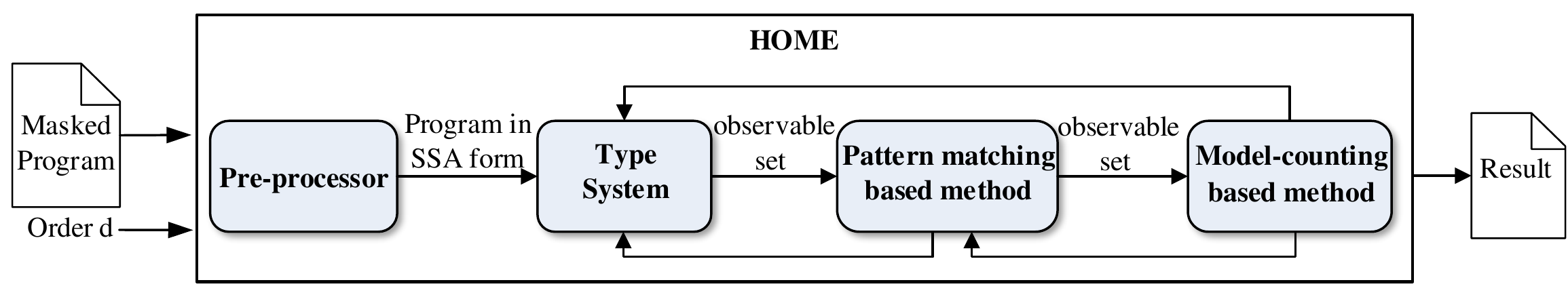}\\
  \caption{Overview of our approach \tool.}\label{fig:overview}
\end{figure}

Given a masked program $P$ and the order $d$, the pre-processor unfolds the static loops (i.e., loops with a predetermined bound of iterations) and inlines the procedure calls,
and then transforms the program into the SSA form.
The type system is used to check whether each size-$d$ observable set $\Oo$ is order-$d$ secure
by deriving valid type judgements. If we can deduce that the observable set $\Oo$ is either order-$d$ secure  or certainly not according to the distribution type, 
then the result is conclusive.

However, as usual, the type system is incomplete, namely,
it is possible that the distribution type cannot be inferred in which case  
we first apply the pattern matching based method. This method iteratively searches an ``isomorphism" between the computation expressions of the
variables in $\Oo$ and the variables in $\Oo'$, where $\Oo'$ is another size-$d$ observable set whose distribution type is already known.
If such an isomorphism exists, we can conclude that these two observable sets $\Oo$ and $\Oo'$ have the
same distribution type, effectively resolving the observable set $\Oo$. 
%
The result of the observable set $\Oo$ will be  fed back to the type system which can be used to gradually improve the accuracy of the type inference.

When the pattern matching based method fails to resolve the observable set $\Oo$,
we will apply the (normally expensive) model-counting based method, which is able to completely decide whether
the observable set $\Oo$ is order-$d$ secure. Finally,
the observable set $\Oo$ is cached for further invocation of pattern matching. As before, the result
of the observable set $\Oo$ will be fed back to the type system.

This procedure gives a sound and complete approach for verification of higher-order security. In the next two sections, we will elucidate the details of the type system, type inference algorithm, model-counting and pattern matching based methods.

\section{Type System} \label{sec:typesystem}
In this section, we first present a type system to infer the \emph{distribution type} of an observable set,
then propose three sound transformations to facilitate type inference, and finally
present the type inference algorithm based on the type system and the  sound transformations.
%
%
 \subsection{Dominant Variables}
We first introduce the notion of dominant variables. 
\begin{definition}
	\label{def:dom}
A random variable $r$ is called a \emph{dominant variable} of an expression $e$
if the following two conditions hold:
\begin{enumerate}
  \item  $r$ (syntactically) occurs in the expression $e$ exactly once; and
 \item for each operator $\circ$ on the path between the leaf $r$ and the root in the abstract syntax tree of the expression  $e$,
\begin{itemize}
    \item if $\circ=\odot$, then one of its children is a non-zero constant; or
    \item$\circ\in\{\oplus,\neg,+,-\}$; or
    \item $\circ$ is  a (univariate) bijective function.
\end{itemize}
\end{enumerate}
\end{definition}

We denote by $\Var(e)$ the set of variables appearing in an expression $e$,
by $\RVar(e)$ the set $\Var(e)\cap X_r$,
and by $\Dom(e)$ the set of all dominant (random) variables of $e$.
All of these sets can be computed in polynomial time in the size of $e$.
Furthermore, note that a particularly useful example of bijective functions is ${\tt Sbox}$ which is ubiquitous in cryptographic programs.

If $r_x$ is a dominant variable of the expression $\Expr(x)$ such that $r_x\not\in \bigcup_{x'\in\Oo.x'\neq x}\RVar(\Expr(x'))$,
then $\Expr(x)$ can be seen as a fresh random variable when evaluating $\sem{P}^{\Oo}$.
Therefore, if each expression $\Expr(x)$ for $x\in\Oo$ has such dominant variables,
we can deduce that $P$ is $\Oo$-uniform.

\begin{proposition}\label{prop:dom}
Given an observable set $\Oo\subseteq X_o$, if for every $x\in\Oo$
there exists a dominant variable $r_x\in \Dom(\Expr(x))$ such that $r_x\not\in \bigcup_{x'\in\Oo.x'\neq x}\RVar(\Expr(x'))$,
then $P$ is $\Oo$-uniform.
\end{proposition}

\begin{proof} To prove this proposition, we first introduce the notion of $i$-invertibility.
An expression (i.e., function) $e(x_1,\cdots,x_n)$ is \emph{$i$-invertible} if, for any concrete
values $c_1,\cdots,c_{i-1},c_{i+1},\cdots, c_n\in \Bb$,
the expression
\begin{center}
$e(c_1/x_1,\cdots,c_{i-1}/x_{i-1},x_i,c_{i+1}/x_{i+1},\cdots, c_n/x_{n})$
\end{center}
obtained by instantiating all the variables $(x_j)_{j\neq i}$
with concrete values $(c_j)_{j\neq i}$ is bijective.
It is easy to see that
$e(c_1/x_1,\cdots,c_{i-1}/x_{i-1},x_i,c_{i+1}/x_{i+1},\cdots, c_n/x_{n})$ and $x_i$ have same distribution.
Thus, if $x_i$ is a random variable, then
$e(c_1/x_1,\cdots,c_{i-1}/x_{i-1},x_i,c_{i+1}/x_{i+1},\cdots, c_n/x_{n})$ must be a uniform distribution.

The following claim reveals the relation between dominated expressions and $i$-invertibility.

\medskip
\fbox{\parbox[c]{0.95\textwidth}{
{\bf Claim.} 
Given an expression $e(x_1,\cdots,x_n)$ over variables $\{x_1,\cdots,x_n\}$,
for every $1\leq i\leq n$, if $x_i$ is a dominant variable of $e(x_1,\cdots,x_n)$,
then $e(x_1,\cdots,x_n)$ is $i$-invertible.
}}

\medskip
We prove that $e(x_1,\cdots,x_n)$ is $i$-invertible by induction on the length $\ell$
of the path between the leaf $r$ and the root in the abstract syntax tree of $e$.

\medskip
\noindent
{\bf Base case $\ell=0$}.
The expression $e(x_1,\cdots,x_n)$ must be $x_r$ which is a bijective function. The result immediately follows.

\medskip
\noindent
{\bf Inductive step $\ell>0$}.
Let $\circ$ be the operator at the root of the syntax tree of $e(x_1,\cdots,x_n)$,
then $e(x_1,\cdots,x_n)$ is in the form of
\begin{enumerate}
  \item $\neg e_1(x_1,\cdots,x_n)$, or
\item $\circ(e_1(x_1,\cdots,x_n))$ where $\circ$ is a (univariate) bijective function, or
  \item $e_1(x_1,\cdots,x_n)\circ e_2(x_1,\cdots,x_{i-1},x_{i+1},\cdots,x_n)$
such that $x_i$ is a dominant variable of $e_1(x_1,\cdots,x_n)$, where $\circ \in\{\odot,\oplus,+,-\}$. (Note that $x_i$ does not appear in $e_2(x_1,\cdots,x_{i-1},x_{i+1},\cdots,x_n)$.)
\end{enumerate}
By the induction hypothesis, $e_1(x_1,\cdots,x_n)$ is $i$-invertible. By the definition of  $i$-invertibility,
for any concrete values $c_1,\cdots,c_{i-1},c_{i+1},\cdots, c_n\in \Bb$, $e_1(c_1/x_1,\cdots,c_{i-1}/x_{i-1},x_i,c_{i+1}/x_{i+1},\cdots,c_n/x_n)$
is bijective. Then, the result immediately follows if $e(x_1,\cdots,x_n)$ is $\neg e_1(x_1,\cdots,x_n)$ or $\circ(e_1(x_1,\cdots,x_n))$, i.e., Item (1) and Item (2).
It remains to consider Item (3).

\begin{itemize}
  \item If $\circ=\odot$, then $e_2(x_1,\cdots,x_{i-1},x_{i+1},\cdots,x_n)$ is a non-zero constant.
   Since the multiplicative group of the non-zero elements in $\Bb$ is cyclic and $0\odot e_2(x_1,\cdots,x_{i-1},x_{i+1},\cdots,x_n)=0,$
then $e_1(c_1/x_1,\cdots,c_{i-1}/x_{i-1},x_i,c_{i+1}/x_{i+1},\cdots,c_n/x_n)\odot e_2(c_1/x_1,\cdots,c_{i-1}/x_{i-1},c_{i+1}/x_{i+1},\cdots,c_n/x_n)$
is also bijective. Hence, the result follows.

  \item If $\circ\in\{\oplus,+,-\}$, then  $e_2(c_1/x_1,\cdots,c_{i-1}/x_{i-1},c_{i+1}/x_{i+1},\cdots,c_n/x_n)$ is a constant.
For any constant $c\in\Bb$,
$e_1(c_1/x_1,\cdots,c_{i-1}/x_{i-1},x_i,c_{i+1}/x_{i+1},\cdots,c_n/x_n)\circ c$
is still bijective (note that $+$ and $-$ are operators over the ring $\Bb$).
Hence, the result follows.
\end{itemize}

\medskip
Now, we prove the proposition.

Suppose for every $x\in\Oo$,
there exists $r_x\in \Dom(\Expr(x))$ such that $r_x\not\in \bigcup_{x'\in\Oo.x'\neq x}\RVar(\Expr(x'))$,
 let $\sem{P[r_x/x]_{x\in\Oo}}^{\Oo}_\eta$ denote the distribution of $\sem{P}^\Oo_\eta$ in which
$\Expr(x)$ is replaced by $r_x$ for all $x\in\Oo$, then  for all valuations $\eta \in \Theta$, $\sem{P}^\Oo_\eta=\sem{P[r_x/x]_{x\in\Oo}}^{\Oo}_\eta$
holds.

By applying the above claim,
we get that $\sem{P[r_x/x]_{x\in\Oo}}^{\Oo}_\eta$ is a uniform distribution.
Therefore, the result immediately follows.
\end{proof}

\begin{example}\label{eample:dom}
Let us consider the motivating example in Section~\ref{sec:motivatingexample}.
$\Expr(x')$ is dominated by the random variable $r$.
$\Expr(y_0)$ and $\Expr(y_3)$ both have two dominant variables $r$ and $r'$.
$\Expr(y_1)$ only has the dominant variable $r$, as $r'$ occurs twice. Similarly, $\Expr(y_4)$
only has the dominant variable $r'$, as $r$ occurs twice. Thus, for every observable set $\Oo \subseteq \{x',y_0,t_1,y_3,y_4\}$ with $|\Oo| =1$,
we can deduce that the program is
$\Oo$-uniform.
$\Expr(y_2)$, $\Expr(y_5)$ and $\Expr(A)$ have no dominant variables, as both $r$ and $r'$ occur more than once.

For the observable set $\{x',y_3\}$, although the dominant variable $r'$ of $\Expr(y_3)$ does not appear in $\Expr(x')$, the dominant variable $r$ of $\Expr(x')$ appears in $\Expr(y_3)$, thus we cannot deduce that the program is
$\{x',y_3\}$-uniform. Indeed, for any observable set $\Oo\subseteq \{x',A,y_1,\cdots,y_5,r,r'\}$ with $|\Oo|\geq 2$,
we cannot deduce that the program is
$\Oo$-uniform.
\end{example}



\subsection{Types and Type Inference Rules}
In this subsection, we introduce distribution types and their inference rules for proving higher-order security.

\begin{definition}
Let $\type$ be the set of (distribution)
types $\{\rud,\sid,\sdd\}$,
\begin{itemize}
\item $\rud$ stands for uniform distribution, i.e., $\Oo:\rud$ means that the program is $\Oo$-uniform;
\item $\sid$ stands for secret independent distribution, i.e.,  $\Oo:\sid$ means that the program is $\Oo$-$\SI$;
\item $\sdd$ stands for leak, i.e.,  $\Oo:\sdd$ means that the program is $\Oo$-leaky, namely, not $\Oo$-$\SI$.
\end{itemize}
where $\Oo$ is an observable set.
\end{definition}

The distribution type $\rud$ is a subtype of $\sid$, i.e., $\rud$ implies $\sid$,
but $\sid$ does not imply $\rud$.
Although, both $\sid$ and $\rud$ can be used to prove that the program is statistically independent of the
secret for an observable set $\Oo$, i.e., no leak, $\rud$ is more desired because the observable set
$\Oo$ not only is statistically independent of the secret (same as in $\sid$),
but also can be used like a set of random variables.
Therefore, we prefer $\rud$ over $\sid$ and want to infer as many $\rud$ as possible.

Type judgements  are in the form of $\vdash \Oo:\tau$
where $\Oo$ is an observable set, $\tau\in\type$ is
the type of $\Oo$. Note that we omitted the context of the type judgement for simplifying presentation.
The type judgement $\vdash \Oo:\tau$ is valid iff
the distribution of the values of variables from $\Oo$ satisfies the property specified by $\tau$ in the program $P$.

\begin{figure*}
	\centering

	\scalebox{.85}{	\begin{tabular}{rrr}
		$\inference{\vdash x_1\star x_2:\tau & x\leftarrow x_2\star x_1} {\vdash x:\tau }[({\sc Com})]$
		&  $\inference{\vdash x':\tau & x\leftarrow \neg x'} {\vdash x:\tau}[({\sc Ide}$_1$)]$ \\ \\
		$\inference{\vdash x':\sid & x\leftarrow  x' \bullet x' }{\vdash x:\sid}[({\sc Ide}$_2$)]$
		& $\inference{x\leftarrow x' \diamond x'}{\vdash x:\sid}[({\sc Ide}$_3$)]$ \\ \\
		$\inference{\vdash x':\sdd & x\leftarrow x' \bowtie x'}{\vdash x:\sdd}[({\sc Ide}$_4$)]$
		&  $\inference{\vdash x_1:\rud & \vdash x_2:\rud & x\leftarrow x_1\circ x_2 \\ \Dom(\Expr(x_1))\setminus\RVar(\Expr(x_2))\neq \emptyset} {\vdash x:\sid}[({\sc Sid}$_{4}$)]$  \\ \\ 
		$\inference{\vdash x_1:\sid &  \vdash x_2:\sid & x\leftarrow x_1\bullet x_2  \\ \RVar(\Expr(e_1))\cap\RVar(\Expr(e_2))= \emptyset} {\vdash x:\sid}[({\sc Sid}$_{5}$)]$
		&  $\inference{ \vdash x_1:\sdd & \vdash x_2:\rud & x\leftarrow x_1\circ x_2\\ \Dom(\Expr(x_2))\setminus\RVar(\Expr(x_1)) \neq \emptyset}  {\vdash x:\sdd }[({\sc Sdd})]$
	\end{tabular}}
	\caption{Type inference rules for first-order security, where  $\star\in\Op, \ \circ\in\{\wedge,\vee,\odot,\times\}$, $\bullet\in\Op^\ast$, $\bowtie\in\{\wedge,\vee\}$ and $\diamond\in \{\oplus,-\}$.}
	\label{tab:firstordersemrules}
\end{figure*}

Figure~\ref{tab:firstordersemrules} presents type inference rules for the \emph{first-order} security.
We denote by $\Op^\ast$ the set $\Op\cup\{\ll,\gg\}$.
Rule ({\sc Com}) captures the commutative law of operators $\star\in\Op$. 
Rules ({\sc Ide}$_i$) for $i=1,2,3,4$ are straightforward.
Rule ({\sc Sid}$_{4}$) states that
$x$ has type $\sid$ if $x\leftarrow x_1\circ x_2$ for $\circ\in\{\wedge,\vee,\odot,\times\}$, both $x_1$ and $x_2$ have type $\rud$, and
$\Expr(x_1)$ has a dominant variable $r$ which is not used by $\Expr(x_2)$.
Indeed, $\Expr(x)$ can be seen as $r\circ \Expr(x_2)$.
Rule ({\sc Sid}$_{5}$) states that expression
$x$ has type $\sid$ if $x\leftarrow x_1\bullet x_2$ for $\bullet\in\Op^\ast$,
both $x_1$ and $x_2$ have type $\sid$ (as well as its subtype $\rud$), and
the sets of random variables used by $\Expr(x_1)$ and $\Expr(x_2)$ are disjoint.
Indeed, for each valuation $\eta\in\Theta$, the distributions $\sem{x_1}_{\eta}$
and $\sem{x_2}_{\eta}$ are independent.
Rule ({\sc Sdd}) states that the variable $x$ has type $\sdd$ if $x\leftarrow x_1\circ x_2$ for $\circ\in\{\wedge,\vee,\odot,\times\}$,
$x_1$ has type $\sdd$, $x_2$ has type $\rud$, and
$\Expr(x_2)$ has a dominant variable $r$ which is not used by $\Expr(x_1)$.
Intuitively, $\Expr(x)$ can be safely seen as $\Expr(x_1)\circ r$.

Note that the type inference rules for the first-order security 
are similar to those from~\cite{GXZSC19}, which are reproduced here for completeness.
%
%
%
The new rules for the higher-order security are given in Figure~\ref{tab:semrules}.
We briefly explain these rules below.

%
%
%
Rule ({\sc No-Key}) states that if $\Oo$ is an observable set whose values are independent of private variables,
then $\Oo$ has type $\sid$.
Rule ({\sc Sid}$_1$) states that if $\Oo_1$ has type $\sid$ and the computations $\Expr(x)$ of variables $x\in\Oo_2$ only involve
public variables, then we can deduce that $\Oo_1\cup\Oo_2$ has type $\sid$.
Rule ({\sc Sid}$_2$) states that if $\Oo$ has type $\sid$ and a variable $x$ is defined using
constants, public variables or variables in $\Oo$, then adding $x$ into $\Oo$ does not change the type.
Intuitively, as the value of $x$ is determined by its operands, for every $(\eta_1,\eta_2)\in \Theta^2_{=X_p}$,
$\sem{P}^\Oo_{\eta_1}=\sem{P}^\Oo_{\eta_2}$ if and only if $\sem{P}^{\Oo\cup\{x\}}_{\eta_1}=\sem{P}^{\Oo\cup\{x\}}_{\eta_2}$.
Rule ({\sc Sid}$_3$) deals with a $\sid$-typed observable set
$\Oo_1$ and a $\rud$-typed observable set $\Oo_2$ (cf. Proposition~\ref{prop:dom}). 
Assume that each computation $\Expr(x)$ for $x\in\Oo_2$ has a dominant variable $r_x$ which
is not used in any computation of variable in $\Oo_1\cup \Oo_2$ except $x$, then $\Oo_1\cup\Oo_2$
has type $\sid$. Intuitively, each computation $\Expr(x)$ for $x\in\Oo_2$ can be seen as the random variable $r_x$,
and $\Oo_1$ has type $\sid$, hence, the distributions $\sem{x}_\eta$ for all $x$ in $\Oo_1\cup\Oo_2$ and all valuations
$\eta\in\Theta$ are independent.
Similarly, rule ({\sc Rud}) deals with two $\rud$-typed observable sets
$\Oo_1$ and $\Oo_2$. Although the dominant variables of $\Expr(x)$
for $x\in \Oo_1$ may appear in the computations of variable in $\Oo_2$, by Proposition~\ref{prop:dom},
all the variables $x\in O_2$ can be seen as fresh random variables $r_x$ so that Proposition~\ref{prop:dom}
can be applied.

\begin{figure}[t]
	\centering
	\scalebox{.85}{
		\begin{tabular}{rr}
			$\inference{(\bigcup_{x\in\Oo}\Var({\Expr(x)}))\cap X_k=\emptyset} {\vdash \Oo:\sid}[({\sc No-Key})]$  &
			$\inference{\vdash \Oo_1:\sid &  \bigcup_{x\in\Oo_2}\Var({\Expr(x)})\subseteq X_p}{\vdash \Oo_1\cup\Oo_2:\sid}[({\sc Sid}$_1$)]$ \\\\
		& 	$\inference{\vdash \Oo:\sid & {\tt Operands}(x)\subseteq \Oo\cup X_p\cup \Bb  }{\vdash \Oo\cup\{x\}:\sid}[({\sc Sid}$_2$)]$ \\\\
			\multicolumn{2}{r}{$\inference{\vdash \Oo_1:\sid & \forall x\in \Oo_2. \exists r_x\in \Dom(\Expr(x))\setminus \bigcup_{y\in \Oo_1\cup\Oo_2\wedge y\neq x } \RVar(\Expr(y))}{\vdash \Oo_1\cup\Oo_2:\sid}[({\sc Sid}$_3$)]$} \\\\
			\multicolumn{2}{r}{$\inference{\vdash \Oo_1:\rud & \forall x\in \Oo_2. \exists r_x\in \Dom(\Expr(x))\setminus \bigcup_{y\in \Oo_1\cup\Oo_2\wedge y\neq x } \RVar(\Expr(y))}{\vdash \Oo_1\cup\Oo_2:\rud}[({\sc Rud})]$}
	\end{tabular}}
	\caption{Type inference rules for \textbf{higher-order} security.}
	\label{tab:semrules}
\end{figure}
%
%

\begin{theorem}[Soundness of the type system] \label{thm:type}
For every set observable $\Oo\subseteq X_o$,
\begin{enumerate}
  \item if $\vdash \Oo:\sid$ is valid, then $P$ is $\Oo$-$\SI$;
  \item if $\vdash \Oo:\rud$ is valid, then $P$ is $\Oo$-uniform;
  \item if $\vdash \Oo:\sdd$ is valid, then $P$ is $\Oo$-leaky.
\end{enumerate}
\end{theorem}

\begin{proof} We only show the soundness of rules for the higher-order security. First, Item (3) directly follows from 
	the first-order case~\cite{GXZSC19}. We now deal with the rules in Figure~\ref{tab:semrules}.

\begin{itemize}
%
 \item
{\bf Rule}  ({\sc No-Key}). Suppose $(\bigcup_{x\in\Oo}\Var({\Expr(x)}))\cap X_k=\emptyset$, then
the expression $\Expr(x)$ does not use any private variable for all $x\in\Oo$. This implies that $\sem{P}^\Oo_{\eta_1}=\sem{P}^\Oo_{\eta_2}$ for every $(\eta_1,\eta_2)\in \Theta^2_{=X_p}$ (note that $\eta_1$ and $\eta_2$ must agree on their values on public input variables).
 \item
{\bf Rule}  ({\sc Sid}$_1$).  Suppose $\vdash \Oo_1:\sid$, then
$\sem{P}^{\Oo_1}_{\eta_1}=\sem{P}^{\Oo_1}_{\eta_2}$ for every $(\eta_1,\eta_2)\in \Theta^2_{=X_p}$.
Consider $\Oo_2$ such that $\bigcup_{x\in\Oo_2}\Var({\Expr(x)})\subseteq X_p$, then for every $x\in\Oo_2$, $(\eta_1,\eta_2)\in \Theta^2_{=X_p}$ and assignment of random variables
$f:X_r\rightarrow \Bb$, the expression $\Expr(x)$ evaluates to same value under $(\eta_1,f)$ and $(\eta_2,f)$.
This implies that  $\sem{P}^{\Oo_1\cup O_2}_{\eta_1}=\sem{P}^{\Oo_1\cup O_2}_{\eta_2}$ for every $(\eta_1,\eta_2)\in \Theta^2_{=X_p}$.
 \item
{\bf Rule}  ({\sc Sid}$_2$).  Suppose $\vdash \Oo:\sid$, then
$\sem{P}^{\Oo}_{\eta_1}=\sem{P}^{\Oo}_{\eta_2}$ for every $(\eta_1,\eta_2)\in \Theta^2_{=X_p}$.
Suppose the observable set $\Oo$ is $\{x_1,\cdots,x_n\}$, then for each vector of concrete values $(c_1,\cdots,c_n)\in\Bb^n$,
$\sem{P}^{\Oo}_{\eta_1}(c_1,\cdots,c_n)=\sem{P}^{\Oo}_{\eta_2}(c_1,\cdots,c_n)$.
Consider $x_{n+1}$ such that ${\tt Operands}(x_{n+1})\subseteq \Oo\cup X_p\cup \Bb$,
let  $c_{n+1}$ denote the value of $x_{n+1}$ under the valuation $\eta_1$ and $x_1=c_1,\cdots,x_n=c_n$,
and $c_{n+1}'$ denote the value of $x_{n+1}$ under the valuation $\eta_2$ and $x_1=c_1,\cdots,x_n=c_n$.
Since $\eta_1$ and $\eta_2$ must agree on their values on public input variables, then $c_{n+1}=c_{n+1}'$.
Therefore, for every concrete value $c$,
$\sem{P}^{\Oo\cup\{x_{n+1}\}}_{\eta_1}(c_1,\cdots,c_n,c)=\sem{P}^{\Oo\cup\{x_{n+1}\}}_{\eta_2}(c_1,\cdots,c_n,c)=0$ if $c\neq c_{n+1}$,
$\sem{P}^{\Oo\cup\{x_{n+1}\}}_{\eta_1}(c_1,\cdots,c_n,c)=\sem{P}^{\Oo}_{\eta_1}(c_1,\cdots,c_n)$ and
$\sem{P}^{\Oo\cup\{x_{n+1}\}}_{\eta_2}(c_1,\cdots,c_n,c)=\sem{P}^{\Oo}_{\eta_2}(c_1,\cdots,c_n)$ if $c=c_{n+1}$.
Hence, the result immediately follows.
 \item
{\bf Rule}  ({\sc Sid}$_3$).  Suppose $\vdash \Oo_1:\sid$, then
$\sem{P}^{\Oo_1}_{\eta_1}=\sem{P}^{\Oo_1}_{\eta_2}$ for every $(\eta_1,\eta_2)\in \Theta^2_{=X_p}$.
Consider $\Oo_2$ such that $\forall x\in \Oo_2. \exists r_x\in \Dom(\Expr(x))\setminus \bigcup_{y\in \Oo_1\cup\Oo_2\wedge x\neq y} \RVar(\Expr(y))$,
i.e., for each $x\in\Oo_2$, there exists a dominant random variable $r_x\in  \Dom(\Expr(x))$ which is not used in other expressions in
$\Expr(y)$ for  $y\in \Oo_1\cup\Oo_2$ with $x\neq y$. Thus,
$\sem{P}^{\Oo_1\cup O_2}_{\eta_1}=\sem{P}^{\Oo_1\cup O_2}_{\eta_2}$ for every $(\eta_1,\eta_2)\in \Theta^2_{=X_p}$.
 \item
{\bf Rule}  ({\sc Rud}). Consider $\Oo_2$ such that $\forall x\in \Oo_2. \exists r_x\in \Dom(\Expr(x))\setminus \bigcup_{y\in \Oo_1\cup\Oo_2\wedge x\neq y} \RVar(\Expr(y))$,
i.e., for each $x\in\Oo_2$, there exists a dominant random variable $r_x\in  \Dom(\Expr(x))$ which is not used in other expressions in
$\Expr(y)$ for  $y\in \Oo_1\cup\Oo_2$ with $x\neq y$. Let $\Oo_3$ denote the set of such random variables $r_x$.
Then, $\sem{P}^{\Oo_1\cup O_2}_{\eta_1}=\sem{P}^{\Oo_1\cup \Oo_3}_{\eta_2}$, for every $(\eta_1,\eta_2)\in \Theta^2_{=X_p}$.

Since the program is $\Oo_1$-uniform, then we get that the program is $\Oo_1\cup \Oo_3$-uniform.
 The result follows from $\sem{P}^{\Oo_1\cup O_2}_{\eta_1}=\sem{P}^{\Oo_1\cup \Oo_3}_{\eta_2}$.
%
%
\end{itemize}
The proof is completed.
\end{proof}

Remark that our type inference rules are designed to be redundant for efficiency consideration.
Namely, they have distinct complexities to check the premises. For instance, rule ({\sc Rud}) is special case of rule ({\sc Sid}$_3$), as $\rud$ is a subtype of $\sid$, but we prefer $\rud$ over $\sid$.
Also, rule ({\sc Rud}) is a reformation and generalization of Proposition~\ref{prop:dom} which allows to add more observable variables to the observable set $\Oo_1$ without searching dominant variables in the computations of variables in $\Oo_1$.
A valid judgement derived by rule ({\sc Sid}$_2$) in  constant-time can also be derived by using other rules, but
rule ({\sc Sid}$_2$) could avoid unfolding the definitions of variables.
When applying these rules, we start with those which can infer the type $\rud$ and whose premises can be established at a lower cost, namely,
in the order of rules ({\sc Rud}), ({\sc Sid}$_2$), ({\sc Sid}$_1$), ({\sc No-Key}) and  ({\sc Sid}$_3$).

\begin{example}
Let us consider the motivating example in Section~\ref{sec:motivatingexample}.
Recalling that $\Expr(y_3)=r' \oplus r$, although
we can derive both $\vdash \{y_3\}:\rud$ by applying rule ({\sc Rud})
and $\vdash \{y_3\}:\sid$ by applying rule ({\sc No-Key}). We will prefer $\vdash \{y_3\}:\rud$.

In Example~\ref{eample:dom}, we claim that for any observable set $\Oo\subseteq \{x',A,y_1,\cdots,y_5,r,r'\}$ with $|\Oo|\geq 2$,
we cannot deduce that the program is
$\Oo$-uniform by applying Proposition~\ref{prop:dom}. As an example,
let us consider the observable set $\{x',y_3\}$.
As $\Expr(x')$ is dominated by the random variable $r$, we derive
that $\vdash \{x'\}:\rud$. As  $\Expr(y_3)$ has the dominant variable $r'$
which does not appear in $\Expr(x')$, we can derive that $\vdash \{x',y_3\}:\rud$
by applying rule ({\sc Rud}). However, $r$ appears in $\Expr(y_3)$, so
we cannot directly apply Proposition~\ref{prop:dom} to prove that the program is $\{x',y_3\}$-uniform.
Similarly, we can deduce that $\vdash \{x',y_0\}:\rud$ and $\vdash \{x',y_4\}:\rud$, but we still cannot deduce the distribution types of the other size-2 observable sets.
For instance, we cannot infer the distribution type of the observable set $\{x',y_1\}$,
as $\Dom(\Expr(x'))=\Dom(\Expr(y_1))=\{r\}$.
\end{example}

\subsection{Sound Transformations}\label{sec:trans}
In this subsection, we describe three sound, domain-specific transformations for facilitating type inference.

The first transformation is based on the observation that 
some computations may share common sub-expressions which are dominated by some random variables, and  these random variables are only used in these sub-expressions.
Such sub-expressions, treated as random variables (i.e., replaced by the dominant variables) when analyzing the computations, are uniform and independent.
This  may enable type inference rules, as the other random variables
in sub-expressions will be eliminated. Therefore, we can simplify computations by leveraging the notion of dominant variables.

For instance, consider the observable set $\{x',y_1\}$ in the motivating example.
Recall that $\Expr(x')=k\oplus r$ and $\Expr(y_1)=((k\oplus r) \oplus r')-r'$.
We can observe that the sub-expression $k\oplus r$ is dominated by the random variable $r$ which occurs exclusively in  $k\oplus r$.  
Therefore, $\Expr(x')$ and $\Expr(y_1)$ can be simplified as $r$ and $(r\oplus r')-r'$ respectively,
as the distributions of $v\oplus r$ and $r$ are identical for any value $v\in \Bb$ of $k$, and $r$ does not affect the values of
other sub-expressions. Using the simplified computations $r$ and $(r\oplus r')-r'$ of $\Expr(x')$ and $\Expr(y_1)$, we can deduce
$\vdash \{x',y_1\}:\sid$ by applying rule ({\sc No-Key}).
This simple, but crucial, observation is formalized as the following definition.

\begin{definition} \label{def:domr}
	A sub-expression $e$ in a set of computations $E$ is \emph{dominated by a random variable $r$}
	if $r\in\Dom(e)$ and $r$ only occurs in $e$, namely, does not occur in $E$ elsewhere.
\end{definition}	

To facilitate type inference, we may replace the largest $r$-dominated sub-expression $e$ by $r$, which can be done in polynomial-time by traversing the abstract syntax tree.
Let $\simply_{\Dom}(E)$ be the set of computations obtained from
$E$ by repeatedly applying this strategy.

$\simply_{\Dom}$ is generally very effective in our experiments, but fails on one
benchmark. This is because $\simply_{\Dom}$ only relies upon syntactic information of the computation.
For instance, consider the observable set $\{x_1,x_2\}$ taking from the second-order masked implementation of the AES Sbox~\cite{SP06}, where
\begin{itemize}
  \item $\Expr(x_1)=\Big({\tt Sbox}\big((0\oplus((k\oplus r_0)\oplus r_1))\oplus r_1\big) \oplus  r_2  \Big) \oplus r_3$,
  \item $\Expr(x_2)=\Big({\tt Sbox}\big((r_0 \oplus ((k\oplus r_0) \oplus r_1)) \oplus r_1\big) \oplus r_2  \Big) \oplus r_3$,
\end{itemize}
$k$ is a private input variable, and $r_0,r_1,r_2,r_3$ are random variables.
$\simply_{\Dom}$ is not able to simplify the sub-expression $r_2\oplus r_3$ into a random variable, though
both $r_2$ and $r_3$ are dominant variables of $r_2\oplus r_3$.


$\simply_{\Dom}$ could be applied if we could transform $\Expr(x_1)$ and $\Expr(x_2)$ to equivalent forms (by the associativity of $\oplus$), i.e.,
\begin{center}
${\tt Sbox}\big((0\oplus((k\oplus r_0)\oplus r_1))\oplus r_1\big) \oplus  (r_2   \oplus r_3)$
and ${\tt Sbox}\big((r_0 \oplus ((k\oplus r_0) \oplus r_1)) \oplus r_1\big)\oplus (r_2   \oplus r_3)$.
\end{center}
However, carrying out such a transformation automatically is very challenge in general, as there is no canonical
representation of the computation to which $\simply_{\Dom}$ can be applied.
To address this challenge, we propose the sound transformation
which aims to collapse several variables into one variable, e.g.,
collapse $r_2$ and $r_3$ into a new random variable even if they do not appear
as the sub-expression  $r_2 \oplus r_3$.
This idea is formalized as the following definition.

\begin{definition} \label{def:coll}
Given a set of computations $E$ and
a set of variables $Z\subseteq \bigcup_{e\in E}\Var(e)$,
$Z$ is \emph{collapsible} with respect to $E$
if the following two conditions hold:
\begin{enumerate}
  \item $Z\subseteq X_p$ or $Z\subseteq X_k$ or $Z\subseteq X_r$, namely, variables in $Z$ have the same type;
  \item and there exist sub-expressions $e_1,\cdots, e_k$ in $E$ such that:
  \begin{itemize}
    \item sub-expression $e_j$ for each $1\leq j\leq k$ can be rewritten as $(\bigoplus_{z\in Z}z)\oplus e_j'$, i.e., clustering the variables in $Z$ together,
    \item and each variable $z\in Z$ only occurs in $\{e_1,\cdots, e_k\}$, and occurs in $e_j$ for each $1\leq j\leq k$ exactly once.
  \end{itemize}
\end{enumerate}
\end{definition}

One can observe that if $Z$ is collapsible, then $\bigoplus_{z\in Z}z$  can be replaced by 
a fresh variable respecting the type (i.e. public, key, or random) when analyzing $\{\Expr(x)\mid x\in\Oo\}$ for the observable set $\Oo$.
For simplicity, we usually use $\overline{Z}$ to denote the fresh variable.
We denote by $\simply_{\Col}(E)$ the set of computations computed from
$E$ by repeatedly applying this strategy.
$\simply_{\Col}(E)$ is implemented in polynomial-time by iteratively searching
pairs of variables $\{x_1,x_2\}$ that are collapsible 
and replacing them by $\overline{\{x_1,x_2\}}$.

The third transformation is the application of algebra laws.
We denote by $\simply_{\Alg}(E)$ the set of computations computed from
$E$ by repeatedly applying algebra laws such as $e\oplus e\equiv 0$, $0\oplus e\equiv e$, $0\times e\equiv 0$, $0\odot e\equiv0$ and $e-e\equiv0$.
For $0\oplus e\equiv e$, $0\times e\equiv 0$, $0\odot e\equiv0$, we directly search for the constant $0$.
For $e\oplus e\equiv 0$ and $e-e\equiv0$, the representation of computations in $E$ shares the same common sub-expressions
so that we do not need to compare whether two sub-expressions are same or not when applying $\simply_{\Alg}(E)$.
Moreover, instead of considering only sub-expressions of the form $e\oplus e$ (resp. $e-e$),
we search for two occurrences of the sub-expression $e$ such that the operators on the path between the roots of two occurrences of $e$ are all $\oplus$ (resp. $-$).

\medskip

It is straightforward to verify the following proposition.
\begin{proposition}
Given a program $P$ and an observable set $\Oo$,
let $\overline{P}$ denote the program
\[(x\leftarrow \overline{\Expr(x)};)_{x\in\Oo}\ \return;\]
where
$\overline{\Expr(x)}$ is obtained from $\Expr(x)$
by applying $\simply_{\Dom}(E)$, $\simply_{\Col}(E)$ and/or $\simply_{\Alg}(E)$,
then	$\sem{\overline{P}}^{\Oo}$ and $\sem{P}^{\Oo}$ generate the same distribution over $\Oo$.
\end{proposition}

\begin{example}Let us consider the above example, i.e.,  the observable set $\{x_1,x_2\}$, where
\begin{itemize}
  \item $\Expr(x_1)=\Big({\tt Sbox}\big((0\oplus((k\oplus r_0)\oplus r_1))\oplus r_1\big) \oplus  r_2  \Big) \oplus r_3$,
  \item $\Expr(x_2)=\Big({\tt Sbox}\big((r_0 \oplus ((k\oplus r_0) \oplus r_1)) \oplus r_1\big) \oplus r_2  \Big) \oplus r_3$,
\end{itemize}
$k$ is a private input variable, and $r_0,r_1,r_2,r_3$ are random variables.
The type system in Figure~\ref{tab:semrules} fails to prove $\vdash\{x_1,x_2\}:\rud$.

One can observe that $Z=\{r_2,r_3\}$ is collapsible with respect to $\{\Expr(x_1),\Expr(x_2)\}$, so by replacing $Z=\{r_2,r_3\}$ with a new random variable $\overline{Z}$,
$\{\Expr(x_1),\Expr(x_2)\}$ can be simplified to
\[E_1=\{ {\tt Sbox}\big((0\oplus((k\oplus r_0)\oplus r_1))\oplus r_1\big) \oplus \overline{Z}, ~
{\tt Sbox}\big((r_0 \oplus ((k\oplus r_0) \oplus r_1)) \oplus r_1\big)\oplus \overline{Z}\}.\]

By iteratively applying $\simply_{\Alg}$ to $E_1$ using algebraic laws $r_0\oplus r_0\equiv 0$
and $r_1\oplus r_1\equiv 0$, we obtain
\[E_2=\{ {\tt Sbox}\big(0\oplus((k\oplus r_0)\oplus 0)\big) \oplus \overline{Z}, ~
{\tt Sbox}\big((k\oplus 0) \oplus 0\big)\oplus \overline{Z}\}.\]

Since $0\oplus e \equiv 0\oplus e \equiv e$, by iteratively applying $\simply_{\Alg}$ to $E_2$, we obtain
\[E_3=\{ {\tt Sbox}(k\oplus r_0) \oplus \overline{Z}, ~
{\tt Sbox}(k)\oplus \overline{Z}\}.\]
Since $r_0$ is the dominant variable of ${\tt Sbox}(k\oplus r_0)\oplus \overline{Z}$ but does not occur in ${\tt Sbox}(k)\oplus \overline{Z}$,
by applying $\simply_{\Dom}$, we obtain $E_4=\{r_0, ~
{\tt Sbox}(k)\oplus \overline{Z}\}.$ Now, $\overline{Z}$ becomes the dominant variable of ${\tt Sbox}(k)\oplus \overline{Z}$
but does not occur in $r_0$, by applying $\simply_{\Dom}$ again, we obtain that
$E_5=\{r_0, ~ \overline{Z}\}$, from which we can deduce 
$\vdash\{x_1,x_2\}:\rud$.
\end{example}

\begin{algorithm}[t]
\SetAlgoNoLine
\small
\SetKwProg{myfunc}{Function}{}{}
${\tt PLS}:=\emptyset$; $\lambda:=$empty\_map; $\pi:=$empty\_map\;
\myfunc{{\sf HOME}$(P,X_p,X_k,X_r,X_o,d)$}{
    $X_{\tt check}:=\{x\in X_o\mid \Var(\Expr(x))\not\subseteq X_p\}$\;
    \ForAll{$x\in X_{\tt check}$}{
        \If{$\simply_{\Alg}(\Expr(x))\neq \Expr(x)$}{
             $\lambda(x):=\simply_{\Alg}(\Expr(x))$\;
            $\pi(x):=\Dom(\lambda(x))$\;
        }
        \lElse{$\pi(x):=\Dom(\Expr(x))$}
    }
	{\sf Explore}$(\{(d,X_{\tt check})\})$\;
    \Return ${\tt PLS}$\;
}
\medskip
\myfunc{{\sf Explore}$(\mathcal{Y})$}{
    \ForAll{$(i,\Oo)\in \mathcal{Y}$}{
        Choose a subset $\mathcal{C}_{i,\Oo}\subseteq \Oo$ in a topological order from leaf to root  s.t. ${|\mathcal{C}_{i,\Oo}|}=i$\;
    }
    \If{{\sf Check}$( \{\mathcal{C}_{i,\Oo}\}_{(i,\Oo)\in \mathcal{Y}} )=\top$}
    {
        \ForAll{$(i,\Oo)\in \mathcal{Y}, \ x\in\Oo\setminus  \mathcal{C}_{i,\Oo}$ in a topological order from leaf to root}{
            \If{{\sf Check}$(\{\mathcal{C}_{i,\Oo}\}_{(i,\Oo)\in \mathcal{Y}},\{x\} )=\top$}{
                $\mathcal{C}_{i,\Oo}:=\mathcal{C}_{i,\Oo}\cup\{x\}$\;
            }
        }
    }
    \lElse{${\tt PLS}:={\tt PLS}\cup \{\bigcup_{(i,\Oo)\in \mathcal{Y}} \mathcal{C}_{i,\Oo}\}$}
    $\mathcal{Y}':=\{(i,\Oo)\in \mathcal{Y} \mid |\Oo|>i\wedge i\neq0\}$\;
    \lIf{$\mathcal{Y}'=\emptyset$}{\Return}
    \ForAll{$(i,\Oo)\in \mathcal{Y}',  0\leq i_j\leq \min(i,|\Oo\setminus \mathcal{C}_{i,\Oo}|)$ s.t. $\sum_{(i,\Oo)\in \mathcal{Y}'} i_j\neq 0$}{
        {\sf Explore}$((\mathcal{Y}\setminus\mathcal{Y}')\cup\bigcup_{(i,\Oo)\in \mathcal{Y}'}\{(i-i_j,\mathcal{C}_{i,\Oo}),(i_j,\Oo\setminus \mathcal{C}_{i,\Oo})\})$\;
    }
    \Return\;
}
\medskip
\myfunc{{\sf Check}$(\{\mathcal{C}_{i,\Oo}\}_{(i,\Oo)\in \mathcal{Y}}, Y=\emptyset )$}{
    \If{$\vdash  Y\cup\bigcup_{(i,\Oo)\in \mathcal{Y}} \mathcal{C}_{i,\Oo}:\tau$ for some $\tau\in\{\rud,\sid\}$ is valid}{
        \Return{$\top$}\;
    }\ElseIf{$\vdash_{\simply_{\Dom}} Y\cup\bigcup_{(i,\Oo)\in \mathcal{Y}} \mathcal{C}_{i,\Oo}:\tau$ for some $\tau\in\{\rud,\sid\}$ is valid}{
        \Return{$\top$}\;
    }\ElseIf{$\vdash_{\simply_{\Col}} Y\cup\bigcup_{(i,\Oo)\in \mathcal{Y}} \mathcal{C}_{i,\Oo}:\tau$ for some $\tau\in\{\rud,\sid\}$ is valid}{
        \Return{$\top$}\;
    }
    \Return{$\bot$\;}

}
\caption{Type inference algorithm.}
\label{alg:DivHOMaRer}
\end{algorithm}

\subsection{Type Inference Algorithm}
In this subsection, we present our type inference algorithm.

To prove that $P$ is order-$d$ secure,
it is necessary to ensure that, for all size-$d$ observable subsets $\Oo\subseteq X_o$,
$P$ is $\Oo$-$\SI$. Evidently, exhaustive enumeration of $\binom{|X_o|}{d}$ subsets 
may not scale.
To address this issue,
the key idea is Proposition~\ref{prop:uniform2SI} which states that
if the program $P$ is $\Oo$-$\SI$ (resp. $\Oo$-uniform), then $P$ is also $\Oo'$-$\SI$ (resp. $\Oo'$-uniform)
for any subset $\Oo'\subseteq \Oo$. Therefore,
the main strategy is to find observable sets $\{\Oo_i\}_{i=1}^n$ as large as possible
such that $P$ is $\Oo_i$-$\SI$ for all $1\leq i\leq n$, and
for each size-$d$ subset $\Oo\subseteq X_o$,
$\Oo\subseteq\Oo_i$ for some $1\leq i\leq n$.

Our idea is formalized in
Algorithm~\ref{alg:DivHOMaRer}, where $\vdash \Oo:\tau$ denotes the type inference without applying the  transformations
$\simply_{\Dom}$ or $\simply_{\Col}$; $\vdash_{\simply_{\Dom}}\Oo:\tau$ denotes the type inference aided with
the transformation $\simply_{\Dom}$; $\vdash_{\simply_{\Col}} \Oo:\tau$ denotes the type inference aided by
both transformations $\simply_{\Dom}$ and $\simply_{\Col}$.
Taking a program $P$, sets of public ($X_p$), private ($X_k$),  random
($X_r$) and observable $(X_o)$ variables, and the order $d$ as inputs,
the algorithm first  initializes three data structures:
${\tt PLS}$ for storing all \emph{potential} leaky observable sets,
$\lambda$ for storing the simplified computation of each variable, and
$\pi$ for storing the set of dominant variables of the (simplified) computation $\Expr(x)$ for each variable $x$.

At Line 3, Algorithm~\ref{alg:DivHOMaRer} computes the set $X_{\tt check}$ of observable variables whose computation involves either private or random variables. This allows  to isolate the set of observable variables 
whose computation
involves public input variables only.
Hence, according to rule ({\sc Sid}$_1$), it suffices to consider
 size-$d$ subsets $\Oo\subseteq X_o\setminus \{x\in X_o\mid \Var(\Expr(x))\subseteq X_p\}$.
At Lines 4-8, it simplifies the computation $\Expr(x)$ for each variable $x\in X_{\tt check}$ by invoking $\simply_{\Alg}$
and computes its dominant variables; the results are stored in $\lambda$ and $\pi$
for later use. 
After that, it invokes the function {\sf Explore} with the set
$\{(d,X_{\tt check})\}$ (Line 9). We assume that $|X_{\tt check}|\geq d$, otherwise we can directly check whether
$\vdash X_{\tt check}:\sid$ is valid or not.

The function {\sf Explore} is more involved.
It aims at proving that for all pairs $(i,\Oo)\in \mathcal{Y}$
and all possible subsets $\Oo_i\subseteq\Oo$ with size $i$,
the type judgement $\vdash \bigcup_{(i,\Oo)\in \mathcal{Y}} \Oo_i:\tau_{\sid}$ is valid.
Taking a set $\mathcal{Y}$ of pairs $(i,\Oo)$ as input which 
satisfies the following three properties:
\begin{itemize}
  \item[(1)] $\sum_{\{i\mid (i,\Oo)\in \mathcal{Y}\}} i=d$, namely, the sum of orders' $i$ for subsets $\Oo$ in $\mathcal{Y}$ is the target order $d$;
  \item[(2)] $\biguplus\{\Oo\mid (i,\Oo)\in \mathcal{Y}\}=X_{\tt check}$, namely, the subsets $\Oo$ in $\mathcal{Y}$ form a partition of $X_{\tt check}$; and
  \item[(3)]  $|\Oo|\geq i$ for all $(i,\Oo)\in \mathcal{Y}$, namely, there are at least $i$ variables in $\Oo$ for each $(i,\Oo)\in \mathcal{Y}$.
\end{itemize}
Remark that these properties are maintained and required to show the  correctness and termination
of our algorithm.

An illustration of the  function {\sf Explore}  is given in Figure~\ref{fig:alg}.
The function {\sf Explore} first chooses a size-$i$ subset $\mathcal{C}_{i,\Oo}\subseteq\Oo$ for each pair $(i,\Oo)\in \mathcal{Y}$  in a topological order from leaf to root (Line 13).
Then it checks whether the type judgement $\vdash \bigcup_{(i,\Oo)\in \mathcal{Y}}\mathcal{C}_{i,\Oo}:\tau$ for some $\tau\in\{\rud,\sid\}$ is valid or not
by invoking the function {\sf Check} (Line 14).

\begin{itemize}
  \item If it is valid, i.e., the observable set $\bigcup_{(i,\Oo)\in \mathcal{Y}}\mathcal{C}_{i,\Oo}$ has distribution type $\rud$ or $\sid$ (as shown in the middle-part of Figure~\ref{fig:alg}), then {\sf Explore} iteratively tries to add the remaining observable variables $x$ to $\mathcal{C}_{i,\Oo}$ for $x \in \Oo\setminus \mathcal{C}_{i,\Oo}$ and $(i,\Oo)\in \mathcal{Y}$ by invoking the function {\sf Check} (Lines 15-17).
      The effect of this addition is shown in the right-part of Figure~\ref{fig:alg}.

  \item Otherwise $\vdash \bigcup_{(i,\Oo)\in \mathcal{Y}}\mathcal{C}_{i,\Oo}:\tau$ for any $\tau\in\{\rud,\sid\}$ is invalid, then $\bigcup_{(i,\Oo)\in \mathcal{Y}}\Oo_i$ is a
potentially leaky set and  is added to the set ${\tt PLS}$ (Line 18).
\end{itemize}
Finally, to cover $\bigcup_{(i,\Oo)\in \mathcal{Y}} \Oo_i$ for all possible size-$i$ subsets $\Oo_i\subseteq\Oo$ and pairs $(i,\Oo)\in \mathcal{Y}$,
it remains to check the observable sets
$\bigcup_{(i,\Oo)\in \mathcal{Y}} \Oo_i$,
where there exists at least one
pair $(i,\Oo)\in \mathcal{Y}$ such that $\Oo_i$ contains at least one variable from $\Oo\setminus \mathcal{C}_{i,\Oo}$.
(Otherwise $\bigcup_{(i,\Oo)\in \mathcal{Y}} \Oo_i\subseteq \bigcup_{(i,\Oo)\in \mathcal{Y}}\mathcal{C}_{i,\Oo}$.)
To do this, we first extract the pairs $(i,\Oo)$ such that $|\Oo|>i$ and $i\neq 0$, i.e.,
$\mathcal{Y}':=\{(i,\Oo)\in \mathcal{Y} \mid |\Oo|>i\wedge i\neq 0\}$ at Line 19.
If $\mathcal{Y}'$ is empty, then all the possible subsets $\bigcup_{(i,\Oo)\in \mathcal{Y}} \Oo_i$ are covered
and Algorithm~\ref{alg:DivHOMaRer} terminates (Line 20).
Otherwise, we partition all the pairs $(i,\Oo)\in \mathcal{Y}'$ into pairs $(i-i_j,\mathcal{C}_{i,\Oo}),(i_j,\Oo\setminus \mathcal{C}_{i,\Oo})$
for all combinations of values $0\leq i_j\leq \min(i,|\Oo\setminus\mathcal{C}_{i,\Oo}|)$ such that
$\sum_{(i,\Oo)\in \mathcal{Y}'} i_j\neq 0$. The condition $\sum_{(i,\Oo)\in \mathcal{Y}'} i_j\neq 0$ is used to avoid the case $\bigcup_{(i,\Oo)\in \mathcal{Y}} \Oo_i\subseteq \bigcup_{(i,\Oo)\in \mathcal{Y}}\mathcal{C}_{i,\Oo}$.
For each such combination of values, 
the partitioned pairs $\{(i-i_j,\mathcal{C}_{i,\Oo}),(i_j,\Oo\setminus\mathcal{C}_{i,\Oo}) \mid (i,\Oo)\in \mathcal{Y}'\}$ together with the pairs $\{(i,\Oo)\in \mathcal{Y} \mid {|\Oo|}=i\vee i=0\}$ (i.e., $\mathcal{Y}\setminus \mathcal{Y}'$)
are checked by recursively calling the function {\sf Explore}.
It is easy to observe that the recursion maintains the above three properties.

The function {\sf Check} first verifies whether $\vdash Y\cup\bigcup_{(i,\Oo)\in \mathcal{Y}} \Oo_i:\tau$ for some $\tau\in\{\rud,\sid\}$ is valid, which may be aided with data structures $\lambda$ and $\pi$  (Line 25).
If it is valid,   $\top$ is returned (Line 26).
Otherwise, it is verified
with the additional transformation $\simply_{\Dom}$  (Line 27).
If it still fails, $\vdash Y\cup\bigcup_{(i,\Oo)\in \mathcal{Y}} \Oo_i\tau$ for some $\tau\in\{\rud,\sid\}$ is checked
using the additional transformation $\simply_{\Col}$ on the expressions yielded by $\simply_{\Dom}$ (Line 29).
Once $\vdash Y\cup\bigcup_{(i,\Oo)\in \mathcal{Y}} \Oo_i:\tau$ for some $\tau\in\{\rud,\sid\}$ is derived,
{\sf Check} returns $\top$ (Lines 28 and 30).
If all of these steps fail, $\bot$ is returned (Line 31).
Notices that during the above type inference, once $\vdash Y\cup\bigcup_{(i,\Oo)\in \mathcal{Y}} \Oo_i:\sdd$ becomes valid,
 {\sf Check} also returns $\bot$.
Moreover, in order to avoid recomputing $\simply_{\Dom}$ and $\simply_{\Col}$,
the sequence of applied transformations are recorded, and the simplified expressions are cached.
When the function {\sf Check} is invoked at Line 16,
i.e., $Y$ is nonempty, we first check whether
the recorded sequence of applied transformations is still legal.
If it is still applicable, we will reuse the simplified expressions 
and apply $\simply_{\Dom}$ and/or $\simply_{\Col}$ to $\Expr(x)$ as well.
Otherwise, the function {\sf Check} immediately returns $\bot$.

Remark that 
the type inference rules are applied in the order of increasing complexities of checking the premises while preferring
$\rud$ over $\sid$.
We also remark that  the choice of the subsets at Line 13 and the variable $x$ at Line 15 may have significant impact on the performance.
We choose variables  from leaf to root following the order of the size of the defining computation, in light of
Rule ({\sc Sid}$_2$) in Figure~\ref{tab:semrules}.

\begin{figure}[t]
  \centering
  \includegraphics[width=.85\textwidth]{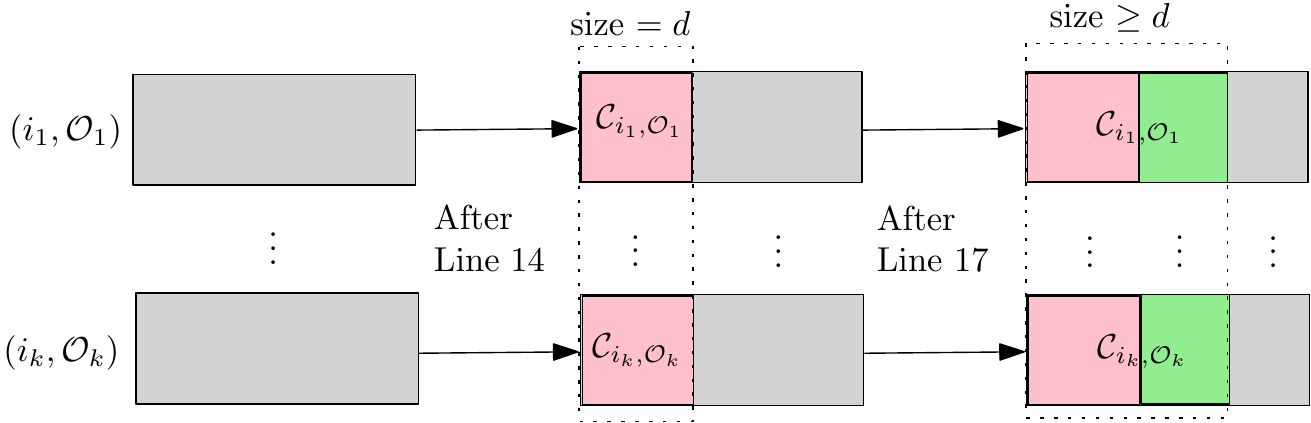}\\
  \caption{Intuition of Algorithm~\ref{alg:DivHOMaRer}.}\label{fig:alg}
\end{figure}

%


The procedure terminates as we only partition pairs $(i,\Oo)\in \mathcal{Y}$ such that $|\Oo|>i$ and $i\neq 0$
and the sizes of $\mathcal{C}_{i,\Oo}$ and $\Oo\setminus\mathcal{C}_{i,\Oo}$ in partitioned pairs $(i-i_j,\mathcal{C}_{i,\Oo})$ and $(i_j,\Oo\setminus\mathcal{C}_{i,\Oo})$
eventually become smaller and smaller in recursive calls until ${|\Oo|}=i$ or $i=0$.
(Note that we keep pairs of the form $(0,\Oo)$ in the worklist for simplifying presentation.
They are indeed removed in our implementation.)

\begin{theorem}\label{thm:HOMV}
$P$ is order-$d$ secure if ${\tt PLS}=\emptyset$.
Moreover, if $P$ is $\Oo$-leaky for $\Oo\subseteq X_{\tt check}$ with ${|\Oo|}=d$,
then $\Oo\in {\tt PLS}$.

\end{theorem}

Note that the reverse of Theorem~\ref{thm:HOMV} may not hold.
To prove Theorem~\ref{thm:HOMV}, we start with the following lemmas.
First, we show that the above three properties always hold.

\begin{lemma}\label{lem:3property}
In Algorithm~\ref{alg:DivHOMaRer}, at each call {\sf Explore}$(\mathcal{Y})$, the following three properties are
hold:
\begin{itemize}
  \item[(1)] $\sum_{\{(i,\Oo)\in \mathcal{Y}\}} i=d$;
  \item[(2)] $\biguplus\{\Oo\mid (i,\Oo)\in \mathcal{Y}\}=X_{\tt check}$;
  \item[(3)] and $|\Oo|\geq i$ for all $(i,\Oo)\in \mathcal{Y}$.
\end{itemize}
\end{lemma}
\begin{proof}
Let $\mathcal{Y}_\ell$ denote the parameter $\mathcal{Y}$ at the $\ell^{th}$ call of {\sf Explore}.
Let us apply induction on $\ell$.
The base case $\ell=1$ immediately follows from the fact that $\mathcal{Y}_1=\{(d,X_{\tt check})\}$ (note that we assumed $|X_{\tt check}|\geq d$).
It remains to prove the inductive step.
Suppose the result holds at $\ell>1$ and $\mathcal{Y}_{\ell+1}=(\mathcal{Y}_\ell\setminus\mathcal{Y}')\cup\{(i-i_j,\mathcal{C}_{i,\Oo}),(i_j,\Oo\setminus\mathcal{C}_{i,\Oo}) \mid (i,\Oo)\in \mathcal{Y}'\}$, where $\mathcal{Y}'=\{(i,\Oo)\in \mathcal{Y}_\ell \mid |\Oo|>i\wedge i\neq 0\}$.

\begin{itemize}
  \item By applying the induction hypothesis, we get that $\sum\{i\mid (i,\Oo)\in \mathcal{Y}_{\ell}\}=d$. Since
\[\begin{array}{ll}
  \sum\{i\mid (i,\Oo)\in \mathcal{Y}_{\ell+1}\} &  =  \sum\{i\mid (i,\Oo)\in \mathcal{Y}_\ell\setminus\mathcal{Y}'\}+\sum\{i-i_j,i_j\mid (i,\Oo)\in \mathcal{Y}'\}\\
  &  = \sum\{i\mid (i,\Oo)\in \mathcal{Y}_\ell\setminus\mathcal{Y}'\}+\sum\{i\mid (i,\Oo)\in \mathcal{Y}'\}\\
  & =  \sum\{i\mid (i,\Oo)\in \mathcal{Y}_{\ell}\},\\
\end{array}\]
we conclude the proof of Item (1).

  \item By applying the induction hypothesis, we get that $\biguplus\{\Oo\mid (i,\Oo)\in \mathcal{Y}_\ell\}=X_{\tt check}$.
  Since \[\begin{array}{ll}
  \biguplus\{\Oo\mid (i,\Oo)\in \mathcal{Y}_{\ell+1}\} & =   \biguplus\{\Oo\mid (i,\Oo)\in \mathcal{Y}_\ell\setminus\mathcal{Y}'\}\uplus\biguplus\{\mathcal{C}_{i,\Oo},\Oo\setminus\mathcal{C}_{i,\Oo}\mid (i,\Oo)\in \mathcal{Y}'\}\\
  &  = \biguplus\{\Oo\mid (i,\Oo)\in \mathcal{Y}_\ell\setminus\mathcal{Y}'\}\uplus\biguplus \{\Oo\mid (i,\Oo)\in \mathcal{Y}'\}\\
  &  = \biguplus\{\Oo\mid (i,\Oo)\in \mathcal{Y}_\ell\},\\
\end{array}\]

we conclude the proof of Item (2).

  \item By applying the induction hypothesis, $|\Oo|\geq i$ for all $(i,\Oo)\in \mathcal{Y}_\ell\setminus\mathcal{Y}'$. For each pair $(i,\Oo)\in \mathcal{Y}_{\ell}$, according to Lines 13 and 17, $|\mathcal{C}_{i,\Oo}|\geq i$ , hence $|\mathcal{C}_{i,\Oo}|\geq i-i_j$.
Since $0\leq i_j\leq \min(i,|\Oo\setminus\mathcal{C}_{i,\Oo}|)$, we get that
$|\Oo\setminus\mathcal{C}_{i,\Oo}|\geq i_j$. We conclude the proof of Item (3).
\end{itemize}
\end{proof}

We now prove the termination of Algorithm~\ref{alg:DivHOMaRer}.
\begin{lemma}\label{lem:terminate}
Algorithm~\ref{alg:DivHOMaRer} always terminates.
\end{lemma}

\begin{proof} It suffices to show that the recursive procedure call of {\sf Explore}  always terminates.
Let $\mathcal{Y}_\ell$ denote the parameter $\mathcal{Y}$ at the $\ell^{th}$ call of {\sf Explore}.
By Lemma~\ref{lem:3property}(3), $|\Oo|\geq i$ for all $(i,\Oo)\in \mathcal{Y}_\ell$.
\begin{itemize}
  \item If ${|\Oo|}=i$ or $i=0$  for all $(i,\Oo)\in \mathcal{Y}_\ell$, then $\mathcal{Y}'=\emptyset$.
In this case, {\sf Explore} will not be called at Line 22 during the $\ell^{th}$ call. Hence, Algorithm~\ref{alg:DivHOMaRer} terminates.

  \item Otherwise, there are some $(i,\Oo)\in \mathcal{Y}_\ell$ such that ${|\Oo|}>i$ and $i\neq 0$.
Then, such pairs $(i,\Oo)$ are always partitioned into $(i-i_j,\mathcal{C}_{i,\Oo})$ and $(i_j,\Oo\setminus\mathcal{C}_{i,\Oo})$.
Since $\sum_{(i,\Oo)\in \mathcal{Y}'} i_j\neq 0$, by Lemma~\ref{lem:3property}(1 and 2),
there exists $\ell'> \ell$ for $\mathcal{Y}_{\ell'}$ such that ${|\Oo|}=i$ or $i=0$  for all $(i,\Oo)\in \mathcal{Y}_{\ell'}$. Hence, Algorithm~\ref{alg:DivHOMaRer} always terminates.
\end{itemize}
\end{proof}

We show the soundness of Algorithm~\ref{alg:DivHOMaRer}.

\begin{lemma}\label{lem:cover}
For every subset $\Oo\subseteq X_{\tt check}$ such that ${|\Oo|}=d$,
$\Oo$ is covered by Algorithm~\ref{alg:DivHOMaRer}, namely,
either $\Oo$ is added into ${\tt PLS}$,
or there exists a subset $\Oo'$ such that $\vdash \Oo:\sid$ is valid and  $\Oo\subseteq\Oo'$.
\end{lemma}

\begin{proof} Given a set $\mathcal{Y}$ of pairs,
let $Cover(\mathcal{Y})$ denote the set of subsets $\Oo\subseteq X_{\tt check}$ such that ${|\Oo|}=d$ and
$\Oo$ contains $i$ elements of $\Oo'$  for each pair $(i,\Oo')\in \mathcal{Y}$.
It suffices to show that for every call {\sf Explore}$(\mathcal{Y})$,
each subset $\Oo\in Cover(\mathcal{Y})$ is covered.

Let $\mathcal{Y}_\ell$ denote the parameter $\mathcal{Y}$ at the $\ell^{th}$ call of {\sf Explore}.
We apply induction on $\ell$, where the base case is the largest $\ell$. Note that such $\ell$ exists by Lemma~\ref{lem:terminate}.

\smallskip
\noindent{\bf Base case}. The base case is the largest $\ell$ such that $\mathcal{Y}'=\emptyset$ at the $\ell^{th}$ call of {\sf Explore}.
Since $\mathcal{Y}'=\{(i,\Oo')\in \mathcal{Y}_\ell \mid |\Oo'|>i\wedge i\neq 0\}=\emptyset$, by Lemma~\ref{lem:3property}(3),
${|\Oo'|}=i$ or $i=0$ for all $(i,\Oo')\in \mathcal{Y}_\ell$.
By Lemma~\ref{lem:3property}(1 and 2), $Cover(\mathcal{Y})$ is singleton set.
Suppose $Cover(\mathcal{Y})=\{\Oo\}$, then $\Oo$ is covered.
Indeed, either $\vdash \Oo:\sid$ is valid or $\Oo$ is added into ${\tt PLS}$.

\smallskip
\noindent  {\bf Inductive step}. There exists some pair $(i,\Oo')\in \mathcal{Y}_\ell$ such that ${|\Oo'|}>i$ and $i\neq 0$.
For every subset $\Oo\in Cover(\mathcal{Y}_\ell)$,
either $\Oo\subseteq \bigcup_{(i,\Oo')\in \mathcal{Y}_\ell} \mathcal{C}_{i,\Oo'}$
or $\Oo\not\subseteq \bigcup_{(i,\Oo')\in \mathcal{Y}_\ell} \mathcal{C}_{i,\Oo'}$.

\begin{itemize}
\item
If $\Oo\subseteq \bigcup_{(i,\Oo')\in \mathcal{Y}_\ell} \mathcal{C}_{i,\Oo'}$, then
$\Oo$ is covered. Indeed, $\vdash \bigcup_{(i,\Oo')\in \mathcal{Y}_\ell}\mathcal{C}_{i,\Oo'}:\sid$ is valid.

\smallskip
\item Otherwise $\Oo\not\subseteq \bigcup_{(i,\Oo')\in \mathcal{Y}_\ell} \mathcal{C}_{i,\Oo'}$, then
$\Oo$ contains at least one variable from $\Oo'\setminus \mathcal{C}_{i,\Oo'}$ for some pair $(i,\Oo')\in \mathcal{Y}_\ell$,
i.e., $\Oo\cap (\Oo'\setminus \mathcal{C}_{i,\Oo'})\neq \emptyset$.
There must exist a combination of values $i_j: 0\leq i_j\leq \min(i,|\Oo'\setminus\mathcal{C}_{i,\Oo'}|)$
for $(i,\Oo')\in \mathcal{Y}'$ such that
\[\Oo\in cover((\mathcal{Y}_\ell\setminus\mathcal{Y}')\cup\{(i-i_j,\mathcal{C}_{i,\Oo'}),(i_j,\Oo'\setminus\mathcal{C}_{i,\Oo'}) \mid (i,\Oo')\in \mathcal{Y}'\}).\]
By applying induction hypothesis, the subset $\Oo$ is covered.
\end{itemize}
We complete the proof.
\end{proof}

\begin{proof}[Proof of Theorem~\ref{thm:HOMV}]
%
If ${\tt PLS}=\emptyset$, then
by Lemma~\ref{lem:cover}, $\vdash \Oo:\sid$ is valid
for every size-$d$ subset $\Oo\subseteq X_{\tt check}$. Hence $P$ is order-$d$ secure.

On the other hand, for every size-$d$ subset $\Oo\subseteq X_{\tt check}$,
if $P$ is $\Oo$-leaky, then $\vdash \Oo:\sid$ is not valid.
By Lemma~\ref{lem:cover},
all the size-$d$ subsets $\Oo\subseteq X_{\tt check}$ are covered by Algorithm~\ref{alg:DivHOMaRer},
hence, $\Oo$ is added into ${\tt PLS}$.
\end{proof}

\begin{example}
We demonstrate Algorithm~\ref{alg:DivHOMaRer} on the motivating example (cf.\ Section~\ref{sec:motivatingexample}) for $d=2$.
First of all,  $X_{\tt check}=X_o=\{x',A,y_0,\cdots,y_5,r,r'\}$ as $X_p=\emptyset$.
After applying the transformation $\simply_{\Alg}$,
$\lambda$ and $\pi$ are given below:
\begin{center}
\begin{tabular}{lll}
 $\lambda(y_4)=r'\oplus k$,   & $\lambda(y_5)=(r'\oplus k)- (r' \oplus r)$, \\
 $\lambda(A)=\big(r'\oplus k)- (r' \oplus r)\big) \oplus
    \big((((k\oplus r) \oplus r')-r') \oplus (k\oplus r)\big)$.
\end{tabular}
\end{center}

\begin{center}
\begin{tabular}{lllll}
  $\pi(x')=\{r\}$, &  $\pi(r)=\{r\}$, & $\pi(r')=\{r'\}$,  & $\pi(y_0)=\{r,r'\}$,  & $\pi(y_1)=\{r\}$, \\
  $\pi(y_2)=\emptyset$,  & $\pi(y_3)=\{r,r'\}$, & $\pi(y_4)=\{r'\}$, & $\pi(y_5)=\{r\}$, & $\pi(A)=\emptyset$.
\end{tabular}
\end{center}
{\sc HOME} invokes {\sf Explore}$(\{(2,X_o)\})$.
Suppose  {\sf Explore} chooses $\{r,r'\}$ at Line 13, i.e., $\mathcal{C}_{2,X_o}=\{r,r'\}$
then, $\vdash \mathcal{C}_{2,X_o}:\rud$ is valid, namely, {\sf Check}$(\{\mathcal{C}_{2,X_o}\})$
will return $\top$. The loop at Lines 15-17
will iteratively test $x',y_0,\cdots,y_5,A$. Among them,
only $y_3$ can be added into $\mathcal{C}_{2,X_o}$ according to rule ({\sc Sid}$_2$).
Now, we can deduce that all size-2 subsets $\Oo\subseteq \mathcal{C}_{2,X_o}=\{r,r',y_3\}$
have type $\rud$ or $\sid$.

It is easy to see that $\mathcal{Y}'=\{(2,X_o)\}$, as $|X_o|>2$.
Therefore, at Line 22, the following two procedure calls will be made:
\begin{itemize}
  \item {\bf Call$_1$:}  {\sf Explore}$(\mathcal{Y}_1)$, where $\mathcal{Y}_1=\{(0,\{r,r',y_3\}),(2,\{x',A,y_0,y_1,y_2,y_4,y_5\})\}$
  \item {\bf Call$_2$:}  {\sf Explore}$(\mathcal{Y}_2)$, where $\mathcal{Y}_2= \{(1,\{r,r',y_3\}),(1,\{x',A,y_0,y_1,y_2,y_4,y_5\})\}$.
\end{itemize}

For {\bf Call$_1$}, suppose {\sf Explore} chooses $\{x',y_0\}$ at Line 13, i.e., $\mathcal{C}_{2,{\{x',A,y_0,y_1,y_2,y_4,y_5\}}}=\{x',y_0\}$.
We derive $\vdash \{x',y_0\}:\rud$ by applying rule ({\sc Rud}) with $\Oo_1=\{x'\}$ and $\Oo_2=\{y_0\}$.
The loop at Lines 15-17 will add $y_1$ into $\mathcal{C}_{2,{\{x',A,y_0,y_1,y_2,y_4,y_5\}}}$, as
replacing the common sub-expression $k\oplus r$ in $\{\Expr(x'),\Expr(y_0),\Expr(y_1)\}$ with $r$ using the transformation $\simply_{\Dom}$,
i.e., $\Expr(x')$, $\Expr(y_0)$ and $\Expr(y_1)$ respectively becoming $r$, $r\oplus r'$, $(r\oplus r')-r'$,
allow us to derive $\vdash \{x',y_0,y_1\}:\sid$. By replaying this transformation on $\Expr(y_2)$,
$y_2$ can also be added into $\mathcal{C}_{2,{\{x',A,y_0,y_1,y_2,y_4,y_5\}}}$, which becomes $\{x',y_0,y_1,y_2\}$.
However, $y_4,y_5$ and $A$ cannot be added into $\mathcal{C}_{2,{\{x',A,y_0,y_1,y_2,y_4,y_5\}}}$.

For  {\bf Call$_2$}, suppose {\sf Explore} chooses $\{r\}$ from $\{r,r',y_3\}$ and $x'$ from $\{x',A,y_0,y_1,y_2,y_4,y_5\}$
at Line 13. We cannot derive any type judgement $\vdash \{r,x'\}:\tau$ for $\tau\in \type$.
So $\{r,x'\}$ is added to ${\tt PLS}$.
Finally,
\[{\tt PLS}=\left\{
\begin{array}{ccccc}
 \{r,r',x',y_2,y_3,y_4,y_5\}\times \{A\}& \cup &  \{r',y_2\}\times \{y_4\} &\cup& \\
 \{r,x',y_0,y_1,y_2\}\times \{y_5\} &\cup &   \{r\}\times \{x',y_1,y_2\}& \cup  &\{y_2,y_3\}
\end{array}
\right\}.\]
 \end{example}

\section{Model-Counting and Pattern Matching based Methods}
\label{sec:modelcounting}
We propose three model-counting based methods (cf.\ Section~\ref{sec:smtmethod}---\ref{sec:gpu}) for resolving potential leaky observable sets prescribed by type inference algorithm.
Generally, model-counting is very costly, so we propose a complementary pattern matching based method
to efficiently resolve potential leaky observable sets from known sets, avoiding
a vast amount of model-counting usage.

\subsection{SMT-based Method}\label{sec:smtmethod}
We first lift the SMT-based method~\cite{GXZSC19} from first-order to higher-order.

Recall that $P$ is $\Oo$-leaky iff
$\sem{P}^\Oo_{\eta_1}\neq\sem{P}^\Oo_{\eta_2}$ for some pair $(\eta_1,\eta_2)\in \Theta^2_{=X_p}$.
Let $\Oo=\{x_1,\ldots,x_m\}$. For every valuation $\eta\in \Theta$ and tuple of values
$(c_1,\ldots,c_m)\in \Bb^m$,
let $\sharp_{\eta}(x_1=c_1,\cdots,x_m=c_m)$ denotes the number of assignments $\eta_r:X_r\rightarrow\Bb$ such that
for all $ 1\leq j\leq m$, $\Expr(x_j)$ evaluates to $c_j$ under $\eta$ and $\eta_r$.
Then,  $\Oo$-leaky can be characterized as the following 
logical formula:

 \begin{equation}\label{eq:omega}
\begin{split}
\Omega^\Oo:=\exists (\eta_1,\eta_2)\in \Theta^2_{=X_p} & \exists (c_1,\cdots,c_m)\in\Bb^m. \\
 &\big(\sharp_{\eta_1}(x_1=c_1,\cdots,x_m=c_m)\neq\sharp_{\eta_2}(x_1=c_1,\cdots,x_m=c_m)\big)
\end{split}
\end{equation}

\begin{proposition}\label{prop:logicalchara}
$\Omega^\Oo$ is satisfiable iff
$P$ is $\Oo$-leaky.
\end{proposition}

\begin{proof}
The program $P$ is $\Oo$-leaky iff  the following formula holds:
\begin{align*}
  \exists (\eta_1,\eta_2)\in \Theta^2_{=X_p}&  \exists (c_1,\cdots,c_m)\in\Bb^m.~ \sem{P}_{\eta_1}^\Oo(c_1,\cdots,c_m)\neq\sem{P}_{\eta_2}^\Oo(c_1,\cdots,c_m)
\end{align*}
Since
$\sem{P}_{\eta}^\Oo(c_1,\cdots,c_m) =\frac{\sharp_{\eta}(x_1=c_1,\cdots,x_m=c_m)}{2^{\kappa\times |X_r|}}$
for $\eta\in \{\eta_1, \eta_2\}$, then the program $P$ is $\Oo$-leaky iff the following formula holds:
\begin{align*}
  \exists (\eta_1,\eta_2)\in \Theta^2_{=X_p}&  \exists (c_1,\cdots,c_m)\in\Bb^m.~ \frac{\sharp_{\eta_1}(x_1=c_1,\cdots,x_m=c_m)}{2^{\kappa\times |X_r|}} \neq\frac{\sharp_{\eta_2}(x_1=c_1,\cdots,x_m=c_m)}{2^{\kappa\times |X_r|}}
\end{align*}
%
The result follows immediately.
\end{proof}

%
We further encode $\Omega^\Oo$
as a first-order logic formula that can be solved by SMT solvers (e.g., Z3~\cite{MB08}).
Suppose $\Expr(x_j)=e_j$ for $1\leq j\leq m$, let $E_\Oo=E^1\uplus E^2$ with $E^{1}=\{e \mid  \Var(e)\cap X_k\neq\emptyset\}$,
and $E^{2}=\{e \mid \Var(e)\cap X_k=\emptyset\}$.
We define the first-order logic formula $\Psi^\Oo$ as
\begin{center}
$\Psi^\Oo:=\left(
\begin{array}{c}
  \left(\bigwedge_{e\in E^{1}}\bigwedge_{f:\RVar(e)\rightarrow \Bb}(\Theta_{e,f}\wedge \Theta'_{e,f})\right)\\
    \bigwedge \\
   \left(\bigwedge_{e\in E^{2}}\bigwedge_{f:\RVar(e)\rightarrow \Bb}\Theta_{e,f}\right)\\
   \bigwedge \\
    \left (\Theta_{v2i}\wedge \Theta_{v2i}' \wedge \Theta_{\neq}\right)
\end{array}
\right)$, where
\end{center}

 \begin{itemize}[topsep=1pt,leftmargin=*]
\item {\bf Program logic ($\Theta_{e,f}$ and $\Theta_{e,f}'$)}:
for every $e=\Expr(x)\in E_\Oo$ and $f:\RVar(e)\rightarrow \Bb$,
$\Theta_{e,f}$ encodes the expression $e$ into a first-order logic formula
and asserts that the value of $e$ under the assignment $f$ is equal to
 a fresh variable $x_f$. (Note there are $2^{|\RVar(e)|}$ distinct conjuncts all of which share the variables $X_p\cup X_k$.)
  \medskip

  $\Theta_{e,f}'$ is similar to $\Theta_{e,f}$ except that $x_f$ and $k\in X_k$
  are replaced by fresh variables $x_f'$ and $k'$. Note that for every $e\in E^2$,
  we do not construct $\Theta_{e,f}'$, as $e\in E^2$ does not have any private variable $k\in X_k$
  and hence $\Theta_{e,f}'$ would be same as $\Theta_{e,f}$.
\medskip

 \item{\bf Vector to integer ($\Theta_{v2i}$ and $\Theta_{v2i}'$)}:
$\Theta_{v2i}$ asserts that for every $f:\bigcup_{e\in E_\Oo}\RVar(e)\rightarrow \Bb$,
a fresh integer variable $I_f$ is $1$
  if the vector $(x_f)_{x\in\Oo}$ is equal to the vector $(c_1,\cdots,c_m)$, otherwise $0$.
  By doing so,   we can count the number of assignments $f$'s of random variables under which
  $(e_1,\cdots,e_m)$ evaluates to $(c_1,\cdots,c_m)$ when variables $x\in X_p\cup X_k$ take some concrete values.
  Formally,
\[\Theta_{v2i}:=\bigwedge_{f:\bigcup_{e\in E_\Oo}\RVar(e)\rightarrow \Bb} ~ \Big(I_f= \big((c_1,\cdots,c_m)==(x_f)_{x\in\Oo}\ ? \ 1 : 0\big)\Big).
\]

  $\Theta_{v2i}'$ is similar to $\Theta_{v2i}$ except that $I_f$ is replaced by $I_f'$,
  and $x_f$ is replaced by $x_f'$ for all $x\in\Oo$ such that $\Expr(x)\in E^1$. Note that
  $k'\in X_k$ may have different value than $k$, but $x\in X_p$ has same value in $\Theta_{v2i}$ and $\Theta_{v2i}'$.
  This conforms to  $(\eta_1,\eta_2)\in \Theta^2_{=X_p}$ in Equ.~(\ref{eq:omega}).
\item{\bf Different sums ($\Theta_{\neq}$)}:  It asserts two sums of assignments $f$'s of random variables for variables $(k)_{k\in X_k}$ and $(k')_{k\in X_K}$ (i.e., integers $I_f$ and $I_f'$)  differ.
This conforms to $\big(\sharp_{\eta_1}(x_1=c_1,\cdots,x_m=c_m)\neq\sharp_{\eta_2}(x_1=c_1,\cdots,x_m=c_m)\big)$ in Equ.~(\ref{eq:omega}).
Formally,
\[
\Theta_{\neq}:=\sum_{f:\bigcup_{e\in E_\Oo}\RVar(e)\rightarrow \Bb} I_f\neq \sum_{f:\bigcup_{e\in E_\Oo}\RVar(e)\rightarrow \Bb} I_f'.
\]
\end{itemize}

\medskip
It is straightforward to get the following proposition,

\begin{proposition}\label{prop:unsat2SMT}
$\Omega^\Oo$ is satisfiable iff
$\Psi^\Oo$ is satisfiable, where the size of $\Psi^\Oo$ is exponential in the number of (bits of) random variables.
\end{proposition}

By Proposition~\ref{prop:logicalchara} and Proposition~\ref{prop:unsat2SMT}, we get that:

\begin{corollary}\label{cor:unsat2SMT}
$\Psi^\Oo$ is satisfiable iff
$P$ is $\Oo$-leaky.
\end{corollary}
%

\begin{figure}[t]
\centering
$\begin{array}{c}
    \left(\begin{array}{c}  \big(y_{000}=(k\oplus 0)\oplus 0 \big)\wedge \big(y_{000}'=(k'\oplus 0)\oplus 0\big)  \wedge
      \big(y_{001}=(k\oplus 1)\oplus 0 \big)\wedge \big(y_{001}'=(k'\oplus 1)\oplus 0 \big) \wedge \\
      \big(y_{010}=(k\oplus 0)\oplus 1 \big)\wedge \big(y_{010}'=(k'\oplus 0)\oplus 1 \big) \wedge
      \big(y_{011}=(k\oplus 1)\oplus 1 \big)\wedge \big(y_{011}'=(k'\oplus 1)\oplus 1 \big)  \wedge\end{array} \right)  \\
\Big((y_{300}=0\oplus 0 ) \wedge (y_{301}=0\oplus 1)\wedge
(y_{310}=1\oplus 0) \wedge (y_{311}=1\oplus 1) \wedge\Big) \\
 \left(\begin{array}{c} \big(I_{00}= (y_{000}=c_1\wedge y_{300}=c_2)  \ ? \ 1 : 0\big) \wedge
      \big(I_{01}= (y_{001}=c_1\wedge y_{301}=c_2)  \ ? \ 1 : 0\big) \wedge \\
      \big(I_{10}= (y_{010}=c_1\wedge y_{310}=c_2)  \ ? \ 1 : 0\big) \wedge
      \big(I_{11}= (y_{011}=c_1\wedge y_{311}=c_2)  \ ? \ 1 : 0\big)\wedge \end{array} \right) \\
 \left(\begin{array}{c} \big(I_{00}'= (y_{000}'=c_1\wedge y_{300}=c_2) \ ? \ 1 : 0\big) \wedge
     \big(I_{01}'= (y_{001}'=c_1\wedge y_{301}=c_2) \ ? \ 1 : 0\big) \wedge \\
      \big(I_{10}'= (y_{010}'=c_1\wedge y_{310}=c_2) \ ? \ 1 : 0\big) \wedge
           \big(I_{11}'= (y_{011}'=c_1\wedge y_{311}=c_2) \ ? \ 1 : 0\big) \wedge \end{array} \right) \\
 \Big(  (I_{00}+I_{01}+I_{10}+I_{11})\neq (I_{00}'+I_{01}'+I_{10}'+I_{11}')\Big)
\end{array}$ 
\caption{The SMT encoding $\Psi^{\{y_0,y_3\}}$.}
\label{fig:smty1y4}
\end{figure}

\begin{example}
Let us consider the observable set $\{y_0,y_3\}$ in the motivating example (cf. Section~\ref{sec:motivatingexample}).
Recall that $\Expr(y_0)=(k\oplus r) \oplus r'$ and $\Expr(y_3)=r' \oplus r$. In this case, $E^1=\{\Expr(y_0)\}$ and $E^2=\{\Expr(y_3)\}$.
The SMT formula $\Psi^{\{y_0,y_3\}}$ is shown in Figure~\ref{fig:smty1y4}, where the first two lines correspond to the logical formulas $\Theta_{\Expr(y_0),f}$,
the third line corresponds to the logical formulas $\Theta_{\Expr(y_3),f}$, the next four lines correspond to the logical formulas $\Theta_{v2i}$ and  $\Theta_{v2i}'$,
and the last one corresponds to the logical formula $\Theta_{\neq}$.
(For clarity, we only show the case when all variables are Boolean.)
$\Psi^{\{y_0,y_3\}}$  is satisfiable,
implying that the program is $\{y_0,y_3\}$-leaky.
\end{example}


\subsection{Brute-force Method}
The brute-force method (cf. Alg.~\ref{alg:bfalg}) enumerates all possible valuations 
and then computes corresponding distributions 
again by enumerating the assignments of random variables.

\begin{algorithm}[t]
\SetAlgoNoLine
\small
\SetKwProg{myfunc}{Function}{}{}
\myfunc{{\sc BFEnum}$(P,\Oo=\{x_1,\cdots,x_m\})$}{
    \ForAll{$\eta_p:\bigcup_{x\in \Oo}X_p\cap \Var(\Expr(x)) \rightarrow \Bb$}{
        $D_1:=\lambda (c_1,\cdots,c_m)\in \Bb^m.0$\;
        $b:={\tt false}$\;
        \ForAll{$\eta_k:\bigcup_{x\in \Oo}X_k\cap \Var(\Expr(x)) \rightarrow \Bb$}{
            $D_2:=\lambda (c_1,\cdots,c_m)\in \Bb^m.0$\;
            \If{$b={\tt false}$}{
                $D_1:=${\sc Counting}$(P,\Oo,\eta_p,\eta_k)$\;
                $b:={\tt true}$\;
            }
            \Else{
                $D_2:=${\sc Counting}$(P,\Oo,\eta_p,\eta_k)$\;
                \lIf{$D_1\neq D_2$}{\Return{${\tt SAT}$}}
            }
        }
    }
    \Return{${\tt UNSAT}$}\;
}
\medskip
\myfunc{{\sc Counting}$(P,\Oo=\{x_1,\cdots,x_m\},\eta_p,\eta_k)$}{
    \ForAll{$\eta_r:\bigcup_{x\in \Oo}\RVar(\Expr(x)) \rightarrow \Bb$}{
        $D[\Expr_{\eta_p,\eta_k,\eta_r}(x_1),\cdots,\Expr_{\eta_p,\eta_k,\eta_r}(x_m)]++$\;
    }
    \Return{$D$}\;
}
\caption{A brute-force algorithm\label{alg:bfalg}}
\end{algorithm}

\begin{proposition}
$\Omega^\Oo$ is satisfiable iff  Algorithm~\ref{alg:bfalg} returns ${\tt SAT}$.
\end{proposition}

The complexity of Algorithm~\ref{alg:bfalg} is exponential in the number of (bits of) variables in
computations $(\Expr(x))_{x\in\Oo}$,
so it would experience significant performance degradation when facing a large number of variables.
We propose a GPU-accelerated parallel algorithm to boost the performance. 

\subsection{GPU-accelerated Parallel Algorithm} \label{sec:gpu}
In this subsection, we show how to leverage GPU's superior compute
capability to check satisfiability of $\Omega^\Oo$ in Eqn.~(\ref{eq:omega}).
In general, given a potential leaky observable set $\Oo$, we automatically synthesize a GPU program from the computations of observable variables in $\Oo$
such that the GPU program outputs ${\tt SAT}$ iff $\Omega^\Oo$ is satisfiable, i.e.,
the program is $\Oo$-leaky.

Our work is based on CUDA, a parallel computing platform and programming model for NVIDIA GPUs.
Specifically, we utilize Nvidia GeForce GTX 1080 (Pascal) with compute capability 6.1.
From a programming perspective,
the CUDA architecture defines three levels of threads, i.e., grid, block and warp, to organize units.
A warp consists of 32 consecutive threads 
which are executed in the Single Instruction Multiple Thread fashion on Streaming Processors,
namely, all threads execute the same instruction, and each thread carries out that operation on its own private data.
A block running on Streaming Multiprocessors contains at most $32$ wraps (giving rise to $32\times 32$ threads).
The maximum number of blocks in a grid is $65,535\times 65,535$ and each grid runs on the Scalable Streaming Processor Array. The code running on GPUs is usually referred to as \emph{Kernel}.

We parallelize Algorithm~\ref{alg:bfalg} as a CUDA program. In this work, we illustrate the idea on byte programs, i.e., each
variable is of 8-bit. 
Typically, the number of random variables is usually much larger than that of the other 
variables. Therefore, we enumerate assignments of random variables in GPUs while enumerate valuations of public and input variables in CPUs. Namely, the {\sc Counting} function in Algorithm~\ref{alg:bfalg} is implemented as a Kernel. However, it would be difficult to implement a generic
Kernel to compute distributions of sets of computations $(\Expr(x))_{x\in\Oo}$, unless
computations are evaluated by traversing their abstract syntax trees, which is control-flow intensive and would downgrade the GPU performance.
As a result, instead of designing a generic Kernel, for each observable set  $\Oo$, we automatically synthesize a CUDA program which checks whether
$\Omega^\Oo$ is satisfiable based on Algorithm~\ref{alg:bfalg}.

\begin{algorithm}[t]
\SetAlgoNoLine
\small
\SetKwProg{myfunc1}{}{}{}
{\sf\_\_device\_\_  unsigned char} {\sc op1}$(...)$\\{
\hspace*{10pt}  ...\;
}
...\\
\medskip
{\sf\_\_device\_\_  unsigned char} {\sc opj}$(...)$\\{
\hspace*{10pt}  ...\;
}
{\sf\_\_device\_\_  unsigned char} {\sc exp$_1$}$(\eta_p,\eta_k,\eta_r,{\tt threadIdx},{\tt blockIdx})$\\{
\hspace*{10pt}  ...\;
}
...\\
\medskip
{\sf\_\_device\_\_  unsigned char} {\sc exp$_m$}$(\eta_p,\eta_k,\eta_r,{\tt threadIdx},{\tt blockIdx})$\\{
\hspace*{10pt}  ...\;
}

\SetKwProg{myfunc}{int main}{}{}
\myfunc{{\sc GPUBFEnum}$(P,\Oo=\{x_1,\cdots,x_m\})$}{
    {\sf int} *$D_1$; {\sf int} *$D_2$\;
    {\sf cudaMallocManaged}$(\&D_1,256^m)$\;  {\sf cudaMallocManaged}$(\&D_2,256^m)$\;
    {\sf dim3} block(16,16)\; {\sf dim3} grid(4096/block.x,4096/block.y)\;
    \ForAll{$\eta_p:\bigcup_{x\in \Oo}X_p\cap \Var(\Expr(x)) \rightarrow \Bb$}{
        {\sf memset}$(D_1,0,{\sf sizeof(unsigned char)})$\;
        $b:={\tt false}$\;
        \ForAll{$\eta_k:\bigcup_{x\in \Oo}X_k\cap \Var(\Expr(x)) \rightarrow \Bb$}{
            {\sf memset}$(D_2,0,{\sf sizeof(unsigned char)})$\;
            \If{$b={\tt false}$}{
                {\sc KernelCounting}$<<<$grid, block$>>>(D_1,\eta_p,\eta_k)$\;
                {\sf cudaDeviceSynchronize}()\;
                $b:={\tt true}$\;
            }
            \Else{
                {\sc KernelCounting}$<<<$grid, block$>>>(D_2,\eta_p,\eta_k)$\;
                {\sf cudaDeviceSynchronize}()\;
                \If{$D_1\neq D_2$}{\Return{${\tt SAT}$}\;}
            }
        }
    }
    \Return{${\tt UNSAT}$}\;
}
\medskip
{\tt \_\_global\_\_ void} {\sc KernelCounting}$(D,\eta_p,\eta_k)$\\
\hspace*{10pt} $\{r_1,\cdots,r_h\}:=\bigcup_{x\in \Oo}\RVar(\Expr(x))$\;
\hspace*{10pt} \ForAll{$\eta_r:\{r_4,\cdots,r_h\} \rightarrow \Bb$}{
\hspace*{10pt}    $c_1:=${\sc exp$_1$}$(\eta_p,\eta_k,\eta_r,{\tt threadIdx},{\tt blockIdx})$\;
\hspace*{10pt}    ...\\
\hspace*{10pt}    $c_m:=${\sc exp$_m$}$(\eta_p,\eta_k,\eta_r,{\tt threadIdx},{\tt blockIdx})$\;
\hspace*{10pt}    ${\tt index}:=\sum_{i=0}^{m-1} c_i\times 256^i$\;
\hspace*{10pt}    {\sf atomicAdd}$(\&D[{\tt index}],1)$\;
    }

\caption{The skeleton of synthesized GPU programs}
\label{alg:bfalggpu}
\end{algorithm}

The numbers of threads per block and blocks per grid in each synthesized CUDA program are determined
by the number $R:=|\bigcup_{x\in\Oo}\RVar(\Expr(x))|$ of random variables in $(\Expr(x))_{x\in\Oo}$.
If $R=3$, we choose 2-D $(16,16)$ blocks each of which has $2^8$ threads,
and 2-D $(256,256)$ grids each of which has $2^{16}$ blocks. (Note that the number $2^{24}$ of threads exactly corresponds to the number of valuations of
three 8-bits random variables.) Moreover, we do not need to enumerate those valuations,
as the thread Id and block Id (i.e., \textsf{threadIdx} and \textsf{blockIdx} in CUDA) of each host thread in GPU exactly corresponds one of those valuations.
If 
$R<3$, we reduce
the number of blocks and/or threads such that the total number of threads is the number of valuations
of random variables. Otherwise $R>3$, 
we set $2^8$ number of threads in each block and $2^{16}$ number blocks in one grid for three random variables,
while the valuations of the rest of the random variables are enumerated in GPU.

For each operation used in computations of $(\Expr(x))_{x\in\Oo}$ but is not supported in CUDA, we synthesize
a corresponding \textsf{\_\_device\_\_} function which will be called from GPUs only and executed therein.
For each computation $\Expr(x)$, we also synthesize a \textsf{\_\_device\_\_} function
which computes the value of $\Expr(x)$ using \textsf{\_\_device\_\_} functions for operations based on
thread Id and block Id of the host thread which represent the valuations of some random variables.


For memory management, we use int arrays to store distributions, which are accessed both from
CPU for comparing distributions (read-only) and GPU for computing distributions (read and write).
We utilize unified memory provided by CUDA to allocate memory for both int arrays,  
namely, to allocate memory by invoking the \textsf{cudaMallocManaged} function,
by which the managed pointers to int arrays are valid on both the GPU and CPU.
To resolve data race, the update of int arrays in Kernel is performed in one atomic transaction (via the \textsf{atomicAdd} function in CUDA).

Concretely,
Algorithm~\ref{alg:bfalggpu} shows a skeleton of synthesized GPU programs, where the number of random variables is greater than $3$.
Other cases are similar.
The {\sf\_\_device\_\_} functions implement all the CUDA non-supported operations and expressions which are invoked and executed on GPU.
$D_1$ and $D_2$ are int arrays for storing distributions.
The function {\sc KernelCounting} is the Kernel which computes distributions
for each valuation of public and private input variables.
The function {\sc KernelCounting} is invoked at Line 20 and Line 24 for each valuation of public and private input variables.
After each invoking of  {\sc KernelCounting}, the function {\sf cudaDeviceSynchronize} is invoked which waits until all preceding commands in all streams of all host threads have completed.
In the body of {\sc KernelCounting},
the valuations of the first three random variables are implicitly represented by \textsf{threadIdx} and \textsf{blockIdx},
while $\eta_r$ denotes a valuation of other random variables.
Finally, the values of expressions are iteratively computed via calling the corresponding {\sf\_\_device\_\_} functions.
The value vector of expressions is encoded as an index to the array $D$, where the value at this index increases $1$ atomically to avoid data race.



\subsection{Method based on Pattern Matching} \label{sec:pm}
In order to avoid (costly) model-counting,
we propose a novel pattern matching based method,
which allows  to resolve potential leaky observable sets more efficiently.
This idea comes from the observation that
cryptographic programs  usually have very similar blocks and 
many observable sets share
common observable variables. As a warmup, let us first consider two observable sets $\{x,y\}$
and $\{x',y'\}$, where $\Expr(x)=r$, $\Expr(y)=k\oplus r$, $\Expr(x')=r'$, $\Expr(y)=k\oplus r'$,
$k$ is a private input and $r,r'$ are two random variables.
Then $\{\Expr(x),\Expr(y)\}$ and $\{\Expr(x'),\Expr(y')\}$ are equivalent up to renaming of random variables,
thus, observable sets $\{x,y\}$ and $\{x',y'\}$ have same distribution type.

Based on this observation,
we propose a pattern matching method for inferring distribution types of observable sets $\Oo$ from
observable sets $\Oo'$  whose distribution types are known.
Before formalizing this idea, we first introduce type-respecting bijection functions.

Given a bijective function $f:X\rightarrow X$, the function
$f$ is \emph{type-respecting} if for every $x\in X$,
$f(x)$ is public (resp. private and random) iff $x$ is public (resp. private and random).

\begin{definition}
Two sets of computations $E$ and $E'$ are {\em isomorphic respecting
the type of variables}, denoted by $E\simeq E'$, if there is a type-respecting bijection $h:\Var(E)\rightarrow \Var(E')$
such that
$E'=\{h(e)\mid e\in E\}$, where $h(e)$ denotes the computation obtained from
$e$ by renaming each variable $x$ with $h(x)$.

\end{definition}

For two observable sets $\Oo$ and $\Oo'$ with the same size,
it is easy to see that $\Oo$ and $\Oo'$ have  the same distribution type
if $\{\Expr(x)\mid x\in\Oo\}\simeq \{\Expr(x')\mid x'\in\Oo'\}$.

One may notice that constants have to be preserved in the definition of isomorphic with respect to
the type of variables. In general, changing a constant in $E$ may change
its distribution type. For instance, let us consider a family of sets $E_i$ of (simplified)
computations taking from the fourth-order masked implementation of the Sbox~\cite{SP06},
\begin{center}
$E_{i,j}:=\{x_0, \ \text{Sbox}(k \oplus j \oplus x_0) \oplus r, \ \text{Sbox}(k  \oplus i\oplus x_0) \oplus r\}$,
for $0\leq i\neq j\leq 255$.
\end{center}
where $x_0$ and $r$ are two random variables and $k$ is a private input.
In this case, for any distinct pairs of constants $(i,j)$ and $(i',j')$,  $E_{i,j}\simeq E_{i',j'}$ does not hold,
thus, we cannot infer the distribution of $E_{i,j}$ from $E_{i',j'}$, although they are almost identical.

To address this issue, we propose a generalization taking into account constants.
Our idea is inspired on the observation that some constant can be assimilated without affecting the distribution of computations.
For instance, regarding $k \oplus j$ to be $k'$, then $k \oplus i\equiv (k'\oplus j)\oplus i\equiv k'\oplus (j\oplus i)$.
Suppose the distribution of $E_{1,2}$ is known and by applying $k \oplus i\equiv k'\oplus (j\oplus i)$, $E_{1,2}$ is normalized as  (note $3=1\oplus 2$)
\begin{center}
${\tt norm}(E_{1,2}):=\{x_0, \ \text{Sbox}(k' \oplus x_0) \oplus r, \ \text{Sbox}(k'  \oplus 3\oplus x_0) \oplus r\}$.
\end{center}
  Then, for any $0\leq i\neq j\leq 255$ such that $(j\oplus i)=3$,
by applying $k \oplus i\equiv k'\oplus (j\oplus i)$, $E_{i,j}$ is also normalized as
\begin{center}
${\tt norm}(E_{i,j}):=\{x_0, \ \text{Sbox}(k' \oplus x_0) \oplus r, \ \text{Sbox}(k'  \oplus 3\oplus x_0) \oplus r\}$.
\end{center}
We can observe that $E_{i,j}\simeq E_{1,2}$, thus, $E_{i,j}$ has the same distribution as $E_{1,2}$.
This idea is formalized in the following definition.

\begin{definition}
A constant $c$ is {\em assimilable} in a set $E$ of computations  if $E$ can be transformed into
a set  $E'$ of equivalent computations by algebra laws such that
  all occurrences of the constant $c$ in $E'$ are  within the context of $x\circ c$ 
	for operator $\circ\in\{\oplus,+,-\}$ and some variable $x$, such that $x$ is either not used elsewhere or used as $x\circ c'$ 
 for some constant $c'$ (note that $c\neq c'$).

%
\end{definition}

If $c$ is assimilable in $E$, we denote by ${\tt norm}(E)$ the set of normalized computations which is obtained
from $E'$ by iteratively
\begin{enumerate}
  \item replacing all occurrences of the constant $c$ in $E'$ (as $x\circ c$) by $x$, and
  \item every possible $x\circ c'$ (for $c'\neq c$) by  $x{\circ} c''$, where $c''=c'~\widehat{\circ} ~c$, $\widehat{+}=-$,
$\widehat{-}=-$ and $\widehat{\oplus}=\oplus$. 
\end{enumerate}
By this replacement, ${\tt norm}(E)=E'[x/(x\circ c)][\forall (x\circ c')~in~E':(x{\circ} c'')/(x\circ c')]$, one can reduce the number of constants in $E$.


\begin{example}
Let us consider $e=(x\oplus 1)+ (x\oplus 2)+ (y\oplus 1)$. The constant $1$ is \emph{not} assimilable because there are two occurrences of $1$ which are within two different contexts $x\oplus 1$ and $y\oplus 1$ respectively.
However, $2$ is assimilable  (by $x$), as $2$ occurs in the context of $x\oplus 2$ and $x$ occurs elsewhere in $x\oplus 1$.
After replacing $x\oplus 2$ by $x$ and $x\oplus 1$ by $x\oplus 3$,
we get $(x\oplus 3)+ x+ (y\oplus 1)$ in which $1$ becomes assimilable by $y$.
Finally, ${\tt norm}(E)= (x\oplus 3)+ x+ y$. 
\end{example}


\begin{theorem}\label{thm:matching}
For observable sets $\Oo$ and $\Oo'$,
if ${\tt norm}(\{\Expr(x)\mid x\in\Oo\})\simeq {\tt norm}(\{\Expr(x)\mid x\in\Oo'\})$, then $\Oo$ and $\Oo'$ have same the distribution types.
\end{theorem}

\begin{proof}
Let $E=\{\Expr(x)\mid x\in\Oo\}$ and $E'=\{\Expr(x)\mid x\in\Oo'\}$.
Let $n$ be the number of constants assimilated when computing ${\tt norm}(E)$ and ${\tt norm}(E')$.
We prove by applying induction on $n$.
\begin{itemize}[topsep=1pt,leftmargin=*]
  \item {\bf Base case} $n=0$. Then, ${\tt norm}(E)$ and ${\tt norm}(E')$ are isomorphic respecting
the type of variables. Let $h:\Var(E)\rightarrow \Var(E')$ be the type-respecting bijection, then for every pair $(\eta_1,\eta_2)\in \Theta^2_{=X_p}$, there exists
a pair $(\eta_1',\eta_2')\in \Theta^2_{=X_p}$ such that
$\eta_i(x)=\eta_i'(h(x))$ for all $i\in\{1,2\}$ and $x\in (X_p\cup X_k)\cap \Var(E)$.
Moreover, $\sem{P}^\Oo_{\eta_1}=\sem{P}^\Oo_{\eta_2}$ iff $\sem{P}^{\Oo'}_{\eta_1'}=\sem{P}^{\Oo'}_{\eta_2'}$.

\smallskip
For every pair $(\eta_1,\eta_2)\in \Theta^2_{=X_p}$, there exists
a pair $(\eta_1',\eta_2')\in \Theta^2_{=X_p}$ such that
$\eta_i(h^{-1}(x))=\eta_i'(x)$ for all $i\in\{1,2\}$ and $x\in (X_p\cup X_k)\cap \Var(E')$,
and $\sem{P}^\Oo_{\eta_1}=\sem{P}^\Oo_{\eta_2}$ iff $\sem{P}^{\Oo'}_{\eta_1'}=\sem{P}^{\Oo'}_{\eta_2'}$.
Thus, the result immediately follows.

  \item {\bf Inductive set} $n\geq 1$. Without loss of generation, we assume
  that $c$ is assimilated by $x$ as $x\circ c$ in $E$, and $x\circ c_1,\cdots, x\circ c_k$ are
  all the occurrences of $x$ with constants $c_1,\cdots, x_k$.
  Then, for every $\eta_1\in \Theta$,
  $\sem{P}^\Oo_{\eta_1}$ using $E$ and $\sem{P}^\Oo_{\eta_1[(\eta_1(x)\circ c)/x]}$
  using $E[x/(x\circ c)][(x\circ c_1')/(x\circ c_1),\cdots, (x\circ c_k')/(x\circ c_k)]$
  have same distribution, where $c_i'=c_i\widehat{\circ} c$ for all $i$.

  By symmetry, for every $\eta_2\in \Theta$,
  $\sem{P}^\Oo_{\eta_2}$ using $E$ and $\sem{P}^\Oo_{\eta_2[(\eta_2(x)\circ c)/x]}$
  using $E[x/(x\circ c)][(x\circ c_1')/(x\circ c_1),\cdots, (x\circ c_k')/(x\circ c_k)]$
  have same distribution.

  Therefore, for every $(\eta_1,\eta_2)\in \Theta^2_{=X_p}$,
  $\sem{P}^\Oo_{\eta_1}=\sem{P}^\Oo_{\eta_2}$ using $E$
  iff $\sem{P}^\Oo_{\eta_1[(\eta_1(x)\circ c)/x]}=\sem{P}^\Oo_{\eta_2[(\eta_2(x)\circ c)/x]}$ using $E[x/(x\circ c)][(x\circ c_1')/(x\circ c_1),\cdots, (x\circ c_k')/(x\circ c_k)]$

By applying the induction hypothesis: for every $(\eta_1,\eta_2)\in \Theta^2_{=X_p}$,
there exist a pair $(\eta_1',\eta_2')\in \Theta^2_{=X_p}$ such that
$\eta_i(h^{-1}(x))=\eta_i'(x)$ for all $i\in\{1,2\}$ and $x\in (X_p\cup X_k)\cap \Var(E')$,
and $\sem{P}^\Oo_{\eta_1}=\sem{P}^\Oo_{\eta_2}$ using $E$ iff $\sem{P}^{\Oo'}_{\eta_1'}=\sem{P}^{\Oo'}_{\eta_2'}$.
Hence, $\sem{P}^\Oo_{\eta_1[(\eta_1(x)\circ c)/x]}=\sem{P}^\Oo_{\eta_2[(\eta_2(x)\circ c)/x]}$ using $E[x/(x\circ c)][(x\circ c_1')/(x\circ c_1),\cdots, (x\circ c_k')/(x\circ c_k)]$ iff $\sem{P}^{\Oo'}_{\eta_1'}=\sem{P}^{\Oo'}_{\eta_2'}$.
Thus, the result immediately follows.
\end{itemize}
\end{proof}

Remark that pattern matching based method could be used to match secure sets and
for program debugging. When a new program is just a minor revision of a verified program,
this method may be able to quickly check many observable sets.

%
%

\section{Implementation and Evaluation} \label{sec:exper}
We have implemented our methods in a tool  \tool. We use Z3~\cite{MB08} as the underlying SMT solver (fixed size bit-vector theory) for the SMT-based method. The tool works as follows:
\begin{enumerate}
\item  apply Algorithm~\ref{alg:DivHOMaRer} to compute the set of potential leaky observable sets;
\item  for each procedure call {\sf Check}$( \{\mathcal{C}_{i,\Oo}\}_{(i,\Oo)\in \mathcal{Y}} )$ at Line 14, when the type inference fails to derive any
distribution type of $\bigcup_{(i,\Oo)\in \mathcal{Y}} \mathcal{C}_{i,\Oo}$,  check whether there is a recorded set of computations $E'$ such that
${\tt norm}(E)\simeq {\tt norm}(E')$ via the pattern matching based method, where $E$ is the set of computations $\{\Expr(x)\mid x \bigcup_{(i,\Oo)\in \mathcal{Y}} \mathcal{C}_{i,\Oo}\}$
after transformations (e.g., $\simply_{\Alg}$, $\simply_{\Dom}$ and $\simply_{\Col}$), if $E'$ exists, then
returns the distribution type of $E'$;
\item  if $E'$ does not exist, apply mode-counting methods to the set $E$ and record $E$ with its corresponding distribution type for later pattern matching.
\end{enumerate}
Finally, ${\tt PLS}$ contains exactly the set of leaky observable sets.
Note that we do not apply pattern matching and model-counting based methods to observable sets whose size is greater than the security order $d$
for efficiency consideration.



\medskip

The experiments were conducted on arithmetic programs over the byte domain. We used a server with 64-bit Ubuntu 16.04.4 LTS, Intel Xeon CPU E5-2690v4, 2.6GHz and 256GB RAM (only one core is used in our computation). For GPU based algorithms, we use NVIDIA GeForce GTX 1080 with compute capability 6.1, as mentioned in Section~\ref{sec:gpu}.




\subsection{Evaluation on Higher-Order Masking} \label{ex:hom}

We evaluate our methods on implementations of masked arithmetic algorithms, ranging from multiplication algorithms to (round-reduced or full) AES/MAC-Keccak. Some of them are provided by authors of~\cite{BBDFGS15}, while the others are implemented according to the published masked algorithms.

The results of our type inference (Algorithm~\ref{alg:DivHOMaRer} without applying model-counting and pattern matching based methods) are presented in Table~\ref{tab:higherorder}.
Column 1 shows the reference and description of the program,
where A2B and B2A denote the implementations of conversion algorithms from Boolean to arithmetic masking
and arithmetic to Boolean masking respectively, SecH and SecR denote the implementations of the non-linear transformation and the round function of {\sc Simon}~\cite{TangZZQ15},
and DOM AND is a $GF(2^8)$ version from~\cite{GrossMK17}.
Columns 2-7 show the statistics of Algorithm~\ref{alg:DivHOMaRer},
including the numbers of potential leaky observable sets,
tuples that should be considered with respect to masking order $d$ (i.e., all non-empty subsets of $X_o$ with size $\leq d$) in order to compare with the tool of~\cite{BBDFGS15},
sets actually checked by Algorithm~\ref{alg:DivHOMaRer},
sets whose verification involves the $\simply_{\Dom}$ and $\simply_{\Col}$ transformations,
and verification time (\emph{excluding} program parsing).
Likewise, Columns 8--11 show the results reported in~\cite{BBDFGS15}, the unique sound (but incomplete) approach that is able to automatically verify masked implementations of higher-order
arithmetic programs under an equivalent leakage model of the ISW model.
Since the tool of~\cite{BBDFGS15} is unavailable, in Columns 8--11, we simply replicate
the statistics of the BBDFGS algorithm from the paper~\cite{BBDFGS15} when it is available (N/A is marked otherwise).
Recall that \cite{BBDFGS15} used a different experimental setup: a headless VM with a dual
core 64-bit processor clocked at 2GHz (only one core is used in the computation).
Note that Sbox~\cite{SP06} under fourth-order masking is verified under third-order security only in order to compare
with~\cite{BBDFGS15}, while other benchmarks are verified under their masking orders.

\begin{table}[t]
\centering \setlength{\tabcolsep}{2pt} 
\caption{Experimental results of type inference on masked programs.}
\label{tab:higherorder}
 \scalebox{0.72}{
\begin{tabular}{c|cccccc|cccc}
\toprule
\multirow{2}{*}{\bf Description}   & \multicolumn{6}{c|}{\bf \tool}  & \multicolumn{4}{c}{\cite{BBDFGS15}}  \\\cline{2-11}
           &     {\bf Result}   & {\bf $\sharp$Tuples} & {\bf $\sharp$Sets} & {\bf $\sharp\simply_{\Dom}$} & {\bf $\sharp\simply_{\Col}$}     & {\bf Time (s)}   &  {\bf Result} & {\bf $\sharp$Tuples}  & {\bf $\sharp$Sets}    & {\bf Time (s)}     \\ \midrule
\multicolumn{11}{c}{\bf First-Order Masking}   \\  \midrule \rowcolor{gray!20}
Multiplication~\cite{RP10}    & 0      & 13       & 6      &   5     &  0      & $\approx$0  & 0 & 13        & 7       & $\approx$0        \\
Sbox (4)~\cite{CPRR13}        & 0      & 73       & 15     & 14      & 0       & $\approx$0  & 0& 64        & 17      & $\approx$0         \\  \rowcolor{gray!20}
Full AES (4)~\cite{CPRR13}    &  0     & 20,060   &{\bf 515}&  514 & 0       & {\bf 2}     &  0&17,206    & 3,342   & 128                \\
Full Keccak~\cite{BBDFGS15} &  0     & 18,218   & {\bf 2,813}  & 2,813  &  0     & {\bf 83}  & 0   & 13,466    & 5,421   & 405                    \\  \rowcolor{gray!20}
B2A~\cite{Goubin01}           &{\bf 1}  & 10       & 4      &    2    &   0     & $\approx$0  & \multicolumn{4}{c}{\bf N/A} \\
A2B ~\cite{Goubin01}          &{\bf 37} & 48       & 39     &  1      &  0      & 0.15        & \multicolumn{4}{c}{\bf N/A}  \\  \rowcolor{gray!20}
A2B~\cite{CGV14}              &  0     &  1,448      &  14     &  13      &  0      &  $\approx$0 & \multicolumn{4}{c}{\bf N/A}   \\
B2A~\cite{CGV14}              &  0     & 2,494          & 2        & 1       & 0        & $\approx$0            & \multicolumn{4}{c}{\bf N/A}   \\  \rowcolor{gray!20}
A2B~\cite{CGTV15}             &{\bf 45}& 86       & 56      & 12       & 0      & $\approx$0  & \multicolumn{4}{c}{\bf N/A}   \\
B2A~\cite{BCZ18}              & 0      & 19       & 3      &   2     & 0       &  $\approx$0 & \multicolumn{4}{c}{\bf N/A}    \\  \rowcolor{gray!20}
B2A~\cite{Coron17}            &{\bf 1}      & 14       & 4      &   2     & 0       &  $\approx$0 & \multicolumn{4}{c}{\bf N/A}    \\ \midrule

\multicolumn{11}{c}{\bf Second-Order Masking}      \\ \midrule \rowcolor{gray!20}
Sbox~\cite{SP06}              &   0    &1,188,111 & {\bf 1,285}  & 1,284   & 256     & {\bf 1.073}        &  0&1,188,111   & 4,104   & 1.649       \\
Multiplication~\cite{RP10}    &  0     & 435      & 52     & 51      & 0       & 0.001        &  0& 435         & 92      & 0.001      \\  \rowcolor{gray!20}
Sbox~\cite{RP10}              &{\bf 2} & 7,503    & {\bf 270}    & 267     &  0      & 0.05         & 2& 7,140        & 866    & 0.045      \\
Key schedule~\cite{RP10}      &  0     &31,828,231& {\bf 475,943} & 475,942   & 0  &  {\bf 3,087}& 0 & 23,041,866   & 771,263 & 340,745 \\  \rowcolor{gray!20}
B2A~\cite{BCZ18}             & 0      &  1,653   &  25     &   23    &  0     & $\approx$0   & \multicolumn{4}{c}{\bf N/A}         \\
{B2A~\cite{SchneiderPOG19}}             & 0      &  780   &  15     &   13    &  0     & $\approx$0   & \multicolumn{4}{c}{\bf N/A}         \\ \rowcolor{gray!20}
{SecH (2)~\cite{TangZZQ15}}             & 0      &  1,770   &  14     &   13    &  0     & $\approx$0   & \multicolumn{4}{c}{\bf N/A}         \\
{SecR~\cite{TangZZQ15}}             & 0      &  3,003   &  25     &   24    &  0     & $\approx$0   & \multicolumn{4}{c}{\bf N/A}         \\  \rowcolor{gray!20}
DOM AND~\cite{GrossMK17}             & 0      &  435   &  46     &   45    &  0     & $\approx$0   & \multicolumn{4}{c}{\bf N/A}         \\ \midrule

\multicolumn{11}{c}{\bf Third-Order Masking}    \\  \midrule \rowcolor{gray!20}
Multiplication~\cite{RP10}    & 0      & 24,804   & 713     &   712   &  0     & {\bf 0.021}        & 0& 24,804        & 1,410    & 0.033   \\
Sbox (4)~\cite{CPRR13}        & 0      &6,784,540 & {\bf 18,734}  & 18,733  & 0      & {\bf 2.021}        & 0& 4,499,950     & 33,075   & 3.894  \\  \rowcolor{gray!20}
Sbox (5)~\cite{CPRR13}        &  0     &6,209,895 & {\bf 10,470}  & 10,469  & 0      & {\bf 3.757}        & 0& 4,499,950     & 39,613   & 5.036  \\
B2A~\cite{CGV14}              &  0     &274,884,292,760 & 7     &  6   &   0    & 0.11        & \multicolumn{4}{c}{\bf N/A}    \\  \rowcolor{gray!20}
B2A~\cite{BCZ18}              &  0     & 457,310  & 816     &  807    & 0      & 0.052        & \multicolumn{4}{c}{\bf N/A}    \\
{B2A~\cite{SchneiderPOG19}}             & 0      &  59,640   &  133     &   132    &  0     & $\approx$0   & \multicolumn{4}{c}{\bf N/A}         \\  \rowcolor{gray!20}
DOM AND~\cite{GrossMK17}             & 0      &  23,426   &  572     &   571    &  0     & $\approx$0   & \multicolumn{4}{c}{\bf N/A}         \\ \midrule

\multicolumn{11}{c}{\bf Fourth-Order Masking}   \\ \midrule \rowcolor{gray!20}
Sbox~\cite{SP06}              & {\bf 98,176} & 4,874,429,560  & {\bf 1,087,630}   & 924,173  & 821,888    & {\bf 702}         &  98,176& 4,874,429,560        & 35,895,437              & 22,119    \\
Multiplication~\cite{RP10}    & 0      & 2,024,785      & 12,845      & 12,844   & 0          & {\bf 0.534}       &  0&2,024,785            & 33,322                  & 1.138      \\  \rowcolor{gray!20}
Sbox (4)~\cite{CPRR13}        & 0      & 3,910,710,930  & {\bf 1,159,295}   & 1,159,294  &  0       & {\bf 376}         &  0&2,277,036,685       & 3,343,587               & 879          \\
B2A~\cite{BCZ18}              & 0      & 387,278,970    & 62,570      &  62,561               & 0           & 10.7                   & \multicolumn{4}{c}{\bf N/A}          \\  \rowcolor{gray!20}
{B2A~\cite{SchneiderPOG19}}             & 0      &  6,438,740   &  1,271     &   1,270    &  0     & 0.11   & \multicolumn{4}{c}{\bf N/A}         \\
DOM AND~\cite{GrossMK17}             & 0      &  2,024,785   &  10,626     &   10,625    &  0     & 0.71   & \multicolumn{4}{c}{\bf N/A}         \\ \midrule

\multicolumn{11}{c}{\bf Fifth-Order Masking}      \\ \midrule \rowcolor{gray!20}
Multiplication~\cite{RP10}    &  0    & 216,071,394    & {\bf 281,731}     & 281,730               & 0           & {\bf 15}            &  0&216,071,394     & 856,147    & 45      \\
Sbox (4)~\cite{CPRR13}        &  0    & 2,782,230,535,161   & 99,996,680   & 99,996,679                  &    0           & 49,598                & \multicolumn{4}{c}{\bf N/A}        \\  \rowcolor{gray!20}
{B2A~\cite{SchneiderPOG19}}             & 0      &  901,289,592   &  29,926     &   29,838    &  0     & 3.03   & \multicolumn{4}{c}{\bf N/A}         \\ \bottomrule
\end{tabular}}
\end{table}

\smallskip\noindent{\bf Results on common benchmarks.}
All the programs not marked as N/A in Columns 8--11 are 
provided by the authors of~\cite{BBDFGS15}.
We only did necessary pre-processing, e.g., transformed them into SSA form.
Because of this,
from Columns 3 and 9 (i.e., $\sharp$Tuples), one can see that we considered
more tuples in several benchmarks (e.g., Full AES (4)~\cite{CPRR13}, Full MAC-Keccak~\cite{BBDFGS15},
Sbox~\cite{RP10}, Key schedule~\cite{RP10}, Sbox (4)~\cite{CPRR13}, Sbox (5)~\cite{CPRR13}) than \cite{BBDFGS15},
namely, we considered more observable variables than \cite{BBDFGS15}.

%
%

From the experimental results, 
we can observe that there are two benchmarks (i.e., Sbox~\cite{RP10} under second-order masking and Sbox~\cite{SP06} under fourth-order masking) which have potential leaky observable sets,  and  Algorithm~\ref{alg:DivHOMaRer} produces the same number 
as \cite{BBDFGS15}.
This demonstrates that Algorithm~\ref{alg:DivHOMaRer} is at least as precise as the one in~\cite{BBDFGS15}.
We will report in Section~\ref{sec:exmc} the results of resolving these potential leaky observable sets using our model-counting and pattern
matching methods.


From Columns 4 and 10 (i.e., $\sharp$Sets),
one can observe that the number of observable sets actually verified by Algorithm~\ref{alg:DivHOMaRer}
is less than the one in \cite{BBDFGS15} on all the common benchmarks (despite there are  more observable variables to be considered
in several benchmarks). The differences are noticeable on several benchmarks (e.g., Full AES (4), Full Keccak,
Sbox, Multiplication, Sbox~\cite{RP10,CPRR13} and Key schedule).
Reducing the number of verified observable sets
allows us to verify 5th-order Sbox (4)~\cite{CPRR13} which has not been done in \cite{BBDFGS15}.
Furthermore, from Columns 7 and 11 (i.e., Time),
we observe that Algorithm~\ref{alg:DivHOMaRer} is faster than  \cite{BBDFGS15} on almost all the benchmarks, and the improvement
is significant on larger benchmarks (e.g., 110X, 64X and 31X speed-up for Key schedule, Full AES (4) and 4th-order Sbox~\cite{SP06}).
These results demonstrate that our type inference algorithm is superior to 
the one in~\cite{BBDFGS15}.
Furthermore, the algorithm presented in~\cite{BBDFGS15} has an issue which may miss the verification
of some observable sets. (We have informed some authors of~\cite{BBDFGS15}.)



\medskip
\noindent{\bf Results on new benchmarks.}
All the programs marked as N/A are new benchmarks. 
We note that B2A~\cite{Coron17} in Common Lisp has been \emph{semi-automatically} verified under the ISW model by Coron~\cite{Coron18},
the AES implementation~\cite{CPRR13} including Sbox (4) has been \emph{semi-automatically} proved under the $d$-NI model~\cite{BBDFGSZ16}.
Some of the first-order A2B and B2A (except A2B~\cite{CGTV15}) 
have been verified in~\cite{GXZSC19}.
All the other higher-order benchmarks have not been verified by computer-aided tools.  

From Table~\ref{tab:higherorder}, we can observe
that almost all benchmarks 
can be proved secure using our type inference algorithms in a few seconds. The exceptions include
B2A~\cite{Goubin01}, A2B ~\cite{Goubin01}, A2B~\cite{CGTV15} and B2A~\cite{Coron17} which respectively have 1, 37, 45 and 1 potential leaky observable set(s). We shall see in Section~\ref{sec:exmc} that these potential leaky observable sets are actually spurious using  model-counting.
To our knowledge, it is the first time that these higher-order programs are \emph{automatically} proved secure by computer-aided tools. Recall that A2B and B2A are two kinds of conversion algorithms between arithmetic and Boolean masking.
Our tool could be used to verify masked implementations of cryptographic algorithms that use A2B and/or B2A conversion algorithms.

\medskip
\noindent{\bf Usage of the transformations $\simply_{\Dom}$ and $\simply_{\Col}$.}
Columns 5 and 6 show the number of sets whose verification involves the transformations $\simply_{\Dom}$ and $\simply_{\Col}$, respectively.
We can see that $\simply_{\Dom}$ is heavily used,
while $\simply_{\Col}$ is used only in one benchmark (i.e., second-order and fourth-order Sbox~\cite{SP06}),
which allows to prove lots of observable sets (e.g., 256 on second-order Sbox~\cite{SP06}) without invoking model-counting.
Moreover, $\simply_{\Dom}$ and $\simply_{\Col}$ and simplify the expressions of the 98,176 potential leaky observable sets for Sbox~\cite{SP06},
so that pattern matching and model-counting methods can be easily applied.
Remark that statistics of the transformation $\simply_{\Alg}$ is not reported,
as its complexity is of constant-time and is negligible.

\subsection{Comparison of Model-Counting Methods} \label{sec:exmc}
One important component of our approach is the model-counting method on which we rely to resolve potential leaky observable sets.
As mentioned in Section~\ref{sec:intro}, we consider two baseline algorithms (based on SMT encoding and brute-force methods) and a novel GPU-accelerated parallel algorithm. For the sake of evaluation, we carry out experiments only on programs that have potential leaky observable sets reported by our type inference algorithm (cf.\ Result in Table~\ref{tab:higherorder}).
We also implemented two programs for computing $k^3$ and $k^{254}$, which contain
$1$ private input variable, $3$ and $5$ random input variables, respectively.
These programs are taken from the first-order secure exponentiation~\cite{RP10} without the first RefreshMask function.

\begin{table}[t]
\centering
\caption{Comparison of three model-counting methods. O.T. denotes run out of time (three hours).}
\label{tab:firstorder}
\scalebox{0.85}{
\begin{tabular}{ccccccc}
\toprule
{\bf Description} & {\bf Order} & {\bf $\sharp$CNT} & {\bf Result} & {\bf SMT}              & {\bf BFEnum} & {\bf GPU}            \\ \rowcolor{gray!20}
 $k^{3}$~\cite{RP10}    & 1     & 2   & 2       &   96m      &  {\bf 0.2s}      &  0.43s                \\
 $k^{254}$~\cite{RP10}  & 1     & 4   & 4       &  O.T. & 30m  & {\bf 7.03s} \\   \rowcolor{gray!20}
 B2A~\cite{Goubin01}    & 1     & 1   & 0       & 17s         &   2s   &  {\bf 0.86s}              \\
 A2B~\cite{Goubin01}    & 1     & 37  & 0       &   O.T.        &   O.T.      &  {\bf 33.18s}      \\ \rowcolor{gray!20}
 A2B~\cite{CGTV15}     & 1     & 45  & 0        &   O.T.        &   O.T.      & {\bf 160m}       \\
 B2A~\cite{Coron17}     & 1     & 1   & 0       &   1m 35s        &   10m 59s      &  {\bf 3.17s}      \\  \rowcolor{gray!20}
 Sbox~\cite{RP10}       & 2     & 2   & 1(1)    &   O.T.        &    O.T.     &  {\bf 3,600s}          \\
 Sbox~\cite{SP06}       & 4     & 766 & 98176    &  O.T.        &    O.T.      &  {\bf 323s}                   \\  \bottomrule
\end{tabular}}
\end{table}

Table~\ref{tab:firstorder} shows the statistics of the three model-counting methods, with time limited to three hours per program.
Column 1 shows the reference and description of the program.
Column 2 shows the security order.
Column 3 ($\sharp$CNTs) shows the time of the model-counting method.
Column 4 shows the number of genuine leaky observable sets.
Columns 5--7 show the verification time (\emph{excluding} the time for type inference algorithm) of the SMT-based, (na\"{i}ve) brute-force and GPU-accelerated parallel methods, respectively.

The resolution shows that all potential leaky observable sets
of B2A~\cite{Goubin01}, A2B~\cite{Goubin01}, A2B~\cite{CGTV15} and B2A~\cite{Coron17} are spurious, while all potential leaky observable sets of Sbox~\cite{SP06}
are genuine. On program Sbox~\cite{RP10}, we resolved one of two potential leaky observable sets as
a genuine one in 1 hour, but the other set cannot be resolved in 2 hour, which is the only case which was unsuccessful in our experiments.

In detail, the GPU-accelerated parallel method significantly outperforms the other two methods on large programs.
In particular, the SMT-based and brute-force methods runs out of time on five and four programs, respectively. On the small program $k^3$,
the brute-force method is significantly faster than the SMT-based one, and is also faster than the GPU-accelerated one.
The latter is because that the GPU-accelerated method synthesizes a GPU program for each expression and
the involved I/O cost is remarkable in small programs. The GPU-accelerated algorithm provides two orders of magnitude improvements on the program $k^{254}$.
A2B~\cite{Goubin01} has been verified in~\cite{GXZSC19} based on the oracle provided by the authors.
However, it is not always the case that one can find such an oracle luckily.
It runs out of time if we use the SMT-based method or the brute-force method without the oracle,
while the GPU-accelerated method can verify this program in less than 1 minute.
As a conclusion, when model-counting is concerned, we recommend the GPU-accelerated algorithm. 

For the fourth-order Sbox~\cite{SP06} which is faulty, it only took 4 minutes to automatically resolve all the 98,176 potential leakage sets as genuine ones. Therefore, our tool is still faster than~\cite{BBDFGS15}
albeit it needs to invoke the model-counting method to resolve those
sets which cannot be determined by type inference.
It should be emphasized that this was not possible without the patter matching based method described in Section~\ref{sec:pm}. Indeed,
we estimate (based on the experiment) that each set takes approximately 0.5s, and in total they would require approximately 14 hours.
Instead, we identified 766 ($255\times 2+256$) patterns which can be used to handle all 98,176 potential leaky observable sets. As a result, only 766 times of model-counting are needed, which took less than 7 minutes, i.e., two orders of magnitude faster.

The 766 patterns are summarized as follows:
\begin{itemize}
	\item[(1)] $\{x_0, \ \text{Sbox}(k \oplus x_0) \oplus r, \ \text{Sbox}(k  \oplus i\oplus x_0) \oplus r\}$;
	\item[(2)] $\{\text{Sbox}(k) \oplus r, \ \text{Sbox}(k \oplus x_0) \oplus r, \ \text{Sbox}(k \oplus i\oplus x_0) \oplus r\}$;
	\item[(3)] $\{x_0, \ \text{Sbox}(k) \oplus r, \ \text{Sbox}(k \oplus x_0\oplus j) \oplus r\}$;
\end{itemize}
where $0< i\leq 255$ and $0\leq j\leq 255$, $x_0$ is a random variable, $k$ is a private input, and $r$ is a random variable which is introduced by our transformations. 
The family in Item (1) captures 65,280 observable sets, namely, 256 observable sets for each  $0< i\leq 255$,
the family in Item (2) captures 32,640 observable sets, namely, 128 observable sets for each  $0< i\leq 255$,
and the family in Item (3) captures 256 observable sets, 1 observable set for each  $0\leq j\leq 255$.

Barthe et al.~\cite{BBDFGS15} manually analyzed the 98,176 potential leaky observable sets which
are summarized by four families. These are similar to our automatically computed patterns except for the patterns in Item (3),
which is $\{x_0, \ y_0, \ \text{Sbox}(k \oplus x_0\oplus j) \oplus r\}$ with $y_0=\text{Sbox}(x_0)$ in~\cite{BBDFGS15} (note the third expression is adjusted for sake of presentation).
After manually analyzing source code of Sbox~\cite{SP06} under fourth-order masking, we confirm that
our pattern is correct while the pattern in~\cite{BBDFGS15} is not correct.
This demonstrates that it is hard to manually examine potential leaky observable sets.

\subsection{Comparison with  {\tt maskVerif}}\label{sec:comparsionmaskverif}
Our tool  \tool is designed to tackle arithmetic programs, but it is also interesting to evaluate its performance on Boolean programs, for which we   compare with the latest version of the open source tool {\tt maskVerif}~\cite{BBFG18},
which is limited to Boolean programs. To the best of our knowledge, {\tt maskVerif} is the only open source tool for verifying higher-order Boolean programs.
We experiment on the largest 6 Boolean programs (P12--P17) from~\cite{EWS14b} which are one-round versions of
the full 24-round MAC-Keccak~\cite{BBDFGS15},
together with randomly selected benchmarks from {\tt maskVerif}.

In our experiment, {\tt maskVerif}  reported ``stack overflow'' error on P12--P17. (We have reported this issue to the developers of {\tt maskVerif}.)
For the sake of experiments, we removed the last 5,000 assignments for each program when testing {\tt maskVerif} while our tool \tool is still tested on the whole programs P12--P17. (For the abridged version no ``stack overflow'' error was reported from  {\tt maskVerif}.)
We also revised DOM AND~\cite{GrossMK17} and DOM Keccak Sbox~\cite{GSM17} by introducing the following extra dummy variables and statements:
\begin{center}
$t_1=r_1\wedge x; \ t_2 = (\neg r_1)\wedge(\neg x); \ t_3 = t_1\wedge t_2; \ t_4 = t_2\wedge r_3; \ ... \ t_{18} = t_{16} \wedge r_{17}; \ t_{19} = t_{17} \wedge r_{18};$
\end{center}
where $r_1$--$r_{17}$ are fresh random variables,
and $x$ denotes a share of a private input variable.
Obviously, $t_3$--$t_{19}$ are always $0$.

\begin{table}[t]
	\centering
	\caption{Comparison with {\tt maskVerif}.  \label{tab:commaksverif}}
\scalebox{0.85}{
		\begin{tabular}{c|c|c|c}
			\toprule
			\multirow{2}{*}{{\bf Description}} & \multicolumn{2}{c|}{{\bf Time (s)}} & \multirow{2}{*}{{\bf Result}} \\ \cline{2-3}
			& {\tt maskVerif}~\cite{BBFG18}       & \tool        &                         \\   \midrule
			\multicolumn{4}{c}{\bf First-Order Masking}                                            \\    \rowcolor{gray!20}
			DOM AND~\cite{GrossMK17}                    & 0.01            & 0.01        &  0                       \\
			DOM Keccak Sbox~\cite{GSM17}           & 0.01            & 0.01        & 0                        \\     \rowcolor{gray!20}
			DOM AES Sbox~\cite{GrossMK17}               & {\bf 0.23}           & 4.52        & 0                        \\
			TI Fides-192 APN~\cite{BilginBKMW13}           & {\bf 86.61}           & 139.40       & 0                        \\  \rowcolor{gray!20}
			P12~\cite{EWS14b}           & 3,223           &  {\bf 2.9}     & 0                        \\
			P13~\cite{EWS14b}           & 3,257(1,234)          & {\bf 122}       &  4.8k                       \\    \rowcolor{gray!20}
			P14-P17~\cite{EWS14b}       & O.T.($\leq$12)           & {\bf 72-168}       & 1.6k-17.6k                        \\ \midrule

			\multicolumn{4}{c}{\bf Second-Order Masking}                                           \\ \midrule    \rowcolor{gray!20}
			DOM AND~\cite{GrossMK17}                    & 0.01            & 0.01        &  0                       \\
			DOM AND (Revised)~\cite{GrossMK17}                    & 16.82            & {\bf 0.89}        &  0                       \\   \rowcolor{gray!20}
			DOM Keccak Sbox~\cite{GSM17}           & {\bf 0.01}               & 0.05        &  0                       \\
			DOM Keccak Sbox(Revised)~\cite{GSM17}           & 16.62               & {\bf 1.93}        & 0                        \\   \rowcolor{gray!20}
			DOM AES Sbox~\cite{GrossMK18}              & {\bf 61.59}           & 7,385        &  0                       \\  \midrule

			\multicolumn{4}{c}{\bf  Third-Order Masking}                                            \\ \midrule  \rowcolor{gray!20}
			DOM AND~\cite{GrossMK17}                   & {\bf 0.01}            & 0.02        &0                         \\
			DOM AND(Revised)~\cite{GrossMK17}                   & 828.70            & {\bf 6.39}        &0                         \\  \rowcolor{gray!20}
			DOM Keccak Sbox~\cite{GSM17}           & {\bf 0.33}       & 1.26        &0                         \\
			DOM Keccak Sbox(Revised)~\cite{GSM17}           & 1,041.67           & {\bf 27.40}       & 0                        \\ \midrule

			\multicolumn{4}{c}{\bf  Fourth-Order Masking}                                           \\ \midrule  \rowcolor{gray!20}
			DOM AND~\cite{GrossMK17}                    & {\bf 0.13}            & 0.40         &  0                       \\
			DOM AND(Revised)~\cite{GrossMK17}                    & O.T.            & {\bf 77.40}         &  0                       \\  \rowcolor{gray!20}
			DOM Keccak Sbox~\cite{GSM17}           & {\bf 16.35}           & 78.13       & 0                        \\
			DOM Keccak Sbox(Revised)~\cite{GSM17}           & O.T.           &  {\bf 690.79}       & 0                        \\  \bottomrule
	\end{tabular}}
\end{table}

Table~\ref{tab:commaksverif} presents the results, with time being limited to two hours per program.
Column 1 gives the programs under comparison.
Columns 2-3 show the verification time of {\tt maskVerif} and our tool respectively.
Column 4 gives the number of leaky observable sets.
On the programs taken from~{\tt maskVerif}, {\tt maskVerif} performs better (up to 5$\times$) than \tool. 
We note that there is  one benchmark (second-order) DOM AES Sbox for which {\tt maskVerif} performs exceptionally well. 
The major reason is that an ad hoc rule is used therein but could not  be used in \tool because it is tailored for Boolean programs.
It is perhaps worth pointing out that we have identified some bugs of {\tt maskVerif}. For instance, when {\tt maskVerif} verifies DOM AND (under second-order),
the leaky observable set  $\{(k \oplus r_0 \oplus r_1) \wedge r_2, r_0, r_1, r_2\}$ where $k$ is private and $r_0, r_1, r_2$ are random variables is considered secure. This bug 
has inadvertently reduced the verification time of {\tt maskVerif} as less sets of variables need to be examined. (We have reported this issue to the developers of {\tt maskVerif}.)

On the programs P12--P17 and revised programs DOM AND~\cite{GrossMK17} and DOM Keccak Sbox~\cite{GSM17},
\tool significantly outperforms {\tt maskVerif}.
Specifically, on the secure program P12, \tool takes 2.9s while {\tt maskVerif} takes 3,223s on the reduced version.
On the insecure program P13, \tool identified all the flaws of the program in 122s,
while {\tt maskVerif} identified 1,234 flaws of the reduced version in 3,257s.
On the insecure programs P14--P17, \tool identified all the flaws using at most 168s,
while {\tt maskVerif} runs out of time (2 hours) and identified at most 12 flaws.
On second-/third-order revised programs, \tool is 8.6--130$\times$ faster than {\tt maskVerif}.
On fourth-order revised programs, {\tt maskVerif} ran out of time.

In conclusion, even for Boolean programs, \tool demonstrates largely comparable performance on the benchmarks tested by {\tt maskVerif}, and indeed considerably better performance on the new benchmarks.

\section{Related Work} \label{sect:relatedwork}

In this section, we review related work on masking countermeasures in
general, as well as existing techniques on the analysis of masked programs and the
detection/mitigation of other types of side-channel leaks.

\smallskip
\noindent \textbf{Masking.}
Boolean and arithmetic masking schemes~\cite{Mess00,Goubin01,ISW03,BGK04,OMPR05,SP06,CB08,RP10,MPLPW11,PR13,RBNGV15,BelaidBPPTV17,FGPPS17,WangYS19}
have been widely investigated in the past two decades with differences in
adversary models, masking schemes, cryptographic algorithms and compactness.
Secure conversion algorithms between Boolean and arithmetic maskings
 have also been investigated~\cite{Goubin01,CGV14,CGTV15,Coron17,SchneiderPOG19,BCZ18,HT19}.
These countermeasures and conversion algorithms are often  designed manually for
specific cryptographic algorithms. In this context, the common problem is the lack of efficient and effective tools for automatically proving their
correctness~\cite{CPR07,CPRR13}. Our work aims to bridge this gap.

\smallskip
\noindent \textbf{Testing.}
The predominant approach addressing security of (masked) implementations of cryptographic algorithms is the empirical leakage assessment by
statistical significance tests or launching state-of-the-art side-channel attacks~\cite{Mangard04,GJJR11,OPB16,BCPST19,WPB19,CCRB18,MBBB16,OrenW16,OrenWW14,BGMVV11,DDGNFRV12,PPEBJ14,MBLM12,BHW12,KimK13}.
These approaches are valuable in  identifying flaws even without any knowledge of the leakage model,
but can neither prove their absence nor identify all flaws, due to the limitation in measurement setup and/or explored traces. This paper purses an alternative, formal verification based approach which is largely complementary to the work based on testing.

\smallskip
\noindent \textbf{Formal Verification.}
Formal verification approaches, which are able to prove the absence of side-channel leaks, have been proposed in the prior work~\cite{MOPT12,BRNI13,EWS14a,EWS14b,BBDFGS15,BBDFGSZ16,BMZ17,OMHE17,BYT17,BBFG18,Coron18,BGIKMW18,ZGSW18,BelaidGR18,WSW19,GXZSC19,GaoZSW19}.  However, as we have explained earlier, these existing formal verification
methods are limited in applicability (i.e., Boolean program, stronger leakage model or first-order security only) and accuracy
(i.e., false alarms).

Early work via type-based proof system refer to~\cite{MOPT12,BRNI13}, which checks if a computation
result is logically dependent of the secret data and, at the same time, logically independent of any
random variable used for masking the secret data. However, these incomplete approaches
only support verification of first-order arithmetic programs and may be even unsound under the ISW model, as pointed out in~\cite{EWS14b}.

To improve accuracy, Eldib et al. proposed model-counting based method~\cite{EWS14a,EWS14b}, which is both sound and complete under the ISW model. This method reduces the verification problem to a series of satisfiability problems encoding model-counting constraints, which is solved by leveraging SMT solvers. However, it is limited to the first-order Boolean programs only. Blot et al. extended this SMT-based method to verify higher-order programs~\cite{BYT17}. The SMT-encoding is exponential in the number of bits of random variables and the number of orders, hence is short of scalability and limited to Boolean programs only.
Our SMT-based method can be seen as an generalization of these methods. Nevertheless, our GPU-accelerated parallel algorithm significantly outperforms the SMT-based method.

To improve efficiency,  Barthe et al. introduced the notion of $d$-NI to characterize security of masked programs
and proposed a sound proof system to verify higher-order masked programs~\cite{BBDFGS15}. The $d$-NI notion was later extended to $d$-SNI~\cite{BBDFGSZ16} which enables compositional verification. 
However, these approaches are incomplete, namely, it may produce
spurious leaky observable sets. Furthermore, as mentioned in Section~\ref{ex:hom}, these approaches may miss the
verification of some observable sets. In this direction, Bisi et al.~\cite{BMZ17} proposed a technique for verifying higher-order masking, which was
limited to Boolean programs with linear operations only. Ouahma et al.\ generalized the approach of~\cite{BBDFGS15} to verify assembly-level
code~\cite{OMHE17}, but is incomplete and limited to first-order programs only. Coron~\cite{Coron18} proposed two complementary
semi-automatic approaches via elementary circuit transforms, and showed how to generate security proofs automatically, for simple circuits, but are also incomplete. 
Barthe et al. developed a unified framework {\tt maskVerif}~\cite{BBFG18} for both software and hardware implementations taking into account glitch and transitions into account,
but is limited to Boolean programs only and their tool missed the verification of some observable sets in our experiments.

As a matter of fact, the most efficient masked programs do not achieve $d$-SNI directly, as mentioned by Bloem et al.~\cite{BGIKMW18}.
Thus, Bloem et al. proposed a sound approach ~\cite{BGIKMW18} via Fourier analysis, which considers the Fourier expansion of the Boolean
functions and reduces the verification to checking whether certain coefficients of the Fourier expansion are zero or not~\cite{BGIKMW18}. They studied the security problem of Boolean programs/hardware circuits in the $d$-threshold probing model~\cite{ISW03} and its extension with glitches for any given $d$. The verification problem is solved by leveraging SAT solvers. However, they considered Boolean programs/hardware circuits only.
Furthermore, it was shown by Barthe et al.~\cite{BBFG18} that {\tt maskVerif} outperforms~\cite{BGIKMW18}.
Bela\"{\i}d et al. proposed another compositional verification approach in~\cite{BelaidGR18} to overcome the limitation of $d$-SNI~\cite{BBDFGSZ16}, but can only verify Boolean programs composed of ISW multiplication functions, sharewise addition functions and $d$-SNI refresh functions.

In our prior work~\cite{ZGSW18,GXZSC19,GaoZSW19}, we have proposed gradual refinement based approaches for verifying
masked Boolean and arithmetic programs respectively, which integrate the semantic type system and model-counting based methods hence bring the best of both worlds. This semantic type system was leveraged by Wang et al.~\cite{WSW19} to identify transition-based flaws.
All these approaches are limited to first-order security only. It is challenging to generalize these approaches to higher-order masked arithmetic
program, which is addressed by the current work.

Compared to the above existing formal verification approaches,
the current work studies formal verification of arithmetic programs
against $d$-threshold probing model for any given $d$.
Both our type system and model-counting based method significantly improve the applicability and efficiency.
Our pattern matching based method is novel and effective at reducing the cost of model-counting
and summarizing patterns of leaky observable sets which can be used for diagnosis and debugging.
Putting them together, our hybrid formal verification approach goes significantly beyond the state-of-the-art in terms of applicability, accuracy and efficiency.

\smallskip
\noindent \textbf{Automated mitigation of power side-channel flaws.}
Automated mitigation techniques have been proposed to repair power side-channel flaws~\cite{BRBSI11,MOPT12,BBP12,EW14,BBDFGSZ16,BYT17,WangS17,WSW19}.
For example, techniques proposed
in~\cite{BRBSI11,MOPT12,BBP12,BBDFGSZ16} rely on compiler-like pattern
matching, whereas the ones proposed in~\cite{EW14,BYT17,WangS17} use
inductive program synthesis and the one in~\cite{WSW19} constraints register allocation.
All these works either rely upon existing formal verification techniques, hence have similar limitations as described above,
or do not use formal verification techniques, thus, correctness can not be guaranteed.
It would be interesting to investigate whether our new approach can aid in the
mitigation of power side-channel flaws, effectively making countermeasures better, as done in~\cite{EW14,BYT17}.

\smallskip
\noindent \textbf{Other types of side channels.}
In addition to power side-channel attacks, there are other types of side-channel attacks against cryptographic programs, where the side channels can be in the form of, e.g., CPU time, faults and cache behaviors. Techniques for verification and mitigation of these types of side-channel attacks have been studied in the literature, such as~\cite{Kocher96,ABBDE16,PasareanuPM16,BangAPPB16,PhanBPMB17,BrennanSB18,ChenFD17,AntonopoulosGHK17,WGS18,WuW19} for timing side-channel attacks,
\cite{GGP07,KopfMO12,DFKMR13,BKMO14,ChuJM16,ChattopadhyayBRZ17,WangWLZW17,SungPW18,WGS18,GWW18,ChattopadhyayR18,BasuAWC20} for cache side-channel attacks and \cite{BS97,BDFGZ14,BHL17,EWW16,HouBZL19,BHL20} for fault attacks. Each type of side-channel has unique characteristics, which usually requires specific verification techniques, so these results are orthogonal to our work.

\section{Conclusion}
\label{sec:concl}
In this work, we have proposed a hybrid formal verification approach for higher-order masked arithmetic programs.
The approach comprise a sound proof system equipped with an efficient algorithm for type inference which significantly outperforms the approach~\cite{BBDFGS15} for arithmetic programs, as well as novel model-counting and pattern matching based methods for resolving potential leaky observable sets automatically that cannot be accomplished by the existing tools. Experimental results show that our approach is not only significantly faster
but also is applicable to  more cryptographic implementations that could not be proved secure automatically before.

Future work includes extending our methods to verifying programs with inherent branching and loops, and/or under other
leakage models such as $d$-NI/SNI, $d$-threshold probing model, as well as their extensions with glitches and transitions
as done in~~\cite{BBFG18,BGIKMW18}








\begin{thebibliography}{110}


\ifx \showCODEN    \undefined \def \showCODEN     #1{\unskip}     \fi
\ifx \showDOI      \undefined \def \showDOI       #1{#1}\fi
\ifx \showISBNx    \undefined \def \showISBNx     #1{\unskip}     \fi
\ifx \showISBNxiii \undefined \def \showISBNxiii  #1{\unskip}     \fi
\ifx \showISSN     \undefined \def \showISSN      #1{\unskip}     \fi
\ifx \showLCCN     \undefined \def \showLCCN      #1{\unskip}     \fi
\ifx \shownote     \undefined \def \shownote      #1{#1}          \fi
\ifx \showarticletitle \undefined \def \showarticletitle #1{#1}   \fi
\ifx \showURL      \undefined \def \showURL       {\relax}        \fi
\providecommand\bibfield[2]{#2}
\providecommand\bibinfo[2]{#2}
\providecommand\natexlab[1]{#1}
\providecommand\showeprint[2][]{arXiv:#2}

\bibitem[\protect\citeauthoryear{Agosta, Barenghi, and Pelosi}{Agosta
  et~al\mbox{.}}{2012}]%
        {BBP12}
\bibfield{author}{\bibinfo{person}{Giovanni Agosta},
  \bibinfo{person}{Alessandro Barenghi}, {and} \bibinfo{person}{Gerardo
  Pelosi}.} \bibinfo{year}{2012}\natexlab{}.
\newblock \showarticletitle{A code morphing methodology to automate power
  analysis countermeasures}. In \bibinfo{booktitle}{\emph{Proceedings of the
  49th Annual Design Automation Conference}}. \bibinfo{pages}{77--82}.
\newblock


\bibitem[\protect\citeauthoryear{Almeida, Barbosa, Barthe, Dupressoir, and
  Emmi}{Almeida et~al\mbox{.}}{2016}]%
        {ABBDE16}
\bibfield{author}{\bibinfo{person}{Jos{\'{e}}~Bacelar Almeida},
  \bibinfo{person}{Manuel Barbosa}, \bibinfo{person}{Gilles Barthe},
  \bibinfo{person}{Fran{\c{c}}ois Dupressoir}, {and} \bibinfo{person}{Michael
  Emmi}.} \bibinfo{year}{2016}\natexlab{}.
\newblock \showarticletitle{Verifying Constant-Time Implementations}. In
  \bibinfo{booktitle}{\emph{Proceedings of the 25th {USENIX} Security
  Symposium}}. \bibinfo{pages}{53--70}.
\newblock


\bibitem[\protect\citeauthoryear{Antonopoulos, Gazzillo, Hicks, Koskinen,
  Terauchi, and Wei}{Antonopoulos et~al\mbox{.}}{2017}]%
        {AntonopoulosGHK17}
\bibfield{author}{\bibinfo{person}{Timos Antonopoulos}, \bibinfo{person}{Paul
  Gazzillo}, \bibinfo{person}{Michael Hicks}, \bibinfo{person}{Eric Koskinen},
  \bibinfo{person}{Tachio Terauchi}, {and} \bibinfo{person}{Shiyi Wei}.}
  \bibinfo{year}{2017}\natexlab{}.
\newblock \showarticletitle{Decomposition instead of self-composition for
  proving the absence of timing channels}. In
  \bibinfo{booktitle}{\emph{Proceedings of the 38th {ACM} {SIGPLAN} Conference
  on Programming Language Design and Implementation}}.
  \bibinfo{pages}{362--375}.
\newblock


\bibitem[\protect\citeauthoryear{Arribas, Nikova, and Rijmen}{Arribas
  et~al\mbox{.}}{2018}]%
        {ANR17}
\bibfield{author}{\bibinfo{person}{Victor Arribas}, \bibinfo{person}{Svetla
  Nikova}, {and} \bibinfo{person}{Vincent Rijmen}.}
  \bibinfo{year}{2018}\natexlab{}.
\newblock \showarticletitle{VerMI: Verification Tool for Masked
  Implementations}. In \bibinfo{booktitle}{\emph{Proceedings of the 25th {IEEE}
  International Conference on Electronics, Circuits and Systems}}.
  \bibinfo{pages}{381--384}.
\newblock


\bibitem[\protect\citeauthoryear{Bang, Aydin, Phan, Pasareanu, and Bultan}{Bang
  et~al\mbox{.}}{2016}]%
        {BangAPPB16}
\bibfield{author}{\bibinfo{person}{Lucas Bang}, \bibinfo{person}{Abdulbaki
  Aydin}, \bibinfo{person}{Quoc{-}Sang Phan}, \bibinfo{person}{Corina~S.
  Pasareanu}, {and} \bibinfo{person}{Tevfik Bultan}.}
  \bibinfo{year}{2016}\natexlab{}.
\newblock \showarticletitle{String analysis for side channels with segmented
  oracles}. In \bibinfo{booktitle}{\emph{Proceedings of the 24th {ACM}
  {SIGSOFT} International Symposium on Foundations of Software Engineering}}.
  \bibinfo{pages}{193--204}.
\newblock


\bibitem[\protect\citeauthoryear{Barthe, Bela{\"{\i}}d, Dupressoir, Fouque,
  Gr{\'{e}}goire, and Strub}{Barthe et~al\mbox{.}}{2015}]%
        {BBDFGS15}
\bibfield{author}{\bibinfo{person}{Gilles Barthe}, \bibinfo{person}{Sonia
  Bela{\"{\i}}d}, \bibinfo{person}{Fran{\c{c}}ois Dupressoir},
  \bibinfo{person}{Pierre{-}Alain Fouque}, \bibinfo{person}{Benjamin
  Gr{\'{e}}goire}, {and} \bibinfo{person}{Pierre{-}Yves Strub}.}
  \bibinfo{year}{2015}\natexlab{}.
\newblock \showarticletitle{Verified Proofs of Higher-Order Masking}. In
  \bibinfo{booktitle}{\emph{Proceedings of the 34th Annual International
  Conference on the Theory and Applications of Cryptographic Techniques}}.
  \bibinfo{pages}{457--485}.
\newblock


\bibitem[\protect\citeauthoryear{Barthe, Bela{\"{\i}}d, Dupressoir, Fouque,
  Gr{\'{e}}goire, Strub, and Zucchini}{Barthe et~al\mbox{.}}{2016}]%
        {BBDFGSZ16}
\bibfield{author}{\bibinfo{person}{Gilles Barthe}, \bibinfo{person}{Sonia
  Bela{\"{\i}}d}, \bibinfo{person}{Fran{\c{c}}ois Dupressoir},
  \bibinfo{person}{Pierre{-}Alain Fouque}, \bibinfo{person}{Benjamin
  Gr{\'{e}}goire}, \bibinfo{person}{Pierre{-}Yves Strub}, {and}
  \bibinfo{person}{R{\'{e}}becca Zucchini}.} \bibinfo{year}{2016}\natexlab{}.
\newblock \showarticletitle{Strong Non-Interference and Type-Directed
  Higher-Order Masking}. In \bibinfo{booktitle}{\emph{Proceedings of the 2016
  {ACM} {SIGSAC} Conference on Computer and Communications Security}}.
  \bibinfo{pages}{116--129}.
\newblock


\bibitem[\protect\citeauthoryear{Barthe, Bela{\"{\i}}d, Fouque, and
  Gr{\'{e}}goire}{Barthe et~al\mbox{.}}{2019}]%
        {BBFG18}
\bibfield{author}{\bibinfo{person}{Gilles Barthe}, \bibinfo{person}{Sonia
  Bela{\"{\i}}d}, \bibinfo{person}{Pierre{-}Alain Fouque}, {and}
  \bibinfo{person}{Benjamin Gr{\'{e}}goire}.} \bibinfo{year}{2019}\natexlab{}.
\newblock \showarticletitle{maskVerif: Automated Verification of Higher-Order
  Masking in Presence of Physical Defaults}. In
  \bibinfo{booktitle}{\emph{Proceedings of the 24th European Symposium on
  Research in Computer Security}}. \bibinfo{pages}{300--318}.
\newblock


\bibitem[\protect\citeauthoryear{Barthe, Dupressoir, Faust, Gr{\'{e}}goire,
  Standaert, and Strub}{Barthe et~al\mbox{.}}{2017}]%
        {BDFGSS17}
\bibfield{author}{\bibinfo{person}{Gilles Barthe},
  \bibinfo{person}{Fran{\c{c}}ois Dupressoir}, \bibinfo{person}{Sebastian
  Faust}, \bibinfo{person}{Benjamin Gr{\'{e}}goire},
  \bibinfo{person}{Fran{\c{c}}ois{-}Xavier Standaert}, {and}
  \bibinfo{person}{Pierre{-}Yves Strub}.} \bibinfo{year}{2017}\natexlab{}.
\newblock \showarticletitle{Parallel Implementations of Masking Schemes and the
  Bounded Moment Leakage Model}. In \bibinfo{booktitle}{\emph{Proceedings of
  the 36th Annual International Conference on the Theory and Applications of
  Cryptographic Techniques}}. \bibinfo{pages}{535--566}.
\newblock


\bibitem[\protect\citeauthoryear{Barthe, Dupressoir, Fouque, Gr{\'{e}}goire,
  and Zapalowicz}{Barthe et~al\mbox{.}}{2014a}]%
        {BDFGZ14}
\bibfield{author}{\bibinfo{person}{Gilles Barthe},
  \bibinfo{person}{Fran{\c{c}}ois Dupressoir}, \bibinfo{person}{Pierre{-}Alain
  Fouque}, \bibinfo{person}{Benjamin Gr{\'{e}}goire}, {and}
  \bibinfo{person}{Jean{-}Christophe Zapalowicz}.}
  \bibinfo{year}{2014}\natexlab{a}.
\newblock \showarticletitle{Synthesis of Fault Attacks on Cryptographic
  Implementations}. In \bibinfo{booktitle}{\emph{Proceedings of the 2014 {ACM}
  {SIGSAC} Conference on Computer and Communications Security}}.
  \bibinfo{pages}{1016--1027}.
\newblock


\bibitem[\protect\citeauthoryear{Barthe, K{\"{o}}pf, Mauborgne, and
  Ochoa}{Barthe et~al\mbox{.}}{2014b}]%
        {BKMO14}
\bibfield{author}{\bibinfo{person}{Gilles Barthe}, \bibinfo{person}{Boris
  K{\"{o}}pf}, \bibinfo{person}{Laurent Mauborgne}, {and}
  \bibinfo{person}{Mart{\'{\i}}n Ochoa}.} \bibinfo{year}{2014}\natexlab{b}.
\newblock \showarticletitle{Leakage Resilience against Concurrent Cache
  Attacks}. In \bibinfo{booktitle}{\emph{Proceedings of the 3rd International
  Conference on Principles of Security and Trust}}. \bibinfo{pages}{140--158}.
\newblock


\bibitem[\protect\citeauthoryear{Basu, Aggarwal, Wang, and Chattopadhyay}{Basu
  et~al\mbox{.}}{2020}]%
        {BasuAWC20}
\bibfield{author}{\bibinfo{person}{Tiyash Basu}, \bibinfo{person}{Kartik
  Aggarwal}, \bibinfo{person}{Chundong Wang}, {and} \bibinfo{person}{Sudipta
  Chattopadhyay}.} \bibinfo{year}{2020}\natexlab{}.
\newblock \showarticletitle{An exploration of effective fuzzing for
  side-channel cache leakage}.
\newblock \bibinfo{journal}{\emph{Software Testing, Verification \&
  Reliability}} \bibinfo{volume}{30}, \bibinfo{number}{1}
  (\bibinfo{year}{2020}).
\newblock


\bibitem[\protect\citeauthoryear{Batina, Chmielewski, Papachristodoulou,
  Schwabe, and Tunstall}{Batina et~al\mbox{.}}{2019}]%
        {BCPST19}
\bibfield{author}{\bibinfo{person}{Lejla Batina}, \bibinfo{person}{Lukasz
  Chmielewski}, \bibinfo{person}{Louiza Papachristodoulou},
  \bibinfo{person}{Peter Schwabe}, {and} \bibinfo{person}{Michael Tunstall}.}
  \bibinfo{year}{2019}\natexlab{}.
\newblock \showarticletitle{Online Template Attacks}.
\newblock \bibinfo{journal}{\emph{Journal of Cryptographic Engineering}}
  \bibinfo{volume}{9}, \bibinfo{number}{1} (\bibinfo{year}{2019}),
  \bibinfo{pages}{21--36}.
\newblock


\bibitem[\protect\citeauthoryear{Batina, Hogenboom, and van Woudenberg}{Batina
  et~al\mbox{.}}{2012}]%
        {BHW12}
\bibfield{author}{\bibinfo{person}{Lejla Batina}, \bibinfo{person}{Jip
  Hogenboom}, {and} \bibinfo{person}{Jasper G.~J. van Woudenberg}.}
  \bibinfo{year}{2012}\natexlab{}.
\newblock \showarticletitle{Getting More from {PCA:} First Results of Using
  Principal Component Analysis for Extensive Power Analysis}. In
  \bibinfo{booktitle}{\emph{Proceedings of the {RSA} Conference Cryptographers'
  Track}}. \bibinfo{pages}{383--397}.
\newblock


\bibitem[\protect\citeauthoryear{Bayrak, Regazzoni, Brisk, Standaert, and
  Ienne}{Bayrak et~al\mbox{.}}{2011}]%
        {BRBSI11}
\bibfield{author}{\bibinfo{person}{Ali~Galip Bayrak},
  \bibinfo{person}{Francesco Regazzoni}, \bibinfo{person}{Philip Brisk},
  \bibinfo{person}{Fran{\c{c}}ois{-}Xavier Standaert}, {and}
  \bibinfo{person}{Paolo Ienne}.} \bibinfo{year}{2011}\natexlab{}.
\newblock \showarticletitle{A first step towards automatic application of power
  analysis countermeasures}. In \bibinfo{booktitle}{\emph{ACM/IEEE Design
  Automation Conference}}. \bibinfo{pages}{230--235}.
\newblock


\bibitem[\protect\citeauthoryear{Bayrak, Regazzoni, Novo, and Ienne}{Bayrak
  et~al\mbox{.}}{2013}]%
        {BRNI13}
\bibfield{author}{\bibinfo{person}{Ali~Galip Bayrak},
  \bibinfo{person}{Francesco Regazzoni}, \bibinfo{person}{David Novo}, {and}
  \bibinfo{person}{Paolo Ienne}.} \bibinfo{year}{2013}\natexlab{}.
\newblock \showarticletitle{Sleuth: Automated Verification of Software Power
  Analysis Countermeasures}. In \bibinfo{booktitle}{\emph{Proceedings of the
  15th International Workshop on Cryptographic Hardware and Embedded Systems}}.
  \bibinfo{pages}{293--310}.
\newblock


\bibitem[\protect\citeauthoryear{Bela{\"{\i}}d, Benhamouda, Passel{\`{e}}gue,
  Prouff, Thillard, and Vergnaud}{Bela{\"{\i}}d et~al\mbox{.}}{2017}]%
        {BelaidBPPTV17}
\bibfield{author}{\bibinfo{person}{Sonia Bela{\"{\i}}d},
  \bibinfo{person}{Fabrice Benhamouda}, \bibinfo{person}{Alain
  Passel{\`{e}}gue}, \bibinfo{person}{Emmanuel Prouff}, \bibinfo{person}{Adrian
  Thillard}, {and} \bibinfo{person}{Damien Vergnaud}.}
  \bibinfo{year}{2017}\natexlab{}.
\newblock \showarticletitle{Private Multiplication over Finite Fields}. In
  \bibinfo{booktitle}{\emph{Proceedings of the 37th Annual International
  Cryptology Conference}}. \bibinfo{pages}{397--426}.
\newblock


\bibitem[\protect\citeauthoryear{Bela{\"{\i}}d, Goudarzi, and
  Rivain}{Bela{\"{\i}}d et~al\mbox{.}}{2018}]%
        {BelaidGR18}
\bibfield{author}{\bibinfo{person}{Sonia Bela{\"{\i}}d},
  \bibinfo{person}{Dahmun Goudarzi}, {and} \bibinfo{person}{Matthieu Rivain}.}
  \bibinfo{year}{2018}\natexlab{}.
\newblock \showarticletitle{Tight Private Circuits: Achieving Probing Security
  with the Least Refreshing}. In \bibinfo{booktitle}{\emph{Proceedings of the
  24th International Conference on the Theory and Application of Cryptology and
  Information Security}}. \bibinfo{pages}{343--372}.
\newblock


\bibitem[\protect\citeauthoryear{Bettale, Coron, and Zeitoun}{Bettale
  et~al\mbox{.}}{2018}]%
        {BCZ18}
\bibfield{author}{\bibinfo{person}{Luk Bettale},
  \bibinfo{person}{Jean{-}S{\'{e}}bastien Coron}, {and} \bibinfo{person}{Rina
  Zeitoun}.} \bibinfo{year}{2018}\natexlab{}.
\newblock \showarticletitle{Improved High-Order Conversion From Boolean to
  Arithmetic Masking}.
\newblock \bibinfo{journal}{\emph{IACR Transactions on Cryptographic Hardware
  and Embedded Systems}} \bibinfo{number}{2} (\bibinfo{year}{2018}),
  \bibinfo{pages}{22--45}.
\newblock


\bibitem[\protect\citeauthoryear{Biham and Shamir}{Biham and Shamir}{1997}]%
        {BS97}
\bibfield{author}{\bibinfo{person}{Eli Biham} {and} \bibinfo{person}{Adi
  Shamir}.} \bibinfo{year}{1997}\natexlab{}.
\newblock \showarticletitle{Differential Fault Analysis of Secret Key
  Cryptosystems}. In \bibinfo{booktitle}{\emph{CRYPTO}}.
  \bibinfo{pages}{513--525}.
\newblock


\bibitem[\protect\citeauthoryear{Bilgin, Bogdanov, Knezevic, Mendel, and
  Wang}{Bilgin et~al\mbox{.}}{2013}]%
        {BilginBKMW13}
\bibfield{author}{\bibinfo{person}{Beg{\"{u}}l Bilgin}, \bibinfo{person}{Andrey
  Bogdanov}, \bibinfo{person}{Miroslav Knezevic}, \bibinfo{person}{Florian
  Mendel}, {and} \bibinfo{person}{Qingju Wang}.}
  \bibinfo{year}{2013}\natexlab{}.
\newblock \showarticletitle{Fides: Lightweight Authenticated Cipher with
  Side-Channel Resistance for Constrained Hardware}. In
  \bibinfo{booktitle}{\emph{Proceedings of the 15th International Workshop on
  Cryptographic Hardware and Embedded Systems}}. \bibinfo{pages}{142--158}.
\newblock


\bibitem[\protect\citeauthoryear{Bisi, Melzani, and Zaccaria}{Bisi
  et~al\mbox{.}}{2017}]%
        {BMZ17}
\bibfield{author}{\bibinfo{person}{Elia Bisi}, \bibinfo{person}{Filippo
  Melzani}, {and} \bibinfo{person}{Vittorio Zaccaria}.}
  \bibinfo{year}{2017}\natexlab{}.
\newblock \showarticletitle{Symbolic Analysis of Higher-Order Side Channel
  Countermeasures}.
\newblock \bibinfo{journal}{\emph{{IEEE Transactions on Computers}}}
  \bibinfo{volume}{66}, \bibinfo{number}{6} (\bibinfo{year}{2017}),
  \bibinfo{pages}{1099--1105}.
\newblock


\bibitem[\protect\citeauthoryear{Bloem, Gro{\ss}, Iusupov, K{\"{o}}nighofer,
  Mangard, and Winter}{Bloem et~al\mbox{.}}{2018}]%
        {BGIKMW18}
\bibfield{author}{\bibinfo{person}{Roderick Bloem}, \bibinfo{person}{Hannes
  Gro{\ss}}, \bibinfo{person}{Rinat Iusupov}, \bibinfo{person}{Bettina
  K{\"{o}}nighofer}, \bibinfo{person}{Stefan Mangard}, {and}
  \bibinfo{person}{Johannes Winter}.} \bibinfo{year}{2018}\natexlab{}.
\newblock \showarticletitle{Formal Verification of Masked Hardware
  Implementations in the Presence of Glitches}. In
  \bibinfo{booktitle}{\emph{Proceedings of the 37th Annual International
  Conference on the Theory and Applications of Cryptographic Techniques}}.
  \bibinfo{pages}{321--353}.
\newblock


\bibitem[\protect\citeauthoryear{Bl{\"o}mer, Guajardo, and Krummel}{Bl{\"o}mer
  et~al\mbox{.}}{2004}]%
        {BGK04}
\bibfield{author}{\bibinfo{person}{Johannes Bl{\"o}mer}, \bibinfo{person}{Jorge
  Guajardo}, {and} \bibinfo{person}{Volker Krummel}.}
  \bibinfo{year}{2004}\natexlab{}.
\newblock \showarticletitle{Provably secure masking of AES}. In
  \bibinfo{booktitle}{\emph{Proceedings of the International Workshop on
  Selected Areas in Cryptography}}. \bibinfo{pages}{69--83}.
\newblock


\bibitem[\protect\citeauthoryear{Blot, Yamamoto, and Terauchi}{Blot
  et~al\mbox{.}}{2017}]%
        {BYT17}
\bibfield{author}{\bibinfo{person}{Arthur Blot}, \bibinfo{person}{Masaki
  Yamamoto}, {and} \bibinfo{person}{Tachio Terauchi}.}
  \bibinfo{year}{2017}\natexlab{}.
\newblock \showarticletitle{Compositional Synthesis of Leakage Resilient
  Programs}. In \bibinfo{booktitle}{\emph{Proceedings of the 6th International
  Conference on Principles of Security and Trust}}. \bibinfo{pages}{277--297}.
\newblock


\bibitem[\protect\citeauthoryear{Breier, Hou, and Liu}{Breier
  et~al\mbox{.}}{2018}]%
        {BHL17}
\bibfield{author}{\bibinfo{person}{Jakub Breier}, \bibinfo{person}{Xiaolu Hou},
  {and} \bibinfo{person}{Yang Liu}.} \bibinfo{year}{2018}\natexlab{}.
\newblock \showarticletitle{Fault Attacks Made Easy: Differential Fault
  Analysis Automation on Assembly Code}.
\newblock \bibinfo{journal}{\emph{{IACR} Transactions on Cryptographic Hardware
  and Embedded Systems}} \bibinfo{volume}{2018}, \bibinfo{number}{2}
  (\bibinfo{year}{2018}), \bibinfo{pages}{96--122}.
\newblock


\bibitem[\protect\citeauthoryear{{Breier}, {Hou}, and {Liu}}{{Breier}
  et~al\mbox{.}}{2019}]%
        {BHL20}
\bibfield{author}{\bibinfo{person}{J. {Breier}}, \bibinfo{person}{X. {Hou}},
  {and} \bibinfo{person}{Y. {Liu}}.} \bibinfo{year}{2019}\natexlab{}.
\newblock \showarticletitle{On Evaluating Fault Resilient Encoding Schemes in
  Software}.
\newblock \bibinfo{journal}{\emph{IEEE Transactions on Dependable and Secure
  Computing}} (\bibinfo{year}{2019}), \bibinfo{pages}{1--14}.
\newblock


\bibitem[\protect\citeauthoryear{Brennan, Saha, Bultan, and Pasareanu}{Brennan
  et~al\mbox{.}}{2018}]%
        {BrennanSB18}
\bibfield{author}{\bibinfo{person}{Tegan Brennan}, \bibinfo{person}{Seemanta
  Saha}, \bibinfo{person}{Tevfik Bultan}, {and} \bibinfo{person}{Corina~S.
  Pasareanu}.} \bibinfo{year}{2018}\natexlab{}.
\newblock \showarticletitle{Symbolic path cost analysis for side-channel
  detection}. In \bibinfo{booktitle}{\emph{Proceedings of the 27th {ACM}
  {SIGSOFT} International Symposium on Software Testing and Analysis}}.
  \bibinfo{pages}{27--37}.
\newblock


\bibitem[\protect\citeauthoryear{Canright and Batina}{Canright and
  Batina}{2008}]%
        {CB08}
\bibfield{author}{\bibinfo{person}{D. Canright} {and} \bibinfo{person}{Lejla
  Batina}.} \bibinfo{year}{2008}\natexlab{}.
\newblock \showarticletitle{A Very Compact "Perfectly Masked" S-Box for {AES}}.
  In \bibinfo{booktitle}{\emph{Proceedings of the 6th International Conference
  on Applied Cryptography and Network Security}}. \bibinfo{pages}{446--459}.
\newblock


\bibitem[\protect\citeauthoryear{Carlet, Goubin, Prouff, Quisquater, and
  Rivain}{Carlet et~al\mbox{.}}{2012}]%
        {CGPQR12}
\bibfield{author}{\bibinfo{person}{Claude Carlet}, \bibinfo{person}{Louis
  Goubin}, \bibinfo{person}{Emmanuel Prouff}, \bibinfo{person}{Micha{\"{e}}l
  Quisquater}, {and} \bibinfo{person}{Matthieu Rivain}.}
  \bibinfo{year}{2012}\natexlab{}.
\newblock \showarticletitle{Higher-Order Masking Schemes for S-Boxes}. In
  \bibinfo{booktitle}{\emph{Proceedings of the 19th International Workshop Fast
  Software Encryption}}. \bibinfo{pages}{366--384}.
\newblock


\bibitem[\protect\citeauthoryear{Chari, Rao, and Rohatgi}{Chari
  et~al\mbox{.}}{2002}]%
        {ChariRR02}
\bibfield{author}{\bibinfo{person}{Suresh Chari}, \bibinfo{person}{Josyula~R.
  Rao}, {and} \bibinfo{person}{Pankaj Rohatgi}.}
  \bibinfo{year}{2002}\natexlab{}.
\newblock \showarticletitle{Template Attacks}. In
  \bibinfo{booktitle}{\emph{roceedings of the 4th International Workshop on
  Cryptographic Hardware and Embedded Systems}}. \bibinfo{pages}{13--28}.
\newblock


\bibitem[\protect\citeauthoryear{Chattopadhyay, Beck, Rezine, and
  Zeller}{Chattopadhyay et~al\mbox{.}}{2019}]%
        {ChattopadhyayBRZ17}
\bibfield{author}{\bibinfo{person}{Sudipta Chattopadhyay},
  \bibinfo{person}{Moritz Beck}, \bibinfo{person}{Ahmed Rezine}, {and}
  \bibinfo{person}{Andreas Zeller}.} \bibinfo{year}{2019}\natexlab{}.
\newblock \showarticletitle{Quantifying the Information Leakage in Cache
  Attacks via Symbolic Execution}.
\newblock \bibinfo{journal}{\emph{{ACM} Transactions on Embedded Computing
  Systems}} \bibinfo{volume}{18}, \bibinfo{number}{1} (\bibinfo{year}{2019}),
  \bibinfo{pages}{7:1--7:27}.
\newblock


\bibitem[\protect\citeauthoryear{Chattopadhyay and Roychoudhury}{Chattopadhyay
  and Roychoudhury}{2018}]%
        {ChattopadhyayR18}
\bibfield{author}{\bibinfo{person}{Sudipta Chattopadhyay} {and}
  \bibinfo{person}{Abhik Roychoudhury}.} \bibinfo{year}{2018}\natexlab{}.
\newblock \showarticletitle{Symbolic Verification of Cache Side-Channel
  Freedom}.
\newblock \bibinfo{journal}{\emph{{IEEE} Trans. on {CAD} of Integrated Circuits
  and Systems}} \bibinfo{volume}{37}, \bibinfo{number}{11}
  (\bibinfo{year}{2018}), \bibinfo{pages}{2812--2823}.
\newblock


\bibitem[\protect\citeauthoryear{Chaves, Chmielewski, Regazzoni, and
  Batina}{Chaves et~al\mbox{.}}{2018}]%
        {CCRB18}
\bibfield{author}{\bibinfo{person}{Ricardo Chaves}, \bibinfo{person}{Lukasz
  Chmielewski}, \bibinfo{person}{Francesco Regazzoni}, {and}
  \bibinfo{person}{Lejla Batina}.} \bibinfo{year}{2018}\natexlab{}.
\newblock \showarticletitle{{SCA}-Resistance for {AES:} How Cheap Can We Go?}.
  In \bibinfo{booktitle}{\emph{Proceedings of the 10th International Conference
  on Cryptology in Africa}}. \bibinfo{pages}{107--123}.
\newblock


\bibitem[\protect\citeauthoryear{Chen, Feng, and Dillig}{Chen
  et~al\mbox{.}}{2017}]%
        {ChenFD17}
\bibfield{author}{\bibinfo{person}{Jia Chen}, \bibinfo{person}{Yu Feng}, {and}
  \bibinfo{person}{Isil Dillig}.} \bibinfo{year}{2017}\natexlab{}.
\newblock \showarticletitle{Precise Detection of Side-Channel Vulnerabilities
  using Quantitative Cartesian Hoare Logic}. In
  \bibinfo{booktitle}{\emph{Proceedings of the 2017 {ACM} {SIGSAC} Conference
  on Computer and Communications Security}}. \bibinfo{pages}{875--890}.
\newblock


\bibitem[\protect\citeauthoryear{Chu, Jaffar, and Maghareh}{Chu
  et~al\mbox{.}}{2016}]%
        {ChuJM16}
\bibfield{author}{\bibinfo{person}{Duc{-}Hiep Chu}, \bibinfo{person}{Joxan
  Jaffar}, {and} \bibinfo{person}{Rasool Maghareh}.}
  \bibinfo{year}{2016}\natexlab{}.
\newblock \showarticletitle{Precise Cache Timing Analysis via Symbolic
  Execution}. In \bibinfo{booktitle}{\emph{Proceedings of the IEEE Symposium on
  Real-Time and Embedded Technology and Applications}}.
  \bibinfo{pages}{293--304}.
\newblock


\bibitem[\protect\citeauthoryear{Contini, Rivest, Robshaw, and Yin}{Contini
  et~al\mbox{.}}{1999}]%
        {CRRY99}
\bibfield{author}{\bibinfo{person}{Scott Contini}, \bibinfo{person}{Ronald~L.
  Rivest}, \bibinfo{person}{Matthew J.~B. Robshaw}, {and}
  \bibinfo{person}{Yiqun~Lisa Yin}.} \bibinfo{year}{1999}\natexlab{}.
\newblock \showarticletitle{Improved Analysis of Some Simplified Variants of
  {RC6}}. In \bibinfo{booktitle}{\emph{Proceedings of the 6th International
  Workshop on Fast Software Encryption}}. \bibinfo{pages}{1--15}.
\newblock


\bibitem[\protect\citeauthoryear{Coron}{Coron}{2017}]%
        {Coron17}
\bibfield{author}{\bibinfo{person}{Jean{-}S{\'{e}}bastien Coron}.}
  \bibinfo{year}{2017}\natexlab{}.
\newblock \showarticletitle{High-Order Conversion from Boolean to Arithmetic
  Masking}. In \bibinfo{booktitle}{\emph{Proceedings of the 19th International
  Conference on Cryptographic Hardware and Embedded Systems}}.
  \bibinfo{pages}{93--114}.
\newblock


\bibitem[\protect\citeauthoryear{Coron}{Coron}{2018}]%
        {Coron18}
\bibfield{author}{\bibinfo{person}{Jean{-}S{\'{e}}bastien Coron}.}
  \bibinfo{year}{2018}\natexlab{}.
\newblock \showarticletitle{Formal Verification of Side-Channel Countermeasures
  via Elementary Circuit Transformations}. In
  \bibinfo{booktitle}{\emph{Proceedings of the 16th International Conference on
  Applied Cryptography and Network Security}}. \bibinfo{pages}{65--82}.
\newblock


\bibitem[\protect\citeauthoryear{Coron, Giraud, Prouff, Renner, Rivain, and
  Vadnala}{Coron et~al\mbox{.}}{2012}]%
        {CGPRRV12}
\bibfield{author}{\bibinfo{person}{Jean{-}S{\'{e}}bastien Coron},
  \bibinfo{person}{Christophe Giraud}, \bibinfo{person}{Emmanuel Prouff},
  \bibinfo{person}{Soline Renner}, \bibinfo{person}{Matthieu Rivain}, {and}
  \bibinfo{person}{Praveen~Kumar Vadnala}.} \bibinfo{year}{2012}\natexlab{}.
\newblock \showarticletitle{Conversion of Security Proofs from One Leakage
  Model to Another: {A} New Issue}. In \bibinfo{booktitle}{\emph{Proceedings of
  the 3rd International Workshop on Constructive Side-Channel Analysis and
  Secure Design}}. \bibinfo{pages}{69--81}.
\newblock


\bibitem[\protect\citeauthoryear{Coron, Gro{\ss}sch{\"{a}}dl, Tibouchi, and
  Vadnala}{Coron et~al\mbox{.}}{2015}]%
        {CGTV15}
\bibfield{author}{\bibinfo{person}{Jean{-}S{\'{e}}bastien Coron},
  \bibinfo{person}{Johann Gro{\ss}sch{\"{a}}dl}, \bibinfo{person}{Mehdi
  Tibouchi}, {and} \bibinfo{person}{Praveen~Kumar Vadnala}.}
  \bibinfo{year}{2015}\natexlab{}.
\newblock \showarticletitle{Conversion from Arithmetic to Boolean Masking with
  Logarithmic Complexity}. In \bibinfo{booktitle}{\emph{Proceedings of the 22nd
  International Workshop on Fast Software Encryption}}.
  \bibinfo{pages}{130--149}.
\newblock


\bibitem[\protect\citeauthoryear{Coron, Gro{\ss}sch{\"{a}}dl, and
  Vadnala}{Coron et~al\mbox{.}}{2014}]%
        {CGV14}
\bibfield{author}{\bibinfo{person}{Jean{-}S{\'{e}}bastien Coron},
  \bibinfo{person}{Johann Gro{\ss}sch{\"{a}}dl}, {and}
  \bibinfo{person}{Praveen~Kumar Vadnala}.} \bibinfo{year}{2014}\natexlab{}.
\newblock \showarticletitle{Secure Conversion between Boolean and Arithmetic
  Masking of Any Order}. In \bibinfo{booktitle}{\emph{Proceedings of the 16th
  International Workshop on Cryptographic Hardware and Embedded Systems}}.
  \bibinfo{pages}{188--205}.
\newblock


\bibitem[\protect\citeauthoryear{Coron, Prouff, and Rivain}{Coron
  et~al\mbox{.}}{2007}]%
        {CPR07}
\bibfield{author}{\bibinfo{person}{Jean{-}S{\'{e}}bastien Coron},
  \bibinfo{person}{Emmanuel Prouff}, {and} \bibinfo{person}{Matthieu Rivain}.}
  \bibinfo{year}{2007}\natexlab{}.
\newblock \showarticletitle{Side Channel Cryptanalysis of a Higher Order
  Masking Scheme}. In \bibinfo{booktitle}{\emph{Proceedings of the 9th
  International Workshop on Cryptographic Hardware and Embedded Systems}}.
  \bibinfo{pages}{28--44}.
\newblock


\bibitem[\protect\citeauthoryear{Coron, Prouff, Rivain, and Roche}{Coron
  et~al\mbox{.}}{2013}]%
        {CPRR13}
\bibfield{author}{\bibinfo{person}{Jean{-}S{\'{e}}bastien Coron},
  \bibinfo{person}{Emmanuel Prouff}, \bibinfo{person}{Matthieu Rivain}, {and}
  \bibinfo{person}{Thomas Roche}.} \bibinfo{year}{2013}\natexlab{}.
\newblock \showarticletitle{Higher-Order Side Channel Security and Mask
  Refreshing}. In \bibinfo{booktitle}{\emph{Proceedings of the 20th
  International Workshop on Fast Software Encryption}}.
  \bibinfo{pages}{410--424}.
\newblock


\bibitem[\protect\citeauthoryear{DaRolt, Das, Ghosh, Natale, Flottes, Rouzeyre,
  and Verbauwhede}{DaRolt et~al\mbox{.}}{2012}]%
        {DDGNFRV12}
\bibfield{author}{\bibinfo{person}{Jean DaRolt}, \bibinfo{person}{Amitabh Das},
  \bibinfo{person}{Santosh Ghosh}, \bibinfo{person}{Giorgio~Di Natale},
  \bibinfo{person}{Marie{-}Lise Flottes}, \bibinfo{person}{Bruno Rouzeyre},
  {and} \bibinfo{person}{Ingrid Verbauwhede}.} \bibinfo{year}{2012}\natexlab{}.
\newblock \showarticletitle{Scan Attacks on Side-Channel and Fault Attack
  Resistant Public-Key Implementations}.
\newblock \bibinfo{journal}{\emph{Journal of Cryptographic Engineering}}
  \bibinfo{volume}{2}, \bibinfo{number}{4} (\bibinfo{year}{2012}),
  \bibinfo{pages}{207--219}.
\newblock


\bibitem[\protect\citeauthoryear{de~Moura and Bj{\o}rner}{de~Moura and
  Bj{\o}rner}{2008}]%
        {MB08}
\bibfield{author}{\bibinfo{person}{Leonardo~Mendon{\c{c}}a de Moura} {and}
  \bibinfo{person}{Nikolaj Bj{\o}rner}.} \bibinfo{year}{2008}\natexlab{}.
\newblock \showarticletitle{{Z3:} An Efficient {SMT} Solver}. In
  \bibinfo{booktitle}{\emph{Proceedings of the 14th International Conference on
  Tools and Algorithms for the Construction and Analysis of Systems}}.
  \bibinfo{pages}{337--340}.
\newblock


\bibitem[\protect\citeauthoryear{Doychev, K{\"{o}}pf, Mauborgne, and
  Reineke}{Doychev et~al\mbox{.}}{2015}]%
        {DFKMR13}
\bibfield{author}{\bibinfo{person}{Goran Doychev}, \bibinfo{person}{Boris
  K{\"{o}}pf}, \bibinfo{person}{Laurent Mauborgne}, {and} \bibinfo{person}{Jan
  Reineke}.} \bibinfo{year}{2015}\natexlab{}.
\newblock \showarticletitle{CacheAudit: {A} Tool for the Static Analysis of
  Cache Side Channels}.
\newblock \bibinfo{journal}{\emph{{ACM Transactions on Information and System
  Security}}} \bibinfo{volume}{18}, \bibinfo{number}{1} (\bibinfo{year}{2015}),
  \bibinfo{pages}{4:1--4:32}.
\newblock


\bibitem[\protect\citeauthoryear{Duc, Dziembowski, and Faust}{Duc
  et~al\mbox{.}}{2019}]%
        {DDF14}
\bibfield{author}{\bibinfo{person}{Alexandre Duc}, \bibinfo{person}{Stefan
  Dziembowski}, {and} \bibinfo{person}{Sebastian Faust}.}
  \bibinfo{year}{2019}\natexlab{}.
\newblock \showarticletitle{Unifying Leakage Models: From Probing Attacks to
  Noisy Leakage}.
\newblock \bibinfo{journal}{\emph{Journal of Cryptology}} \bibinfo{volume}{32},
  \bibinfo{number}{1} (\bibinfo{year}{2019}), \bibinfo{pages}{151--177}.
\newblock


\bibitem[\protect\citeauthoryear{Eldib and Wang}{Eldib and Wang}{2014}]%
        {EW14}
\bibfield{author}{\bibinfo{person}{Hassan Eldib} {and} \bibinfo{person}{Chao
  Wang}.} \bibinfo{year}{2014}\natexlab{}.
\newblock \showarticletitle{Synthesis of Masking Countermeasures against Side
  Channel Attacks}. In \bibinfo{booktitle}{\emph{Proceedings of the 26th
  International Conference on Computer Aided Verification}}.
  \bibinfo{pages}{114--130}.
\newblock


\bibitem[\protect\citeauthoryear{Eldib, Wang, and Schaumont}{Eldib
  et~al\mbox{.}}{2014a}]%
        {EWS14b}
\bibfield{author}{\bibinfo{person}{Hassan Eldib}, \bibinfo{person}{Chao Wang},
  {and} \bibinfo{person}{Patrick Schaumont}.} \bibinfo{year}{2014}\natexlab{a}.
\newblock \showarticletitle{Formal Verification Of Software Countermeasures
  against Side-Channel Attacks}.
\newblock \bibinfo{journal}{\emph{ACM Transactions on Software Engineering and
  Methodology}} \bibinfo{volume}{24}, \bibinfo{number}{2}
  (\bibinfo{year}{2014}), \bibinfo{pages}{11}.
\newblock


\bibitem[\protect\citeauthoryear{Eldib, Wang, and Schaumont}{Eldib
  et~al\mbox{.}}{2014b}]%
        {EWS14a}
\bibfield{author}{\bibinfo{person}{Hassan Eldib}, \bibinfo{person}{Chao Wang},
  {and} \bibinfo{person}{Patrick Schaumont}.} \bibinfo{year}{2014}\natexlab{b}.
\newblock \showarticletitle{{SMT}-Based Verification of Software
  Countermeasures against Side-Channel Attacks}. In
  \bibinfo{booktitle}{\emph{Proceedings of the 20th International Conference on
  Tools and Algorithms for the Construction and Analysis of Systems}}.
  \bibinfo{pages}{62--77}.
\newblock


\bibitem[\protect\citeauthoryear{Eldib, Wu, and Wang}{Eldib
  et~al\mbox{.}}{2016}]%
        {EWW16}
\bibfield{author}{\bibinfo{person}{Hassan Eldib}, \bibinfo{person}{Meng Wu},
  {and} \bibinfo{person}{Chao Wang}.} \bibinfo{year}{2016}\natexlab{}.
\newblock \showarticletitle{Synthesis of Fault-Attack Countermeasures for
  Cryptographic Circuits}. In \bibinfo{booktitle}{\emph{Proceedings of the 28th
  International Conference Computer Aided Verification}}.
  \bibinfo{pages}{343--363}.
\newblock


\bibitem[\protect\citeauthoryear{Faust, Grosso, Pozo, Paglialonga, and
  Standaert}{Faust et~al\mbox{.}}{2017}]%
        {FGPPS17}
\bibfield{author}{\bibinfo{person}{Sebastian Faust}, \bibinfo{person}{Vincent
  Grosso}, \bibinfo{person}{Santos Merino~Del Pozo}, \bibinfo{person}{Clara
  Paglialonga}, {and} \bibinfo{person}{Fran{\c{c}}ois{-}Xavier Standaert}.}
  \bibinfo{year}{2017}\natexlab{}.
\newblock \showarticletitle{Composable Masking Schemes in the Presence of
  Physical Defaults and the Robust Probing Model}.
\newblock \bibinfo{journal}{\emph{{IACR} Cryptology ePrint Archive}}
  \bibinfo{volume}{2017} (\bibinfo{year}{2017}), \bibinfo{pages}{711}.
\newblock


\bibitem[\protect\citeauthoryear{Faust, Grosso, Pozo, Paglialonga, and
  Standaert}{Faust et~al\mbox{.}}{2018}]%
        {FGPPS18}
\bibfield{author}{\bibinfo{person}{Sebastian Faust}, \bibinfo{person}{Vincent
  Grosso}, \bibinfo{person}{Santos Merino~Del Pozo}, \bibinfo{person}{Clara
  Paglialonga}, {and} \bibinfo{person}{Fran{\c{c}}ois{-}Xavier Standaert}.}
  \bibinfo{year}{2018}\natexlab{}.
\newblock \showarticletitle{Composable Masking Schemes in the Presence of
  Physical Defaults {\&} the Robust Probing Model}.
\newblock \bibinfo{journal}{\emph{IACR Transactions on Cryptographic Hardware
  and Embedded Systems}} \bibinfo{number}{3} (\bibinfo{year}{2018}),
  \bibinfo{pages}{89--120}.
\newblock


\bibitem[\protect\citeauthoryear{Gao, Xie, Zhang, Song, and Chen}{Gao
  et~al\mbox{.}}{2019a}]%
        {GXZSC19}
\bibfield{author}{\bibinfo{person}{Pengfei Gao}, \bibinfo{person}{Hongyi Xie},
  \bibinfo{person}{Jun Zhang}, \bibinfo{person}{Fu Song}, {and}
  \bibinfo{person}{Taolue Chen}.} \bibinfo{year}{2019}\natexlab{a}.
\newblock \showarticletitle{Quantitative Verification of Masked Arithmetic
  Programs Against Side-Channel Attacks}. In
  \bibinfo{booktitle}{\emph{Proceedings of the 25th International Conference on
  Tools and Algorithms for the Construction and Analysis of Systems}}.
  \bibinfo{pages}{155--173}.
\newblock


\bibitem[\protect\citeauthoryear{Gao, Zhang, Song, and Wang}{Gao
  et~al\mbox{.}}{2019b}]%
        {GaoZSW19}
\bibfield{author}{\bibinfo{person}{Pengfei Gao}, \bibinfo{person}{Jun Zhang},
  \bibinfo{person}{Fu Song}, {and} \bibinfo{person}{Chao Wang}.}
  \bibinfo{year}{2019}\natexlab{b}.
\newblock \showarticletitle{Verifying and Quantifying Side-channel Resistance
  of Masked Software Implementations}.
\newblock \bibinfo{journal}{\emph{{ACM} Trans. Softw. Eng. Methodol.}}
  \bibinfo{volume}{28}, \bibinfo{number}{3} (\bibinfo{year}{2019}),
  \bibinfo{pages}{16:1--16:32}.
\newblock


\bibitem[\protect\citeauthoryear{Genkin, Shamir, and Tromer}{Genkin
  et~al\mbox{.}}{2017}]%
        {GST17}
\bibfield{author}{\bibinfo{person}{Daniel Genkin}, \bibinfo{person}{Adi
  Shamir}, {and} \bibinfo{person}{Eran Tromer}.}
  \bibinfo{year}{2017}\natexlab{}.
\newblock \showarticletitle{Acoustic Cryptanalysis}.
\newblock \bibinfo{journal}{\emph{Journal of Cryptology}} \bibinfo{volume}{30},
  \bibinfo{number}{2} (\bibinfo{year}{2017}), \bibinfo{pages}{392--443}.
\newblock


\bibitem[\protect\citeauthoryear{Goodwill, Jun, Jaffe, and Rohatgi}{Goodwill
  et~al\mbox{.}}{2011}]%
        {GJJR11}
\bibfield{author}{\bibinfo{person}{Gilbert Goodwill}, \bibinfo{person}{Benjamin
  Jun}, \bibinfo{person}{Josh Jaffe}, {and} \bibinfo{person}{Pankaj Rohatgi}.}
  \bibinfo{year}{2011}\natexlab{}.
\newblock \showarticletitle{A Testing Methodology for Side Channel Resistance
  Validation}. In \bibinfo{booktitle}{\emph{NIST non-invasive attack testing
  workshop}}.
\newblock


\bibitem[\protect\citeauthoryear{Goubin}{Goubin}{2001}]%
        {Goubin01}
\bibfield{author}{\bibinfo{person}{Louis Goubin}.}
  \bibinfo{year}{2001}\natexlab{}.
\newblock \showarticletitle{A Sound Method for Switching between Boolean and
  Arithmetic Masking}. In \bibinfo{booktitle}{\emph{Proceedings of the 3rd
  International Workshop on Cryptographic Hardware and Embedded Systems}}.
  \bibinfo{pages}{3--15}.
\newblock


\bibitem[\protect\citeauthoryear{Grabher, Gro{\ss}sch{\"{a}}dl, and
  Page}{Grabher et~al\mbox{.}}{2007}]%
        {GGP07}
\bibfield{author}{\bibinfo{person}{Philipp Grabher}, \bibinfo{person}{Johann
  Gro{\ss}sch{\"{a}}dl}, {and} \bibinfo{person}{Dan Page}.}
  \bibinfo{year}{2007}\natexlab{}.
\newblock \showarticletitle{Cryptographic Side-Channels from Low-Power Cache
  Memory}. In \bibinfo{booktitle}{\emph{Proceedings of the 11th {IMA}
  International Conference on Cryptography and Coding}}.
  \bibinfo{pages}{170--184}.
\newblock


\bibitem[\protect\citeauthoryear{Gro{\ss}, Krenn, and Mangard}{Gro{\ss}
  et~al\mbox{.}}{2018}]%
        {GrossMK18}
\bibfield{author}{\bibinfo{person}{Hannes Gro{\ss}}, \bibinfo{person}{Martin
  Krenn}, {and} \bibinfo{person}{Stefan Mangard}.}
  \bibinfo{year}{2018}\natexlab{}.
\newblock \showarticletitle{Second and Third Order Verilog Implementations of
  {AES} S-box}.
\newblock


\bibitem[\protect\citeauthoryear{Gro{\ss} and Mangard}{Gro{\ss} and
  Mangard}{2018}]%
        {GrossM18}
\bibfield{author}{\bibinfo{person}{Hannes Gro{\ss}} {and}
  \bibinfo{person}{Stefan Mangard}.} \bibinfo{year}{2018}\natexlab{}.
\newblock \showarticletitle{A Unified Masking Approach}.
\newblock \bibinfo{journal}{\emph{Journal of Cryptographic Engineering}}
  \bibinfo{volume}{8}, \bibinfo{number}{2} (\bibinfo{year}{2018}),
  \bibinfo{pages}{109--124}.
\newblock


\bibitem[\protect\citeauthoryear{Gro{\ss}, Mangard, and Korak}{Gro{\ss}
  et~al\mbox{.}}{2017a}]%
        {GrossMK17}
\bibfield{author}{\bibinfo{person}{Hannes Gro{\ss}}, \bibinfo{person}{Stefan
  Mangard}, {and} \bibinfo{person}{Thomas Korak}.}
  \bibinfo{year}{2017}\natexlab{a}.
\newblock \showarticletitle{An Efficient Side-Channel Protected {AES}
  Implementation with Arbitrary Protection Order}. In
  \bibinfo{booktitle}{\emph{Proceedings of the {RSA} Conference Cryptographers'
  Track}}. \bibinfo{pages}{95--112}.
\newblock


\bibitem[\protect\citeauthoryear{Gro{\ss}, Schaffenrath, and Mangard}{Gro{\ss}
  et~al\mbox{.}}{2017b}]%
        {GSM17}
\bibfield{author}{\bibinfo{person}{Hannes Gro{\ss}}, \bibinfo{person}{David
  Schaffenrath}, {and} \bibinfo{person}{Stefan Mangard}.}
  \bibinfo{year}{2017}\natexlab{b}.
\newblock \showarticletitle{Higher-Order Side-Channel Protected Implementations
  of {KECCAK}}. In \bibinfo{booktitle}{\emph{Proceedings of the Euromicro
  Conference on Digital System Design}}. \bibinfo{pages}{205--212}.
\newblock


\bibitem[\protect\citeauthoryear{Guo, Wu, and Wang}{Guo et~al\mbox{.}}{2018}]%
        {GWW18}
\bibfield{author}{\bibinfo{person}{Shengjian Guo}, \bibinfo{person}{Meng Wu},
  {and} \bibinfo{person}{Chao Wang}.} \bibinfo{year}{2018}\natexlab{}.
\newblock \showarticletitle{Adversarial Symbolic Execution for Detecting
  Concurrency-related Cache Timing Leaks}. In
  \bibinfo{booktitle}{\emph{Proceedings of the 2018 {ACM} Joint Meeting on
  European Software Engineering Conference and Symposium on the Foundations of
  Software Engineering}}. \bibinfo{pages}{377--388}.
\newblock


\bibitem[\protect\citeauthoryear{Hospodar, Gierlichs, Mulder, Verbauwhede, and
  Vandewalle}{Hospodar et~al\mbox{.}}{2011}]%
        {BGMVV11}
\bibfield{author}{\bibinfo{person}{Gabriel Hospodar}, \bibinfo{person}{Benedikt
  Gierlichs}, \bibinfo{person}{Elke~De Mulder}, \bibinfo{person}{Ingrid
  Verbauwhede}, {and} \bibinfo{person}{Joos Vandewalle}.}
  \bibinfo{year}{2011}\natexlab{}.
\newblock \showarticletitle{Machine learning in side-channel analysis: a first
  study}.
\newblock \bibinfo{journal}{\emph{Journal of Cryptographic Engineering}}
  \bibinfo{volume}{1}, \bibinfo{number}{4} (\bibinfo{year}{2011}),
  \bibinfo{pages}{293--302}.
\newblock


\bibitem[\protect\citeauthoryear{Hou, Breier, Zhang, and Liu}{Hou
  et~al\mbox{.}}{2019}]%
        {HouBZL19}
\bibfield{author}{\bibinfo{person}{Xiaolu Hou}, \bibinfo{person}{Jakub Breier},
  \bibinfo{person}{Fuyuan Zhang}, {and} \bibinfo{person}{Yang Liu}.}
  \bibinfo{year}{2019}\natexlab{}.
\newblock \showarticletitle{Fully Automated Differential Fault Analysis on
  Software Implementations of Block Ciphers}.
\newblock \bibinfo{journal}{\emph{{IACR} Transactions on Cryptographic Hardware
  and Embedded Systems}} \bibinfo{volume}{2019}, \bibinfo{number}{3}
  (\bibinfo{year}{2019}), \bibinfo{pages}{1--29}.
\newblock


\bibitem[\protect\citeauthoryear{Hutter and Tunstall}{Hutter and
  Tunstall}{2019}]%
        {HT19}
\bibfield{author}{\bibinfo{person}{Michael Hutter} {and}
  \bibinfo{person}{Michael Tunstall}.} \bibinfo{year}{2019}\natexlab{}.
\newblock \showarticletitle{Constant-Time Higher-Order Boolean-To-Arithmetic
  Masking}.
\newblock \bibinfo{journal}{\emph{Journal of Cryptographic Engineering}}
  \bibinfo{volume}{9}, \bibinfo{number}{2} (\bibinfo{year}{2019}),
  \bibinfo{pages}{173--184}.
\newblock


\bibitem[\protect\citeauthoryear{Ishai, Sahai, and Wagner}{Ishai
  et~al\mbox{.}}{2003}]%
        {ISW03}
\bibfield{author}{\bibinfo{person}{Yuval Ishai}, \bibinfo{person}{Amit Sahai},
  {and} \bibinfo{person}{David~A. Wagner}.} \bibinfo{year}{2003}\natexlab{}.
\newblock \showarticletitle{Private Circuits: Securing Hardware against Probing
  Attacks}. In \bibinfo{booktitle}{\emph{Proceedings of the 23rd Annual
  International Cryptology Conference}}. \bibinfo{pages}{463--481}.
\newblock


\bibitem[\protect\citeauthoryear{Karl, Schilling, Bloem, and Mangard}{Karl
  et~al\mbox{.}}{2019}]%
        {KSBM19}
\bibfield{author}{\bibinfo{person}{Anja~F. Karl}, \bibinfo{person}{Robert
  Schilling}, \bibinfo{person}{Roderick Bloem}, {and} \bibinfo{person}{Stefan
  Mangard}.} \bibinfo{year}{2019}\natexlab{}.
\newblock \showarticletitle{Small Faults Grow Up - Verification of Error
  Masking Robustness in Arithmetically Encoded Programs}. In
  \bibinfo{booktitle}{\emph{Proceedings of the 20th International Conference on
  Verification, Model Checking, and Abstract Interpretation}}.
  \bibinfo{pages}{183--204}.
\newblock


\bibitem[\protect\citeauthoryear{Kim, Hong, and Lim}{Kim et~al\mbox{.}}{2011}]%
        {KHL11}
\bibfield{author}{\bibinfo{person}{HeeSeok Kim}, \bibinfo{person}{Seokhie
  Hong}, {and} \bibinfo{person}{Jongin Lim}.} \bibinfo{year}{2011}\natexlab{}.
\newblock \showarticletitle{A Fast and Provably Secure Higher-Order Masking of
  {AES} S-Box}. In \bibinfo{booktitle}{\emph{Proceedings of the 13th
  International Workshop on Cryptographic Hardware and Embedded Systems}}.
  \bibinfo{pages}{95--107}.
\newblock


\bibitem[\protect\citeauthoryear{Kim and Ko}{Kim and Ko}{2013}]%
        {KimK13}
\bibfield{author}{\bibinfo{person}{Yongdae Kim} {and}
  \bibinfo{person}{Haengseok Ko}.} \bibinfo{year}{2013}\natexlab{}.
\newblock \showarticletitle{Using Principal Component Analysis for Practical
  Biasing of Power Traces to Improve Power Analysis Attacks}. In
  \bibinfo{booktitle}{\emph{Proceedings of the 16th International Conference on
  Information Security and Cryptology}}. \bibinfo{pages}{109--120}.
\newblock


\bibitem[\protect\citeauthoryear{Kocher}{Kocher}{1996}]%
        {Kocher96}
\bibfield{author}{\bibinfo{person}{Paul~C. Kocher}.}
  \bibinfo{year}{1996}\natexlab{}.
\newblock \showarticletitle{Timing Attacks on Implementations of
  Diffie-Hellman, RSA, DSS, and Other Systems}. In
  \bibinfo{booktitle}{\emph{Proceedings of the 16th Annual International
  Cryptology Conference}}. \bibinfo{pages}{104--113}.
\newblock


\bibitem[\protect\citeauthoryear{Kocher, Jaffe, and Jun}{Kocher
  et~al\mbox{.}}{1999}]%
        {KJJ99}
\bibfield{author}{\bibinfo{person}{Paul~C. Kocher}, \bibinfo{person}{Joshua
  Jaffe}, {and} \bibinfo{person}{Benjamin Jun}.}
  \bibinfo{year}{1999}\natexlab{}.
\newblock \showarticletitle{Differential Power Analysis}. In
  \bibinfo{booktitle}{\emph{Proceedings of the 19th Annual International
  Cryptology Conference}}. \bibinfo{pages}{388--397}.
\newblock


\bibitem[\protect\citeauthoryear{K{\"o}pf, Mauborgne, and Ochoa}{K{\"o}pf
  et~al\mbox{.}}{2012}]%
        {KopfMO12}
\bibfield{author}{\bibinfo{person}{Boris K{\"o}pf}, \bibinfo{person}{Laurent
  Mauborgne}, {and} \bibinfo{person}{Mart\'{\i}n Ochoa}.}
  \bibinfo{year}{2012}\natexlab{}.
\newblock \showarticletitle{Automatic quantification of cache side-channels}.
  In \bibinfo{booktitle}{\emph{International Conference on Computer Aided
  Verification}}. \bibinfo{pages}{564--580}.
\newblock


\bibitem[\protect\citeauthoryear{Lai and Massey}{Lai and Massey}{1990}]%
        {LM90}
\bibfield{author}{\bibinfo{person}{Xuejia Lai} {and} \bibinfo{person}{James~L.
  Massey}.} \bibinfo{year}{1990}\natexlab{}.
\newblock \showarticletitle{A Proposal for a New Block Encryption Standard}. In
  \bibinfo{booktitle}{\emph{Proceedings of the Workshop on the Theory and
  Application of of Cryptographic Techniques}}. \bibinfo{pages}{389--404}.
\newblock


\bibitem[\protect\citeauthoryear{Mahmudlu, Banciu, Batina, and Buhan}{Mahmudlu
  et~al\mbox{.}}{2016}]%
        {MBBB16}
\bibfield{author}{\bibinfo{person}{Rauf Mahmudlu}, \bibinfo{person}{Valentina
  Banciu}, \bibinfo{person}{Lejla Batina}, {and} \bibinfo{person}{Ileana
  Buhan}.} \bibinfo{year}{2016}\natexlab{}.
\newblock \showarticletitle{{LDA}-Based Clustering as a Side-Channel
  Distinguisher}. In \bibinfo{booktitle}{\emph{Proceedings of the 12th
  International Workshop on Radio Frequency Identification and {IoT}
  Security}}. \bibinfo{pages}{62--75}.
\newblock


\bibitem[\protect\citeauthoryear{Mangard}{Mangard}{2004}]%
        {Mangard04}
\bibfield{author}{\bibinfo{person}{Stefan Mangard}.}
  \bibinfo{year}{2004}\natexlab{}.
\newblock \showarticletitle{Hardware Countermeasures against {DPA} ? {A}
  Statistical Analysis of Their Effectiveness}. In
  \bibinfo{booktitle}{\emph{Proceedings of the {RSA} Conference Cryptographers'
  Track}}. \bibinfo{pages}{222--235}.
\newblock


\bibitem[\protect\citeauthoryear{Mangard, Oswald, and Popp}{Mangard
  et~al\mbox{.}}{2007}]%
        {MOP07}
\bibfield{author}{\bibinfo{person}{Stefan Mangard}, \bibinfo{person}{Elisabeth
  Oswald}, {and} \bibinfo{person}{Thomas Popp}.}
  \bibinfo{year}{2007}\natexlab{}.
\newblock \bibinfo{booktitle}{\emph{Power Analysis Attacks - Revealing the
  Secrets of Smart Cards}}.
\newblock \bibinfo{publisher}{Springer}.
\newblock
\showISBNx{978-0-387-30857-9}


\bibitem[\protect\citeauthoryear{Mangard, Pramstaller, and Oswald}{Mangard
  et~al\mbox{.}}{2005}]%
        {MPO05}
\bibfield{author}{\bibinfo{person}{Stefan Mangard}, \bibinfo{person}{Norbert
  Pramstaller}, {and} \bibinfo{person}{Elisabeth Oswald}.}
  \bibinfo{year}{2005}\natexlab{}.
\newblock \showarticletitle{Successfully Attacking Masked {AES} Hardware
  Implementations}. In \bibinfo{booktitle}{\emph{Proceedings of the 7th
  International Workshop on Cryptographic Hardware and Embedded Systems}}.
  \bibinfo{pages}{157--171}.
\newblock


\bibitem[\protect\citeauthoryear{Mavroeidis, Batina, van Laarhoven, and
  Marchiori}{Mavroeidis et~al\mbox{.}}{2012}]%
        {MBLM12}
\bibfield{author}{\bibinfo{person}{Dimitrios Mavroeidis},
  \bibinfo{person}{Lejla Batina}, \bibinfo{person}{Twan van Laarhoven}, {and}
  \bibinfo{person}{Elena Marchiori}.} \bibinfo{year}{2012}\natexlab{}.
\newblock \showarticletitle{{PCA}, Eigenvector Localization and Clustering for
  Side-Channel Attacks on Cryptographic Hardware Devices}. In
  \bibinfo{booktitle}{\emph{Proceedings of the European Conference on Machine
  Learning and Knowledge Discovery in Databases}}. \bibinfo{pages}{253--268}.
\newblock


\bibitem[\protect\citeauthoryear{Messerges}{Messerges}{2000}]%
        {Mess00}
\bibfield{author}{\bibinfo{person}{Thomas~S. Messerges}.}
  \bibinfo{year}{2000}\natexlab{}.
\newblock \showarticletitle{Securing the {AES} Finalists Against Power Analysis
  Attacks}. In \bibinfo{booktitle}{\emph{Proceedings of the 7th International
  Workshop on Fast Software Encryption}}. \bibinfo{pages}{150--164}.
\newblock


\bibitem[\protect\citeauthoryear{Moradi, Poschmann, Ling, Paar, and
  Wang}{Moradi et~al\mbox{.}}{2011}]%
        {MPLPW11}
\bibfield{author}{\bibinfo{person}{Amir Moradi}, \bibinfo{person}{Axel
  Poschmann}, \bibinfo{person}{San Ling}, \bibinfo{person}{Christof Paar},
  {and} \bibinfo{person}{Huaxiong Wang}.} \bibinfo{year}{2011}\natexlab{}.
\newblock \showarticletitle{Pushing the Limits: {A} Very Compact and a
  Threshold Implementation of {AES}}. In \bibinfo{booktitle}{\emph{Proceedings
  of the 30th Annual International Conference on the Theory and Applications of
  Cryptographic Techniques}}. \bibinfo{pages}{69--88}.
\newblock


\bibitem[\protect\citeauthoryear{Moss, Oswald, Page, and Tunstall}{Moss
  et~al\mbox{.}}{2012}]%
        {MOPT12}
\bibfield{author}{\bibinfo{person}{Andrew Moss}, \bibinfo{person}{Elisabeth
  Oswald}, \bibinfo{person}{Dan Page}, {and} \bibinfo{person}{Michael
  Tunstall}.} \bibinfo{year}{2012}\natexlab{}.
\newblock \showarticletitle{Compiler Assisted Masking}. In
  \bibinfo{booktitle}{\emph{Proceedings of the 14th International Workshop on
  Cryptographic Hardware and Embedded Systems}}. \bibinfo{pages}{58--75}.
\newblock


\bibitem[\protect\citeauthoryear{Oren, Weisse, and Wool}{Oren
  et~al\mbox{.}}{2014}]%
        {OrenWW14}
\bibfield{author}{\bibinfo{person}{Yossef Oren}, \bibinfo{person}{Ofir Weisse},
  {and} \bibinfo{person}{Avishai Wool}.} \bibinfo{year}{2014}\natexlab{}.
\newblock \showarticletitle{A New Framework for Constraint-Based Probabilistic
  Template Side Channel Attacks}. In \bibinfo{booktitle}{\emph{Proceedings of
  the 16th International Workshop on Cryptographic Hardware and Embedded
  Systems}}. \bibinfo{pages}{17--34}.
\newblock


\bibitem[\protect\citeauthoryear{Oren and Wool}{Oren and Wool}{2016}]%
        {OrenW16}
\bibfield{author}{\bibinfo{person}{Yossef Oren} {and} \bibinfo{person}{Avishai
  Wool}.} \bibinfo{year}{2016}\natexlab{}.
\newblock \showarticletitle{Side-channel Cryptographic Attacks using
  Pseudo-Boolean Optimization}.
\newblock \bibinfo{journal}{\emph{Constraints}} \bibinfo{volume}{21},
  \bibinfo{number}{4} (\bibinfo{year}{2016}), \bibinfo{pages}{616--645}.
\newblock


\bibitem[\protect\citeauthoryear{Oswald, Mangard, Pramstaller, and
  Rijmen}{Oswald et~al\mbox{.}}{2005}]%
        {OMPR05}
\bibfield{author}{\bibinfo{person}{Elisabeth Oswald}, \bibinfo{person}{Stefan
  Mangard}, \bibinfo{person}{Norbert Pramstaller}, {and}
  \bibinfo{person}{Vincent Rijmen}.} \bibinfo{year}{2005}\natexlab{}.
\newblock \showarticletitle{A Side-Channel Analysis Resistant Description of
  the {AES} S-Box}. In \bibinfo{booktitle}{\emph{Proceedings of the 12th
  International Workshop on Fast Software Encryption}}.
  \bibinfo{pages}{413--423}.
\newblock


\bibitem[\protect\citeauthoryear{Ouahma, Meunier, Heydemann, and
  Encrenaz}{Ouahma et~al\mbox{.}}{2017}]%
        {OMHE17}
\bibfield{author}{\bibinfo{person}{In{\`e}s Ben~El Ouahma},
  \bibinfo{person}{Quentin Meunier}, \bibinfo{person}{Karine Heydemann}, {and}
  \bibinfo{person}{Emmanuelle Encrenaz}.} \bibinfo{year}{2017}\natexlab{}.
\newblock \showarticletitle{Symbolic Approach for Side-Channel Resistance
  Analysis of Masked Assembly Codes}. In \bibinfo{booktitle}{\emph{Proceedings
  of the 6th International Workshop on Security Proofs for Embedded Systems}}.
\newblock


\bibitem[\protect\citeauthoryear{Ouahma, Meunier, Heydemann, and
  Encrenaz}{Ouahma et~al\mbox{.}}{2019}]%
        {OuahmaMHE19}
\bibfield{author}{\bibinfo{person}{In{\`{e}}s Ben~El Ouahma},
  \bibinfo{person}{Quentin~L. Meunier}, \bibinfo{person}{Karine Heydemann},
  {and} \bibinfo{person}{Emmanuelle Encrenaz}.}
  \bibinfo{year}{2019}\natexlab{}.
\newblock \showarticletitle{Side-channel robustness analysis of masked assembly
  codes using a symbolic approach}.
\newblock \bibinfo{journal}{\emph{Journal of Cryptographic Engineering}}
  \bibinfo{volume}{9}, \bibinfo{number}{3} (\bibinfo{year}{2019}),
  \bibinfo{pages}{231--242}.
\newblock


\bibitem[\protect\citeauthoryear{Ozgen, Papachristodoulou, and Batina}{Ozgen
  et~al\mbox{.}}{2016}]%
        {OPB16}
\bibfield{author}{\bibinfo{person}{Elif Ozgen}, \bibinfo{person}{Louiza
  Papachristodoulou}, {and} \bibinfo{person}{Lejla Batina}.}
  \bibinfo{year}{2016}\natexlab{}.
\newblock \showarticletitle{Template attacks using classification algorithms}.
  In \bibinfo{booktitle}{\emph{Proceedings of the {IEEE} International
  Symposium on Hardware Oriented Security and Trust}}.
  \bibinfo{pages}{242--247}.
\newblock


\bibitem[\protect\citeauthoryear{Pasareanu, Phan, and Malacaria}{Pasareanu
  et~al\mbox{.}}{[n.d.]}]%
        {PasareanuPM16}
\bibfield{author}{\bibinfo{person}{Corina~S. Pasareanu},
  \bibinfo{person}{Quoc{-}Sang Phan}, {and} \bibinfo{person}{Pasquale
  Malacaria}.} \bibinfo{year}{[n.d.]}\natexlab{}.
\newblock \showarticletitle{Multi-run Side-Channel Analysis Using Symbolic
  Execution and {Max-SMT}}. In \bibinfo{booktitle}{\emph{Proceedings of the
  29th {IEEE} Computer Security Foundations Symposium, pages = {387--400}, year
  = {2016}}}.
\newblock


\bibitem[\protect\citeauthoryear{Phan, Bang, Pasareanu, Malacaria, and
  Bultan}{Phan et~al\mbox{.}}{2017}]%
        {PhanBPMB17}
\bibfield{author}{\bibinfo{person}{Quoc{-}Sang Phan}, \bibinfo{person}{Lucas
  Bang}, \bibinfo{person}{Corina~S. Pasareanu}, \bibinfo{person}{Pasquale
  Malacaria}, {and} \bibinfo{person}{Tevfik Bultan}.}
  \bibinfo{year}{2017}\natexlab{}.
\newblock \showarticletitle{Synthesis of Adaptive Side-Channel Attacks}. In
  \bibinfo{booktitle}{\emph{Proceedings of the 30th {IEEE} Computer Security
  Foundations Symposium}}. \bibinfo{pages}{328--342}.
\newblock


\bibitem[\protect\citeauthoryear{Picek, Papagiannopoulos, Ege, Batina, and
  Jakobovic}{Picek et~al\mbox{.}}{2014}]%
        {PPEBJ14}
\bibfield{author}{\bibinfo{person}{Stjepan Picek}, \bibinfo{person}{Kostas
  Papagiannopoulos}, \bibinfo{person}{Baris Ege}, \bibinfo{person}{Lejla
  Batina}, {and} \bibinfo{person}{Domagoj Jakobovic}.}
  \bibinfo{year}{2014}\natexlab{}.
\newblock \showarticletitle{Confused by Confusion: Systematic Evaluation of
  {DPA} Resistance of Various S-boxes}. In
  \bibinfo{booktitle}{\emph{Proceedings of the 15th International Conference on
  Cryptology in India}}. \bibinfo{pages}{374--390}.
\newblock


\bibitem[\protect\citeauthoryear{Prouff and Rivain}{Prouff and Rivain}{2013}]%
        {PR13}
\bibfield{author}{\bibinfo{person}{Emmanuel Prouff} {and}
  \bibinfo{person}{Matthieu Rivain}.} \bibinfo{year}{2013}\natexlab{}.
\newblock \showarticletitle{Masking against Side-Channel Attacks: {A} Formal
  Security Proof}. In \bibinfo{booktitle}{\emph{Proceedings of the 32nd Annual
  International Conference on the Theory and Applications of Cryptographic
  Techniques}}. \bibinfo{pages}{142--159}.
\newblock


\bibitem[\protect\citeauthoryear{Reparaz, Bilgin, Nikova, Gierlichs, and
  Verbauwhede}{Reparaz et~al\mbox{.}}{2015}]%
        {RBNGV15}
\bibfield{author}{\bibinfo{person}{Oscar Reparaz}, \bibinfo{person}{Beg{\"{u}}l
  Bilgin}, \bibinfo{person}{Svetla Nikova}, \bibinfo{person}{Benedikt
  Gierlichs}, {and} \bibinfo{person}{Ingrid Verbauwhede}.}
  \bibinfo{year}{2015}\natexlab{}.
\newblock \showarticletitle{Consolidating Masking Schemes}. In
  \bibinfo{booktitle}{\emph{Proceedings of the 35th Annual Cryptology
  Conference}}. \bibinfo{pages}{764--783}.
\newblock


\bibitem[\protect\citeauthoryear{Rivain and Prouff}{Rivain and Prouff}{2010}]%
        {RP10}
\bibfield{author}{\bibinfo{person}{Matthieu Rivain} {and}
  \bibinfo{person}{Emmanuel Prouff}.} \bibinfo{year}{2010}\natexlab{}.
\newblock \showarticletitle{Provably Secure Higher-Order Masking of {AES}}. In
  \bibinfo{booktitle}{\emph{Proceedings of the 12th International Workshop on
  Cryptographic Hardware and Embedded Systems}}. \bibinfo{pages}{413--427}.
\newblock


\bibitem[\protect\citeauthoryear{Schneider, Paglialonga, Oder, and
  G{\"{u}}neysu}{Schneider et~al\mbox{.}}{2019}]%
        {SchneiderPOG19}
\bibfield{author}{\bibinfo{person}{Tobias Schneider}, \bibinfo{person}{Clara
  Paglialonga}, \bibinfo{person}{Tobias Oder}, {and} \bibinfo{person}{Tim
  G{\"{u}}neysu}.} \bibinfo{year}{2019}\natexlab{}.
\newblock \showarticletitle{Efficiently Masking Binomial Sampling at Arbitrary
  Orders for Lattice-Based Crypto}. In \bibinfo{booktitle}{\emph{Proceedings of
  the 22nd {IACR} International Conference on Practice and Theory of Public-Key
  Cryptography}}. \bibinfo{pages}{534--564}.
\newblock


\bibitem[\protect\citeauthoryear{Schramm and Paar}{Schramm and Paar}{2006}]%
        {SP06}
\bibfield{author}{\bibinfo{person}{Kai Schramm} {and} \bibinfo{person}{Christof
  Paar}.} \bibinfo{year}{2006}\natexlab{}.
\newblock \showarticletitle{Higher Order Masking of the {AES}}. In
  \bibinfo{booktitle}{\emph{Proceedings of the {RSA} Conference Cryptographers'
  Track}}. \bibinfo{pages}{208--225}.
\newblock


\bibitem[\protect\citeauthoryear{Standaert}{Standaert}{2017}]%
        {Standaert17}
\bibfield{author}{\bibinfo{person}{Fran{\c{c}}ois{-}Xavier Standaert}.}
  \bibinfo{year}{2017}\natexlab{}.
\newblock \showarticletitle{How (not) to Use Welch's T-test in Side-Channel
  Security Evaluations}.
\newblock \bibinfo{journal}{\emph{{IACR} Cryptology ePrint Archive}}
  \bibinfo{volume}{2017} (\bibinfo{year}{2017}), \bibinfo{pages}{138}.
\newblock


\bibitem[\protect\citeauthoryear{Sung, Paulsen, and Wang}{Sung
  et~al\mbox{.}}{2018}]%
        {SungPW18}
\bibfield{author}{\bibinfo{person}{Chungha Sung}, \bibinfo{person}{Brandon
  Paulsen}, {and} \bibinfo{person}{Chao Wang}.}
  \bibinfo{year}{2018}\natexlab{}.
\newblock \showarticletitle{{CANAL:} a cache timing analysis framework via
  {LLVM} transformation}. In \bibinfo{booktitle}{\emph{Proceedings of the 33rd
  {ACM/IEEE} International Conference on Automated Software Engineering}}.
  \bibinfo{pages}{904--907}.
\newblock


\bibitem[\protect\citeauthoryear{Tang, Zhou, Zhang, and Qiu}{Tang
  et~al\mbox{.}}{2015}]%
        {TangZZQ15}
\bibfield{author}{\bibinfo{person}{Jiehui Tang}, \bibinfo{person}{Yongbin
  Zhou}, \bibinfo{person}{Hailong Zhang}, {and} \bibinfo{person}{Shuang Qiu}.}
  \bibinfo{year}{2015}\natexlab{}.
\newblock \showarticletitle{Higher-Order Masking Schemes for Simon}. In
  \bibinfo{booktitle}{\emph{Proceedings of the 17th International Conference on
  Information and Communications Security}}. \bibinfo{pages}{379--392}.
\newblock


\bibitem[\protect\citeauthoryear{Wagner}{Wagner}{2004}]%
        {Wagner04}
\bibfield{author}{\bibinfo{person}{David~A. Wagner}.}
  \bibinfo{year}{2004}\natexlab{}.
\newblock \showarticletitle{Cryptanalysis of a Provably Secure {CRT-RSA}
  Algorithm}. In \bibinfo{booktitle}{\emph{Proceedings of the 11th {ACM}
  Conference on Computer and Communications Security}}.
  \bibinfo{pages}{92--97}.
\newblock


\bibitem[\protect\citeauthoryear{Wang and Schaumont}{Wang and
  Schaumont}{2017}]%
        {WangS17}
\bibfield{author}{\bibinfo{person}{Chao Wang} {and} \bibinfo{person}{Patrick
  Schaumont}.} \bibinfo{year}{2017}\natexlab{}.
\newblock \showarticletitle{Security by compilation: an automated approach to
  comprehensive side-channel resistance}.
\newblock \bibinfo{journal}{\emph{{ACM SIGLOG} News}} \bibinfo{volume}{4},
  \bibinfo{number}{2} (\bibinfo{year}{2017}), \bibinfo{pages}{76--89}.
\newblock


\bibitem[\protect\citeauthoryear{Wang, Sung, and Wang}{Wang
  et~al\mbox{.}}{2019a}]%
        {WSW19}
\bibfield{author}{\bibinfo{person}{Jingbo Wang}, \bibinfo{person}{Chungha
  Sung}, {and} \bibinfo{person}{Chao Wang}.} \bibinfo{year}{2019}\natexlab{a}.
\newblock \showarticletitle{Mitigating power side channels during compilation}.
  In \bibinfo{booktitle}{\emph{Proceedings of the {ACM} Joint Meeting on
  European Software Engineering Conference and Symposium on the Foundations of
  Software Engineering}}. \bibinfo{pages}{590--601}.
\newblock


\bibitem[\protect\citeauthoryear{Wang, Wang, Liu, Zhang, and Wu}{Wang
  et~al\mbox{.}}{2017}]%
        {WangWLZW17}
\bibfield{author}{\bibinfo{person}{Shuai Wang}, \bibinfo{person}{Pei Wang},
  \bibinfo{person}{Xiao Liu}, \bibinfo{person}{Danfeng Zhang}, {and}
  \bibinfo{person}{Dinghao Wu}.} \bibinfo{year}{2017}\natexlab{}.
\newblock \showarticletitle{Cache{D}: Identifying Cache-Based Timing Channels
  in Production Software}. In \bibinfo{booktitle}{\emph{Proceedings of the 26th
  USENIX Security Symposium}}. \bibinfo{pages}{235--252}.
\newblock


\bibitem[\protect\citeauthoryear{Wang, Yu, and Standaert}{Wang
  et~al\mbox{.}}{2019b}]%
        {WangYS19}
\bibfield{author}{\bibinfo{person}{Weijia Wang}, \bibinfo{person}{Yu Yu}, {and}
  \bibinfo{person}{Fran{\c{c}}ois{-}Xavier Standaert}.}
  \bibinfo{year}{2019}\natexlab{b}.
\newblock \showarticletitle{Provable Order Amplification for Code-Based
  Masking: How to Avoid Non-Linear Leakages Due to Masked Operations}.
\newblock \bibinfo{journal}{\emph{{IEEE} Transactions on Information Forensics
  and Security}} \bibinfo{volume}{14}, \bibinfo{number}{11}
  (\bibinfo{year}{2019}), \bibinfo{pages}{3069--3082}.
\newblock


\bibitem[\protect\citeauthoryear{Weissbart, Picek, and Batina}{Weissbart
  et~al\mbox{.}}{2019}]%
        {WPB19}
\bibfield{author}{\bibinfo{person}{Leo Weissbart}, \bibinfo{person}{Stjepan
  Picek}, {and} \bibinfo{person}{Lejla Batina}.}
  \bibinfo{year}{2019}\natexlab{}.
\newblock \showarticletitle{One Trace Is All It Takes: Machine Learning-Based
  Side-Channel Attack on EdDSA}. In \bibinfo{booktitle}{\emph{Proceedings of
  the 9th International Conference on Security, Privacy, and Applied
  Cryptography Engineering}}. \bibinfo{pages}{86--105}.
\newblock


\bibitem[\protect\citeauthoryear{Wu, Guo, Schaumont, and Wang}{Wu
  et~al\mbox{.}}{2018}]%
        {WGS18}
\bibfield{author}{\bibinfo{person}{Meng Wu}, \bibinfo{person}{Shengjian Guo},
  \bibinfo{person}{Patrick Schaumont}, {and} \bibinfo{person}{Chao Wang}.}
  \bibinfo{year}{2018}\natexlab{}.
\newblock \showarticletitle{Eliminating timing side-channel leaks using program
  repair}. In \bibinfo{booktitle}{\emph{Proceedings of the 27th {ACM} {SIGSOFT}
  International Symposium on Software Testing and Analysis}}.
  \bibinfo{pages}{15--26}.
\newblock


\bibitem[\protect\citeauthoryear{Wu and Wang}{Wu and Wang}{2019}]%
        {WuW19}
\bibfield{author}{\bibinfo{person}{Meng Wu} {and} \bibinfo{person}{Chao Wang}.}
  \bibinfo{year}{2019}\natexlab{}.
\newblock \showarticletitle{Abstract interpretation under speculative
  execution}. In \bibinfo{booktitle}{\emph{Proceedings of the 40th {ACM}
  {SIGPLAN} Conference on Programming Language Design and Implementation}}.
  \bibinfo{pages}{802--815}.
\newblock


\bibitem[\protect\citeauthoryear{Zhang, Gao, Song, and Wang}{Zhang
  et~al\mbox{.}}{2018}]%
        {ZGSW18}
\bibfield{author}{\bibinfo{person}{Jun Zhang}, \bibinfo{person}{Pengfei Gao},
  \bibinfo{person}{Fu Song}, {and} \bibinfo{person}{Chao Wang}.}
  \bibinfo{year}{2018}\natexlab{}.
\newblock \showarticletitle{SCInfer: Refinement-Based Verification of Software
  Countermeasures Against Side-Channel Attacks}. In
  \bibinfo{booktitle}{\emph{Proceedings of the 30th International Conference on
  Computer Aided Verification}}. \bibinfo{pages}{157--177}.
\newblock


\end{thebibliography}
\end{document}